\documentclass[10pt,reqno,a4paper,final]{amsart}
\usepackage[british]{babel}
\usepackage{amssymb,amstext,amsthm,eucal,bbm,amscd,bbold}
\usepackage[margin=1in]{geometry}
\usepackage{array}
\usepackage[svgnames,table]{xcolor}
\usepackage[T1]{fontenc}
\usepackage[utf8x]{inputenc}
\usepackage[small]{eulervm}
\usepackage{tgpagella}
\usepackage{mathrsfs}
\usepackage{verbatim}
\usepackage{pdfsync}
\usepackage[final]{microtype}
\usepackage{booktabs} %for \toprule,\midrule,\bottomrule
\usepackage{caption}
\usepackage{appendix}
\usepackage{chngcntr}
\usepackage{apptools}
\usepackage{tikz,pgfplots}
\usepgfplotslibrary{fillbetween}
\usetikzlibrary{calc,arrows,arrows.meta,cd,shapes.geometric,positioning,through,decorations.markings,decorations.text}
\tikzset{>=latex}
\usepackage[unicode,pagebackref=true]{hyperref}
\hypersetup{%
  pdftitle   = {On the intrinsic torsion of spacetime structures},
  pdfkeywords = {Newton-Cartan, Carroll, Bargmann, G-structure, intrinsic torsion},
  pdfauthor  = {José Figueroa-O'Farrill},
  pdfcreator = {\LaTeX\ with package \flqq hyperref\frqq},
  linkcolor=NavyBlue,
  citecolor=ForestGreen,
  urlcolor=OrangeRed,
  anchorcolor=OrangeRed,
  colorlinks=true
}
\renewcommand*{\backref}[1]{}
\renewcommand*{\backrefalt}[4]{%
  \ifcase #1%
  \or [Page~#2.]%
  \else [Pages~#2.]%
  \fi%
}
\theoremstyle{plain}
\newtheorem{lemma}{Lemma}
\newtheorem{proposition}[lemma]{Proposition}
\newtheorem{theorem}[lemma]{Theorem}
\newtheorem{corollary}[lemma]{Corollary}

\theoremstyle{definition}
\newtheorem{definition}[lemma]{Definition}

\newcommand{\g}{\mathfrak{g}}

\renewcommand{\d}{\partial}
\newcommand{\gl}{\mathfrak{gl}}

\newcommand{\so}{\mathfrak{so}}

\newcommand{\be}{\boldsymbol{e}}

\newcommand{\bv}{\boldsymbol{v}}

\newcommand{\bzero}{\boldsymbol{0}}

\newcommand{\F}{\mathcal{C}^\infty}
\newcommand{\eH}{\mathcal{H}}
\newcommand{\eV}{\mathcal{V}}
\newcommand{\eL}{\mathcal{L}}
\newcommand{\eX}{\mathcal{X}}

\newcommand{\eG}{\mathcal{G}}
\newcommand{\eC}{\mathcal{C}}
\newcommand{\eA}{\mathcal{A}}
\newcommand{\eB}{\mathcal{B}}

\newcommand{\Ad}{\operatorname{Ad}}

\newcommand{\id}{\operatorname{id}}
\newcommand{\tr}{\operatorname{tr}}
\renewcommand{\div}{\operatorname{div}}

\newcommand{\RR}{\mathbb{R}}

\newcommand{\V}{\mathbb{V}}
\newcommand{\W}{\mathbb{W}}
\newcommand{\ZZ}{\mathbb{Z}}

\newcommand{\GL}{\operatorname{GL}}
\newcommand{\SO}{\operatorname{SO}}

\newcommand{\Ann}{\operatorname{Ann}}
\newcommand{\End}{\operatorname{End}}
\newcommand{\Hom}{\operatorname{Hom}}
\newcommand{\coker}{\operatorname{coker}}
\newcommand{\im}{\operatorname{im}}
\newcommand{\Xbar}{\overline{X}}
\newcommand{\Ybar}{\overline{Y}}
\newcommand{\yes}{{\color{dkgr}{\boldsymbol{\checkmark}}}}
\newcommand{\pder}[2]{\frac{\partial #1}{\partial #2}}

\definecolor{dkgr}{rgb}{0,0.6,0}
\definecolor{gris}{rgb}{0.5,0.5,0.5}

\allowdisplaybreaks[0]
\numberwithin{equation}{section}
\AtAppendix{\counterwithin{lemma}{section}}

\begin{document}

\title{On the intrinsic torsion of spacetime structures}

\author{José Figueroa-O'Farrill}
\address{Maxwell Institute and School of Mathematics, The University
  of Edinburgh, James Clerk Maxwell Building, Peter Guthrie Tait Road,
  Edinburgh EH9 3FD, Scotland, United Kingdom}
\email{\href{mailto:j.m.figueroa@ed.ac.uk}{j.m.figueroa@ed.ac.uk}}
\begin{abstract}
  We briefly review the notion of the intrinsic torsion of a
  $G$-structure and then go on to classify the intrinsic torsion of
  the $G$-structures associated with spacetimes: namely, galilean (or
  Newton--Cartan), carrollian, aristotelian and bargmannian.  In the case
  of galilean structures, the intrinsic torsion classification
  agrees with the well-known classification
  into torsionless, twistless torsional and torsional Newton--Cartan
  geometries.  In the case of carrollian structures, we find that
  intrinsic torsion allows us to classify Carroll manifolds into four
  classes, depending on the action of the Carroll vector field on the
  spatial metric, or equivalently in terms of the nature of the null hypersurfaces of a
  lorentzian manifold into which a carrollian geometry may embed.
  By a small refinement of the results for galilean and carrollian
  structures, we show that there are sixteen classes of aristotelian
  structures, which we characterise geometrically.  Finally, the bulk
  of the paper is devoted to the case of bargmannian structures, where we
  find twenty-seven classes which we also characterise geometrically
  while simultaneously relating some of them to the galilean and
  carrollian structures.
\end{abstract}
\dedicatory{Dedicated to Dmitri Vladimirovich Alekseevsky on his eightieth birthday}
\thanks{EMPG-20-14, \href{https://orcid.org/0000-0002-9308-9360}{ORCID: 0000-0002-9308-9360}}
\maketitle
\tableofcontents

\section{Introduction}
\label{sec:introduction}

What are the possible geometries of space and time?  An answer to this
question was given (subject to some assumptions) by Bacry and
Lévy-Leblond \cite{MR0238545}, who pioneered the classification of
kinematical symmetries.  Later work of Bacry and Nuyts \cite{MR857383}
relaxed some of the assumptions in the earlier work and classified
kinematical Lie algebras in four space-time dimensions.  Taking these
ideas to their logical conclusion, Stefan Prohazka and I classified
(simply-connected, spatially isotropic) homogeneous kinematical
spacetimes in arbitrary dimension \cite{Figueroa-OFarrill:2018ilb}.
We found that such spacetimes are of one of several classes: lorentzian,
galilean (a.k.a. Newton--Cartan), carrollian and
aristotelian.\footnote{The classification also gives some riemannian
  spaces and in two dimensions also some spacetimes without any
  discernable structure.}  The geometry of such homogeneous
kinematical spacetimes was further studied in
\cite{Figueroa-OFarrill:2019sex}, together with Ross Grassie.

Being homogeneous, these spacetimes serve as Klein models for more
realistic spacetime geometries, in the same way that Minkowski
spacetime serves as a model for the lorentzian spacetimes of General
Relativity.  Technically, the realistic spacetimes are Cartan
geometries modelled on the kinematical Klein geometries.  A closer
analysis of the Klein geometries reveals that they fall into far fewer
classes than their number might suggest: all galilean homogeneous
spacetimes, for example, are Klein models for the same Cartan
geometry and the same is true for aristotelian and carrollian (with the
exception of the lightcone) spacetimes.  Hence it makes sense to study
galilean, carrollian, aristotelian structures without reference to a
particular homogeneous model.

One way to do this is to re-interpret the relevant structure as a
$G$-structure; that is, as a principal $G$-subbundle of the frame
bundle or, more prosaically, as a consistent way to restrict to moving
frames which are related by local $G$-transformations; such as
orthonormal frames in a riemannian manifold.  There is a notion of
affine connection adapted to a $G$-structure.  Typically these
connections have torsion and there exists a component of the torsion
which is independent of the adapted connection.  This is the intrinsic
torsion of the $G$-structure and it is the subject of the present
paper.  The intrinsic torsion is the first obstruction to the
integrability of the $G$-structure, which roughly speaking says that
there exists a coordinate atlas to the manifold whose transition
functions take values in $G$.

Studying the intrinsic torsion might seem a strange approach coming
from the direction of General Relativity because in lorentzian
geometry and in the absence of any additional structure, the intrinsic
torsion of a metric connection vanishes -- that being essentially the
Fundamental Theorem of riemannian geometry.  That this is not the end
of the story can be gleaned from the emergence of natural connections
other than the Levi-Civita connection in the context of $1/c$
expansions of General Relativity
\cite{VandenBleeken:2017rij,Hansen:2019svu} and in particular from the
emergence of (torsional) Newton--Cartan geometry in that limit.
Indeed, we will see that for the non-lorentzian $G$-structures, the
intrinsic torsion will give us some information. For example, we will
see that the classification of galilean $G$-structures \cite{MR334831}
by intrinsic torsion coincides with the classification of
Newton--Cartan geometries into torsionless (NC), twistless torsional
(TTNC) and torsional (TNC) \cite{Christensen:2013lma}. For carrollian,
aristotelian and indeed bargmannian structures, their classifications via
intrinsic torsion seem to be novel. We will see that there are four
classes of carrollian $G$-structures, sixteen classes of aristotelian
$G$-structures and twenty-seven classes of bargmannian $G$-structures,
which we will classify and characterise geometrically in terms of the
tensor fields which characterise the $G$-structure.

This paper is organised as follows.  In
Section~\ref{sec:intrinsic-torsion-g} we review the very basic notions
about $G$-structures and their intrinsic torsion.  In
Section~\ref{sec:g-structures} we review the useful language of
associated vector bundles, which we use implicitly in much of the
paper, discuss adapted connections in
Section~\ref{sec:adapted-connections} and the intrinsic torsion of a
$G$-structure in Section~\ref{sec:intrinsic-torsion}.  The rest of the
paper consists of four worked out examples of increasing complexity of
the classification of the intrinsic torsions of a $G$-structure.
For each one of the spacetime $G$-structures (galilean, carrollian,
aristotelian and bargmannian) we first work out the group $G$ and
identify the characteristic tensor fields which define and are defined
by the $G$-structure, work out the lattice of $G$-submodules where the intrinsic
torsion lives and hence classify the distinct classes of
$G$-structures, and then we characterise them geometrically in terms
of the characteristic tensor fields of the $G$-structure.

Section~\ref{sec:gal-g-struct} is devoted to galilean structures.  A
galilean structure is defined by a nowhere-vanishing ``clock''
one-form $\tau$ spanning the kernel of a positive-semidefinite
cometric $\gamma$.  Proposition~\ref{prop:gal-d-tau} shows that the
intrinsic torsion of a galilean structure is captured by $d\tau$.
Theorem~\ref{thm:gal} then shows that there are three types of
galilean structures, depending on whether or not $\tau \wedge d\tau$
is zero and, if so, whether or not $d\tau = 0$.

Section~\ref{sec:car-g-struct} is devoted to carrollian structures,
which are defined by a nowhere-vanishing vector field $\xi$ spanning
the kernel of a positive-semidefinite metric $h$.
Proposition~\ref{prop:L-xi-h} shows that the intrinsic torsion is
captured by $\eL_\xi h$ and in Theorem~\ref{thm:car} we show that
there are four types of carrollian structures: depending on whether or
not $\eL_\xi h = 0$, and if not, whether the symmetric tensor $\eL_\xi
h$  is traceless or pure trace or neither.  Recently a different
approach to the study of carrollian geometry has been proposed
\cite{Ciambelli:2019lap}, exhibiting the carrollian geometry as a
principal line bundle (with structure group the one-dimensional group
generated by $\xi$) over a riemannian manifold with metric $h$.  It
would be interesting to relate our two approaches.

Section~\ref{sec:ari-g-struct} is devoted to aristotelian structures.
An aristotelian geometry admits simultaneously a galilean structure
and a carrollian structure and hence we can re-use and refine the
results in the previous two sections to arrive at
Theorem~\ref{thm:ari}, which lists the sixteen types of aristotelian
structures.

Section~\ref{sec:bar-g-struct} is the longest and is devoted to
bargmannian structures.  A bargmannian structure consists of a lorentzian
manifold $(M,g)$ and a nowhere-vanishing null vector field $\xi$.  As
advocated in \cite{Duval:2014uoa}, bargmannian structures serve as a
bridge between galilean and carrollian structures and one can recover
the results of Sections~\ref{sec:gal-g-struct} and
\ref{sec:car-g-struct} as special cases.
Proposition~\ref{prop:intr-tors-is-nabla} shows that the intrinsic
torsion of a bargmannian structure is captured by $\nabla^g \xi$ -- the
covariant derivative of the null vector field relative to the
Levi-Civita connection of $g$.  We find that, perhaps surprisingly,
there are twenty-seven types of bargmannian structures, as described in
Theorem~\ref{thm:bar}.  These structures defined a partially ordered
set which is depicted in Figure~\ref{fig:bar-hasse-submodules}.  In
deriving the results on bargmannian structures we found the need to
extend the theory of null hypersurfaces (e.g.,
\cite{MR886772,MR1777311}) to non-involutive null distributions.  In
Section~\ref{sec:corr-betw-bargm} we relate them to galilean and
carrollian structures.  We will find that all three classes of
galilean structures can arise as null reductions of bargmannian
manifolds, whereas all four classes of carrollian structures can
arise as embedded null hypersurfaces in bargmannian manifolds.  This
then allows us to rephrase the carrollian classification in terms of
the classification of null hypersurfaces in a lorentzian manifold.
The rôle of null hypersurfaces in carrollian geometry was already
emphasised in \cite{Hartong:2015xda}.

The paper ends with some conclusions and two appendices.
Appendix~\ref{sec:hypers-orth}, included for completeness, contains a
proof of a result concerning hypersurface orthogonality which is used
often in Section~\ref{sec:bar-g-struct}.  The result is often quoted,
but hardly ever proved.  Finally, Appendix~\ref{sec:some-spec-dimens}
treats some special dimensions.  In the bulk of the paper we work with
generic $n$-dimensional galilean, carrollian and aristotelian
structures and ($n+1$)-dimensional bargmannian structures and the results
hold for $n>2$ and $n\neq 5$.  When $n=2$ there is no distinction
between carrollian and galilean structures and hence we will need to
look again at the classifications.  This is done in
Appendix~\ref{sec:two-dimens-galil}, which also treats the
two-dimensional aristotelian structures. We find that there are now
two galilean, two carrollian and four aristotelian structures in two
dimensions. When $n=2$ the classification of bargmannian structures also simplifies and this is
described in Appendix~\ref{sec:three-dimens-bargm}.  There are now
only eleven three-dimensional bargmannian structures.  When $n=5$ we
find in the galilean, aristotelian and bargmannian cases,
$\so(4)$-submodules of type $\wedge^2\RR^4$, which are not
irreducible, leading to a refinement of the classifications.  This is
described briefly in Appendices~\ref{sec:five-dimens-galil} for
galilean structures, \ref{sec:five-dimens-arist} for aristotelian
structures and \ref{sec:six-dimens-bargm} for bargmannian structures.  We
find that there are 5 galilean structures, 32 aristotelian structures
and 47 bargmannian structures in these dimensions.

\section{The intrinsic torsion of a $G$-structure}
\label{sec:intrinsic-torsion-g}

In this section we briefly review the language associated to
$G$-structures, adapted connections and their intrinsic torsion.  It
sets the stage for the calculations in the remaining sections.
Readers familiar with this language may simply skim for notation and
go directly to the calculations starting in the next section.  I do
not include any proofs, which can be found in, say,
\cite{MR1336823,MR532831}.

\subsection{$G$-structures}
\label{sec:g-structures}

Let $M$ be an $n$-dimensional smooth manifold and let $p \in M$.  By a
\textbf{frame} at $p$ we mean a vector space isomorphism $u : \RR^n
\to T_p M$.  Since $\RR^n$ has a distinguished basis (the elementary
vectors $\be_i$), its image under $u$ is a basis $(u(\be_1),
u(\be_2),\dots,u(\be_n))$  for $T_p M$.  If $u, u'$ are two frames at $p$ then
$g := u^{-1} \circ u  \in \GL(n,\RR)$.  Rewriting this as $u' = u \circ
g$ defines a right action of $\GL(n,\RR)$ on the set $F_p(M)$
of frames at $p$.  This action is transitive and free, making $F_p(M)$
into a torsor (a.k.a. principal homogeneous space) of $\GL(n,\RR)$.

The disjoint union $F(M) = \bigsqcup_{p\in M} F_p(M)$ can be made into the
total space of a principal $\GL(n,\RR)$-bundle called the
\textbf{frame bundle} of $M$.  In particular, we have a smooth free
right-action  of $\GL(n,\RR)$; that is, a diffeomorphism $R_g : F(M)
\to F(M)$ for every $g \in \GL(n,\RR)$, where $R_g u = u \circ g$ for
every frame $u \in F(M)$.  Let $\pi : F(M) \to M$ be the smooth map
sending a frame $u \in F_p(M)$ to $p \in M$.  It follows that $\pi
\circ R_g = \pi$ for all $g \in \GL(n,\RR)$, since $\GL(n,\RR)$ acts
on the frames at $p$.  A local section $s : U \to F(M)$, where $U
\subset M$, defines a \textbf{moving frame} (or \textbf{vielbein})
$(X_1,\dots,X_n)$ in $U$, where $(X_i)_p = s(p)(\be_i)$ for all $p \in
U$.

Moving frames exist on $M$ by virtue of it being a smooth manifold.
Indeed, if $(U,x^1,\dots,x^n)$ is a local coordinate chart, then
$\left(\pder{~}{x^1},\dots,\pder{~}{x^n} \right)$ is a moving frame in
$U$.  If $(V,y^1,\dots,y^n)$ is an overlapping coordinate chart, then
in the overlap $U \cap V$, the moving frames are related by a local
$\GL(n,\RR)$ transformation $g_{UV}: U \cap V \to \GL(n,\RR)$: namely,
the jacobian matrix of the change of coordinates.

It may happen, though, that we can restrict ourselves to distinguished
moving frames which are related on overlaps by local
$G$-transformations, for some subgroup $G < \GL(n,\RR)$.  For example,
we may endow $M$ with a riemannian metric and restrict ourselves to
orthonormal moving frames, which are related on overlaps by local
$O(n)$ transformations.  For every $p\in M$ let $P_p \subset F_p(M)$
denote the set of orthonormal frames at $p$.  Then $\gamma \in O(n)$
acts on $P_p$ by sending an orthonormal frame $u$ to $u' := u \circ
\gamma$, which is also an orthonormal frame.  The disjoint union $P =
\bigsqcup_{p \in M} P_p$ defines a principal $O(n)$-subbundle of $F(M)$.
We call $P \subset F(M)$ an \textbf{$O(n)$-structure} on $M$.

More generally, a \textbf{$G$-structure} on $M$ is a principal
$G$-subbundle $P \subset F(M)$.  As in the riemannian example just
considered, a $G$-structure on $M$ can be defined in terms of certain
characteristic tensor fields on $M$.  In order to explain this, we
have to briefly recall the concepts of an associated vector bundle to
a $G$-structure and of the soldering form.

Let $P \subset F(M)$ be a $G$-structure on $M$.  Then $P \to M$ is a
principal $G$-bundle.  Let $\rho : G \to \GL(\V)$ be a representation
of $G$ on some finite-dimensional vector space $\V$.  The group $G$
acts on $P \times \V$ on the right:
\begin{equation}
  \label{eq:right-G-action-PxV}
  (u,v) \cdot g := (u \circ g, \rho(g^{-1})v).
\end{equation}
Since $G$ acts freely on $P$, this action is free and the quotient $(P
\times \V)/G$ is the total space of an associated vector bundle $P
\times_G \V \to M$.  Sections of $P\times_G \V$ may be identified with
$G$-equivariant functions $P \to \V$.  More precisely, there is an
isomorphism of $C^\infty(M)$-modules
\begin{equation}
  \label{eq:avb-sections-iso}
  \Gamma(P \times_G \V) \cong C^\infty_G(P,\V) := \left\{\sigma : P
    \to \V ~\middle |~ R_g^*\sigma = \rho(g^{-1}) \circ \sigma \right\}.
\end{equation}
(We observe that $\pi : P \to M$ allows us to pull-back smooth
functions on $M$ to $P$ and view $C^\infty(M)$ as the $G$-invariant
functions
$C^\infty_G(P) = \left\{f \in C^\infty(P)~\middle |~ R_g^*f =
  f\quad\forall g \in G\right\}$.  Hence any $C^\infty(P)$-module
becomes a $C^\infty(M)$-module by restricting scalars.)  If $\W$ is
another representation, a $G$-equivariant linear map
$\phi : \V \to \W$ defines a bundle map
$\Phi: P\times_G \V \to P \times_G \W$, whose corresponding map on
sections sends $\sigma \in C^\infty_G(P,\V)$ to
$\phi \circ \sigma \in C^\infty_G(P,\W)$.

The preceding discussion holds for any principal $G$-bundle, but in
the case of a $G$-structure we have an additional structure not
present in a general principal bundle: namely, an $\RR^n$-valued
one-form $\theta$ on $P$.  To define it, suppose that $X_u \in T_u P$
is tangent to $P$ at $u \in P_p$.  Then $\theta_u (X_u) := u^{-1}
(\pi_* X_u)$, where $\pi : P \to M$ is the restriction to $P$ of the
map sending a frame $u$ at $p$ to $p\in M$.  In words, $\theta_u(X_u)$
is the coordinate vector of $\pi_* X_u \in T_pM$ relative to the frame
$u : \RR^n \to T_p M$.  The components of $\theta$ relative to the
standard basis $(\be_1,\dots,\be_n)$ for $\RR^n$ are one-forms
$\theta^i \in \Omega^1(P)$.  If $s = (X_1,\dots,X_n) : U \to P$ is a
local moving frame, then the pull-backs
$(s^*\theta^1,\dots,s^*\theta^n)$ make up the local coframe on $U$
canonically dual to $s$: that is, $s^*(\theta^i)(X_j) = \delta^i_j$.
We call $\theta \in \Omega^1(P,\RR^n)$ the \textbf{soldering form} of
the $G$-structure.

The soldering form defines an isomorphism $TM \cong P \times_G \RR^n$,
where $G$ acts on $\RR^n$ via the defining representation $G <
\GL(n,\RR)$.  In general, the soldering form allows us to identify
tensor bundles over $M$ with the corresponding associated vector
bundles $P \times_G \V$.  We will use this often and tacitly in this
paper.

Let $\rho : G \to \GL(\V)$ be a representation and let $0 \neq v \in
\V$ be $G$-invariant: namely, $\rho(g)v = v$ for all $g \in G$.  Then
the constant function $\sigma_v : P \to \V$ sending $u \mapsto v$
obeys $\sigma_v(u \circ g) = \rho(g^{-1}) \sigma_v(u)$ and therefore
gives a (nowhere-vanishing) section of the associated vector bundle $P
\times_G \V$.  If $\V$ is a tensor representation of $\RR^n$, then the
soldering form allows us to view $\sigma_v$ as a (nowhere-vanishing)
tensor field on $M$.

For example, if $\V = \odot^2 (\RR^n)^*$ is the
space of symmetric bilinear forms on $\RR^n$, then $\delta \in \V$
defined by $\delta(\be_i,\be_j)= \delta_{ij}$ is $O(n)$-invariant.
In fact, $O(n)$ is precisely the subgroup of $\GL(n,\RR)$ which leaves
$\delta$ invariant.  If $P \to M$  is an $O(n)$-structure, the
constant function $\sigma_\delta : P \to \V$ sending $u \mapsto
\delta$ defines a section of $P \times_G \V$.  The soldering form
induces an isomorphism $P \times_G \V \cong \odot^2 T^*M$ and hence
$\sigma_\delta$ defines a section $g \in \Gamma(\odot^2 T^*M)$, which
relative to a local moving frame $s = (X_1,\dots,X_n) : U \to P$
satisfies $g(X_i,X_j) = \delta_{ij}$.  Equivalently, $g = \delta_{ij}
s^*\theta^i s^*\theta^j$ (using Einstein summation convention here and
from now on).  In other words, $g$ is the riemannian metric which
defines the $O(n)$-structure.  The group $O(n)$ is not connected and
it may happen that a $O(n)$-structure further reduces to an
$\SO(n)$-structure.  In that case, there is an additional invariant
tensor: namely the volume form of the riemannian metric.

In this paper we shall be interested in several different types of
$G$-structures on an $n$-dimensional smooth manifold.  Each such group
$G$ can be defined as the subgroup of $GL(n,\RR)$ which leaves
invariant one or more tensors of the defining representation.  These
$G$-invariant tensors will then give rise to a set of
\textbf{characteristic tensor fields} on $M$ in the manner
illustrated above in the case of a riemannian structure.

\subsection{Adapted connections}
\label{sec:adapted-connections}

From now on we shall write $V$ for $\RR^n$.  In other words, $V$ is
not an abstract vector space but simply our notation for $\RR^n$.
We shall also write $\GL(V)$ for $\GL(n,\RR)$ and $\gl(V)$ for its Lie
algebra.  If $G < \GL(V)$ we shall let $\g < \gl(V)$ denote its Lie
algebra.  Let $\pi: P \to M$ be a $G$-structure and let $\theta \in
\Omega^1(P,V)$ be the soldering form.

If $u \in P$ is a frame at $p = \pi(u)$, then $(\pi_*)_u : T_u P \to
T_pM$ is a surjective linear map, whose kernel $\eV_u = \ker
(\pi_*)_u$ is called the \textbf{vertical subspace} of $T_uP$.  The
rank theorem says that $\dim \eV_u = \dim \g$.  The disjoint union
$\eV = \bigsqcup_{u \in P} \eV_u$ defines a G-invariant distribution
$\eV \subset TP$.  Indeed $\pi \circ R_g = \pi$ implies that
$(R_g)_*$ preserves the kernel of $\pi_*$.  The distribution $\eV$ is
also involutive and the leaves of the corresponding foliation of $P$
are the fibres $\pi^{-1}(p)$.

By an \textbf{Ehresmann connection} on $P$ we mean a $G$-invariant
distribution $\eH \subset TP$ complementary to $\eV$.  At every frame
$u \in P_p$, $T_u P = \eV_u \oplus \eH_u$ and $(\pi_*)_u$ restricts to
an isomorphism $\eH_u \cong T_pM$.  We will let
$h_u : T_u P \to \eH_u$ denote the horizontal projector along $\eV_u$.
Equivalently, we may define an Ehresmann connection via a
\textbf{connection one-form} $\omega \in \Omega^1(P,\g)$ defined
uniquely by the properties:
\begin{equation}
  \label{eq:conn-one-form}
  \ker \omega_u = \eH_u \qquad\text{and}\qquad \omega(\xi_X) =
  X\quad\forall X \in \g,
\end{equation}
where $\xi_X \in \eX(P)$ is the fundamental vector field corresponding
to $X \in \g$ and defined by $(\xi_X)_u = \left.\frac{d}{dt} \left(  u \circ e^{t X}
  \right)\right|_{t=0}$. It follows that
\begin{equation}
  \label{eq:G-equiv-conn}
  R_g^* \omega = \Ad(g^{-1}) \circ \omega,
\end{equation}
where $\Ad : G \to \GL(\g)$ is the adjoint representation.

An Ehresmann connection allows us to extend the $C^\infty(M)$-module
isomorphism \eqref{eq:avb-sections-iso} to differential forms.  Let us
define
\begin{equation}
  \label{eq:basic-forms}
  \Omega^p_G(P,\V) := \left\{ \varphi \in \Omega^p(P,\V) ~ \middle | ~
  R_g^* \varphi = \rho(g^{-1}) \circ \varphi\quad\text{and}\quad
  h^*\varphi = \varphi\right\},
\end{equation}
where $h^*\varphi(Y_1,\dots,Y_p) = \varphi(h Y_1,\dots,h Y_p)$, with
$h$ the horizontal projector.  The condition $R_g^* \varphi =
\rho(g^{-1}) \circ \varphi$ says that $\varphi$ is \textbf{invariant},
whereas the condition $h^*\varphi = \varphi$ says that it is
\textbf{horizontal}.  A form $\varphi \in \Omega^p_G(P,\V)$ is said to
be \textbf{basic} because it defines a $p$-form on $M$ with values in the
associated bundle $P \times_G \V$.  Indeed, we have a
$C^\infty(M)$-module isomorphism
\begin{equation}
  \label{eq:forms-avb-iso}
  \Omega^p_G(P,\V) \cong \Omega^p(M,P\times_G \V).
\end{equation}

An Ehresmann connection on $P$ defines a Koszul connection on any
associated vector bundle.  Its expression is particularly transparent
in terms of the equivariant functions $C^\infty_G(P,\V)$, where the
Koszul connection defines a covariant derivative operator:
\begin{equation}
  \label{eq:cov-der-P}
  \nabla : C^\infty_G(P,\V) \to \Omega^1_G(P,\V)
  \qquad\text{with}\qquad \nabla \sigma := h^* d\sigma.
\end{equation}
In calculations, it is more convenient to use the equivalent
expression $\nabla \sigma = d\sigma + \rho_*(\omega) \circ \sigma$,
where $\rho_* : \g \to \gl(\V)$ is the representation of $\g$ induced
by $\rho: G \to \GL(\V)$.

It is easy to see that the soldering form is actually basic: $\theta
\in \Omega^1_G(P,V)$ and hence it defines a one-form on $M$ with values
in $P \times_G V$; that is, a section of $\Hom(TM,P \times_G V)$.  This
is none other but the isomorphism $TM \cong P \times_G V$.
Functorially, it induces isomorphisms between the bundle of
$(r,s)$-tensors on $M$ and $P \times _G T^r_S(V)$, with $T^r_s(V) =
(V^{\otimes r}) \otimes (V^*)^{\otimes s}$.

The Koszul connection on $P\times_G V$ induces an affine connection
(also denoted $\nabla$) on $TM$, which is said to be \textbf{adapted to the
$G$-structure} $P$:
\begin{equation}
  \begin{tikzcd}
    \Gamma(P \times_G V) \ar[r,"\nabla"] \ar[d,"\cong"] & \Omega^1(M,P\times_G V) \ar[d,"\cong"]\\
    \eX(M) \ar[r,"\nabla"] & \Omega^1(M,TM).
  \end{tikzcd}
\end{equation}
Every characteristic tensor field on $M$ is parallel relative to this
affine connection.  This is particularly simple to see using the
description on $P$ in terms of equivariant functions.  Indeed, if $S
\in T^r_s(V)$ is a $G$-invariant tensor, then the section $\sigma_S
\in C^\infty_G(P, T^r_s(V))$ is constant, so that $d\sigma_S = 0$ and,
in particular, so is its horizontal component $\nabla\sigma_S = h^*d\sigma_S$.

\subsection{Intrinsic torsion}
\label{sec:intrinsic-torsion}

The torsion tensor $T^\nabla \in \Omega^2(M,TM)$ of an adapted affine
connection $\nabla$ is defined as usual by
\begin{equation}
  \label{eq:torsion-tensor}
  T^\nabla(X,Y) = \nabla_X Y - \nabla_Y X - [X,Y]\qquad\forall X,Y \in
  \eX(M).
\end{equation}
On $P$, the torsion tensor is represented by the \textbf{torsion
  two-form} $\Theta \in \Omega^2_G(P,V)$ defined by
$\Theta = h^*d\theta$, or equivalently by the \textbf{first structure
  equation}
\begin{equation}
  \label{eq:fss}
  \Theta = d\theta + \omega \wedge \theta,
\end{equation}
where the second term in the RHS involves also the action of $\g$ on
$V$ via the embedding $\g < \gl(V)$; that is, for all $X,Y \in
\eX(P)$, we have
\begin{equation}
  \label{eq:torsion-2-form}
  \Theta(X,Y) = d\theta(X,Y) + \omega(X) \theta(Y) - \omega(Y)\theta(X).
\end{equation}

Let us now investigate how the torsion changes when we change the
connection.  Let $\eH' \subset TP$ be a second Ehresmann connection on
$P$ with connection one-form $\omega' \in \Omega^1(P,\g)$.  Let
$\kappa = \omega' - \omega \in \Omega^1(P,\g)$.  Since $\omega$ and
$\omega'$ are invariant, so is $\kappa$; but since $\omega$ and
$\omega'$ agree on vertical vectors, $\kappa$ is now also horizontal.
Therefore $\kappa \in \Omega^1_G(P,\g)$ and hence it descends to a
one-form with values in $\Ad P := P \times_G \g$.

In general, the difference $\nabla' - \nabla$ between two affine
connections belongs to $\Omega^1(M,\End TM)$, but if the connections
are adapted to the $G$-structure, then $\nabla' - \nabla$ is a
one-form with values in the sub-bundle of $\End TM$ corresponding to
$\Ad P$ via the soldering form.

Let $\Theta'$ be the torsion two-form of $\eH'$.  From the first
structure equation \eqref{eq:fss}, we see that
\begin{equation}
  \Theta'  - \Theta = \kappa \wedge \theta
\end{equation}
or, equivalently, for all $X,Y \in \eX(P)$,
\begin{equation}
  (\Theta' - \Theta)(X,Y) = \kappa(X) \theta(Y) - \kappa(Y) \theta(X)
\end{equation}
The passage from $\kappa$ to $\Theta' - \Theta$ defines a
$C^\infty(M)$-linear map
\begin{equation}
  \Omega^1(M,P\times_G \g) \xrightarrow{} \Omega^2(M,P \times_G V)
\end{equation}
which is induced from a bundle map
\begin{equation}
  P\times_G (\g \otimes V^*) \xrightarrow{} P \times_G (V \otimes
  \wedge^2 V^*),
\end{equation}
which is in turn induced from a $G$-equivariant linear map, a special
instance of a \textbf{Spencer differential},
\begin{equation}
  \label{eq:spencer}
  \Hom(V,\g)  \xrightarrow{\d} \Hom(\wedge^2V, V)  \qquad\text{defined
    by}\qquad (\d \kappa)(v,w) = \kappa_v w - \kappa_w v, 
\end{equation}
for all $v,w \in V$ and where $\kappa: V \to \g$ sends $v \mapsto \kappa_v$.

We may summarise this discussion as follows.

\begin{proposition}
  Let $P \xrightarrow{\pi} M$ be a $G$-structure and
  $\omega \in \Omega^1(P,\g)$ the connection one-form of an Ehresmann
  connection with torsion two-form $\Theta \in \Omega^2_G(P,V)$.  If
  $\omega' = \omega + \kappa$ is another Ehresmann connection, then
  its torsion two-form $\Theta' = \Theta + \d \kappa$, where
  $\d : \Omega^1_G(P,\g) \to \Omega^2_G(P,V)$ is induced from the
  Spencer differential
  \begin{equation}
    \d : \Hom(V,\g) \to \Hom(\wedge^2 V, V)
  \end{equation}
  defined by $\d\kappa(v,w) = \kappa_v w - \kappa_w v$ for all $v,w
  \in V$.
\end{proposition}

Under the isomorphisms $\Hom(V,\g) \cong \g \otimes V^*$ and
$\Hom(\wedge^2 V, V) = V \otimes \wedge^2 V^*$, the Spencer
differential is the composition
\begin{equation}
  \label{eq:spencer-compo}
  \begin{tikzcd}
    \g \otimes V^* \ar[r,"i \otimes \id_{V^*}"] & V \otimes V^*
    \otimes V^* \ar[r, "\id_V \otimes \wedge"] & V \otimes \wedge^2 V^*
  \end{tikzcd}
\end{equation}
where $i : \g \to V \otimes V^*$ is the embedding $\g < \gl(V)$
composed with the isomorphism $\gl(V) \cong V \otimes V^*$, and
$\wedge : V^* \otimes V^* \to \wedge^2 V^*$ is skew-symmetrisation.

To the linear map $\d : \g \otimes V^* \to V \otimes \wedge^2 V^*$
there is associated an exact sequence:
\begin{equation}
  \label{eq:exact-seq}
  \begin{tikzcd}
    0 \ar[r] & \ker \d \ar[r] & \g \otimes V^* \ar[r,"\d"] & V \otimes
    \wedge^2 V^* \ar[r] & \coker \d \ar[r] & 0,
  \end{tikzcd}
\end{equation}
where $\coker\d = (V \otimes \wedge^2V^*)/\im\d$.  Since these maps
are $G$-equivariant, we obtain an exact sequence of associated vector
bundles:
\begin{equation}
  \begin{tikzcd}
    0 \ar[r] & P \times_G \ker \d \ar[r] & P \times_G (\g \otimes V^*) \ar[r,"\d"] & P\times_G(V \otimes
    \wedge^2 V^*) \ar[r] & P \times_G \coker \d \ar[r] & 0.
  \end{tikzcd}
\end{equation}
These bundles have the following interpretation:
\begin{itemize}
\item the torsion of (adapted) affine connections are sections of
  $P\times_G (V \otimes \wedge^2V^*) \cong TM \otimes \wedge^2 T^*M$;
\item the contorsions (i.e., the differences between adapted
  affine connections) are sections of $P\times_G (\g \otimes V^*)
  \cong \Ad P \otimes T^*M$;  
\item the contorsions which do not alter
  the torsion are sections of $P \times_G \ker \d$; and
\item the ``intrinsic torsion'' (see below) of
  an adapted connection is a section of $P \times_G
  \coker\d$.
\end{itemize}
Since $T^{\nabla'} - T^\nabla = \d(\nabla' - \nabla)$, we see that the
image $[T^\nabla] \in \Gamma(P \times_G \coker\d)$ of the torsion is
\emph{independent} of the connection and is an intrinsic property of
the $G$-structure.  We say $[T^\nabla]\in \Gamma(P \times_G \coker\d)$
is the \textbf{intrinsic torsion} of the $G$-structure.

As an example, consider a lorentzian $G$-structure. It is customary
here to label the standard basis of $V = \RR^n$ as
$(\be_0,\be_1,\dots,\be_{n-1})$ with canonical dual basis
$(\alpha^0,\alpha^1,\dots,\alpha^{n-1})$ for $V^*$.  Then $G < \GL(V)$ is the
subgroup leaving invariant the lorentzian inner product
\begin{equation}
  \eta = -(\alpha^0)^2 + \sum_{i=1}^{n-1} (\alpha^i)^2.
\end{equation}
The Lie algebra $\g = \so(V)$ is the space of $\eta$-skew-symmetric
endomorphisms of $V$.  As we now show, the Spencer differential is an
isomorphism in this case.

\begin{lemma}\label{lem:spencer-iso-soV}
  The Spencer differential
  \begin{equation*}
    \d : \so(V) \otimes V^* \to V \otimes \wedge^2 V^*
  \end{equation*}
  is an isomorphism.
\end{lemma}

\begin{proof}
  Notice that $\dim (\so(V) \otimes V^*) = \dim(V \otimes \wedge^2
  V^*)$, so the result will follow if we show that $\ker \d = 0$.
  Let $\kappa \in \so(V) \otimes V^*$ so that $\d\kappa(v,w) =
  \kappa_v w - \kappa_w v$.  Introduce the notation $T(v,w,z):=
  \eta(\kappa_v w,z)$.  Since $\kappa_v \in \so(V)$, $T(v,w,z) = -
  T(v,z,w)$ and if $\d\kappa = 0$ then also $T(v,w,z) = T(w,v,z)$, so
  that for all $v,w,z \in V$,
  \begin{equation*}
    T(v,w,z) = T(w,v,z) = - T(w,z,v) = - T(z,w,v) = T(z,v,w) =
    T(v,z,w) = - T(v,w,z) \implies T = 0.
  \end{equation*}
  Since $\eta$ is non-degenerate, it follows that $\kappa=0$ and hence $\ker\d = 0$.
\end{proof}

It follows therefore that $\coker \d = 0$ and hence any adapted
connection (here any metric connection) can be modified to be
torsionless, and since $\ker \d = 0$, there is a unique such
modification. In other words, we have rederived the Fundamental
Theorem of riemannian geometry: the existence of a unique torsionless
metric connection; namely, the Levi-Civita connection.

We close this short review with two observations.  Firstly, if $\g <
\so(V)$, then since $\d$ is the restriction to $\g$ of the map in
Lemma~\ref{lem:spencer-iso-soV}, it is still the case that $\ker \d =
0$ and hence the exact sequence \eqref{eq:exact-seq} becomes short
exact:
\begin{equation}
  \begin{tikzcd}
    0 \ar[r] & \g \otimes V^* \ar[r,"\d"] & V \otimes \wedge^2 V^*
    \ar[r] & \coker \d \ar[r] & 0,
  \end{tikzcd}
\end{equation}
and in particular $\dim \coker \d = n \left(  \binom{n}2 - \dim \g
\right)$.

The second observation is that $P\times_G(V \otimes \wedge^2 V^*)$ is
the bundle of which the torsion of \emph{any} connection is a section.
It is not clear that any section of that bundle can be identified with
the torsion tensor of an \emph{adapted} connection.  This is not
unrelated to the fact that the classification of $G$-structures by
their intrinsic torsion may result in classes which may not actually
be realised geometrically.  For example, it is well-known that in the
case of $G_2 < \SO(7)$ structures in a $7$-manifold, only 15 of the
possible 16 structures are realised \cite{FernandezGray,CMS}.

The rest of the paper consists in the calculation of $\coker\d$ for
four types of $G$-structures relevant to spacetime geometries:
galilean, carrollian, aristotelian and bargmannian.  Our strategy will be
the following.  For each such type of geometry we will first determine
the corresponding subgroup $G < \GL(V)$ and determine $\coker\d$ as a
$G$-module.  This usually allows us to interpret $\coker\d$, which is
a quotient module, as a certain tensor module and hence will allow us
to determine which expression in terms of the characteristic tensors
of the $G$-structure captures the intrinsic torsion.  We will then
classify the $G$-submodules of $\coker \d$ and in this way
characterise them geometrically in terms of properties of the
characteristic tensors of the $G$-structure.

\section{Galilean $G$-structures}
\label{sec:gal-g-struct}

Galilean $G$-structures were first discussed by Hans-Peter Künzle
\cite{MR334831}, who proved, among other things, that they are of
infinite type.  Some of the results in this section can already be
found in \cite{MR334831}: the determination of the group $G$ and of
the characteristic tensors and the identification of the intrinsic
torsion (which is termed the ``first structure function'') with the
exterior derivative $d\tau$ of the clock one-form.  The main deviation
from \cite{MR334831} is that we claim that there is an additional
``distinguished condition'' other than ``flatness'' which can be
imposed on the torsion of an adapted connection, which follows from
our more detailed analysis of the $G$-module structure of $\coker\d$.
Later papers on the subject of adapted connections to a galilean
structure are \cite{Bernal:2002ph,Bekaert:2014bwa,Bekaert:2015xua}.

\subsection{The group $G$ of a galilean structure}
\label{sec:gal-group}

Let $V = \RR^n$.  We will use a suggestive notation
for the standard basis for V: namely, $(H,P_1,\dots,P_{n-1})$ with
canonical dual basis $(\eta,\pi^1,\dots,\pi^{n-1})$ for $V^*$.
Indices $a,b,\dots$ will run from $1$ to $n-1$ and we will write $P_a$
and $\pi^a$.  Let $G < \GL(V)$ be the subgroup which leaves invariant
$\eta \in V^*$ and $\delta^{ab} P_a P_b \in \odot^2 V$.  It
is not hard to show that
\begin{equation}
  \label{eq:G-gal}
  G = \left\{
    \begin{pmatrix}1 & \bzero^T \\ \bv & A\end{pmatrix} ~ \middle | ~  \bv
  \in \RR^{n-1},~A \in O(n-1)\right\} < \GL(n,\RR),
\end{equation}
with Lie algebra
\begin{equation}
  \label{eq:g-gal}
  \g = \left\{
    \begin{pmatrix}0 & \bzero^T \\ \bv & A\end{pmatrix} ~ \middle | ~  \bv
  \in \RR^{n-1},~A \in \so(n-1)\right\} < \gl(n,\RR).
\end{equation}
The characteristic tensor fields of a galilean $G$-structure are a
nowhere-vanishing one-form $\tau \in \Omega^1(M)$, typically called
the \textbf{clock one-form} and a corank-one positive-semidefinite
$\gamma \in \Gamma(\odot^2 TM)$ with $\gamma(\tau,-) = 0$, typically
called the \textbf{spatial cometric}.

We will choose a basis $J_{ab}=-J_{ba},B_a$ for $\g$, with Lie
brackets
\begin{equation}
  \label{eq:gal-g-brackets}
  \begin{split}
    [J_{ab},J_{cd}] &= \delta_{bc} J_{ad} - \delta_{ac} J_{bd} - \delta_{bd} J_{ac} + \delta_{ad} J_{bc}\\
    [J_{ab}, B_c] &= \delta_{bc} B_a - \delta_{ac} B_b\\
    [B_a, B_b] &= 0.
  \end{split}
\end{equation}
The actions of $\g$ on $V$ and $V^*$ are given by
\begin{equation}
  \label{eq:gal-g-action-V}
  \begin{aligned}[m]
    J_{ab} \cdot P_c &= \delta_{bc} P_a - \delta_{ac} P_b\\
    J_{ab} \cdot H &= 0\\
    B_a \cdot P_b &= 0\\
    B_a \cdot H &= P_a
  \end{aligned}\qquad\text{and}\qquad
  \begin{aligned}[m]
    J_{ab} \cdot \pi^c &= \left(-\delta^c_b \delta_{ad} + \delta^c_a \delta_{bd}\right) \pi^d\\
    J_{ab} \cdot \eta &= 0\\
    B_a \cdot \pi^b &= - \delta^b_a \eta\\
    B_a \cdot \eta &= 0.
  \end{aligned}
\end{equation}
Letting $\left<\cdots\right>$ denote the real span, we see that
$\left<P_a\right> \subset V$ and $\left<\eta\right> \subset V^*$ are
$\g$-submodules.  Hence neither $V$ nor $V^*$ are irreducible.  The
absence of complementary submodules says that they are nevertheless
indecomposable.

\subsection{The intrinsic torsion of a galilean structure}
\label{sec:gal-intr-tors}

The Spencer differential $\d : \g \otimes V^* \to V \otimes \wedge^2
V^*$ is given by
\begin{equation}
  \label{eq:gal-spencer}
  \begin{split}
   \d (J_{ab} \otimes \pi^c) &= (\delta_{bd} P_a - \delta_{ad} P_b) \otimes \pi^d \wedge \pi^c\\
   \d(J_{ab} \otimes \eta) &= (\delta_{bc}P_a - \delta_{ac}P_b) \otimes \pi^c \wedge \eta\\
   \d(B_a \otimes \pi^b) &= P_a \otimes \eta \wedge \pi^b\\
   \d(B_a \otimes \eta) &= 0.
  \end{split}
\end{equation}
Therefore we see that its kernel is given by
\begin{equation}\label{eq:gal-spencer-ker}
  \ker \d = \left<B_a \otimes \eta, J_{ab}\otimes \eta + (\delta_{bc}
    B_a - \delta_{ac} B_b) \otimes \pi^c\right>.
\end{equation}

\begin{lemma}
  As $\g$-modules, $\ker \d \cong \wedge^2 V^*$.
\end{lemma}

\begin{proof}
  The $\g$ action on $\wedge^2 V^*$ is given by the obvious action of
  $\so(n-1)$ and then
  \begin{equation*}
    \begin{split}
      B_c \cdot \pi^a \wedge \pi^b &= - \delta^a_c \eta \wedge \pi^b + \delta^b_c \eta \wedge \pi^a\\
      B_c \cdot \pi^a \wedge \eta &= 0.
    \end{split}
  \end{equation*}
  The action of $\g$ on $\ker \d$ is again given by the obvious action
  of $\so(n-1)$ and then
  \begin{equation*}
    \begin{split}
      B_c \cdot (J_{ab}\otimes \eta + (\delta_{bd}  B_a - \delta_{ad}  B_b) \otimes \pi^d) &= 2 (\delta_{ca} B_b - \delta_{cb} B_a) \otimes \eta\\
      B_c \cdot (B_a \otimes \eta) &= 0.
    \end{split}
  \end{equation*}
  This suggests defining a linear map $\varphi : \ker\d \to \wedge^2
  V^*$ by
  \begin{equation*}
    \begin{split}
      \varphi(B_a \otimes \eta) &= \delta_{ab} \pi^b \wedge \eta\\
      \varphi(J_{ab}\otimes \eta + (\delta_{bc} B_a - \delta_{ac} B_b)
      \otimes \pi^c) &= 2 (\delta_{ac} \delta_{bd}) \pi^c \wedge \pi^d.
    \end{split}
  \end{equation*}
  This map is clearly an $\so(n-1)$-equivariant isomorphism and one
  can easily check that it is also equivariant under the action of
  $B_a$.
\end{proof}

The cokernel of the Spencer differential is spanned by the image in
$\coker \d$ of $\left< H \otimes \pi^a \wedge \pi^b, H \otimes \eta
  \wedge \pi^a\right>$.

\begin{lemma}\label{lem:gal-coker-d-iso}
  As $\g$-modules, $\coker\d \cong \wedge^2 V^*$.
\end{lemma}

\begin{proof}
  We consider the $\g$-equivariant linear map $\eta \otimes
  \id_{\wedge^2 V^*} : V \otimes \wedge^2V^* \to \wedge^2 V^*$, which
  simply applies $\eta \in V^*$ to the $V$-component.  From
  equation~\eqref{eq:gal-spencer}, we see that the image of the
  Spencer differential is contained in its kernel and hence it induces
  a $\g$-equivariant linear map $\coker \d \to \wedge^2 V^*$.
  Explicitly, it is given on the basis for $\coker \d$ by
  \begin{equation*}
    \begin{split}
      [H \otimes \pi^a \wedge \pi^b] &\mapsto \pi^a \wedge \pi^b\\
      [H \otimes \eta \wedge \pi^a] &\mapsto \eta \wedge \pi^a,
    \end{split}
  \end{equation*}
  which is clearly seen to be an isomorphism.
\end{proof}

As a $\g$-module, $\wedge^2 V^*$ is indecomposable but not
irreducible.  Indeed, we have the following chain of
submodules:\footnote{This is for $n\neq 5$.  The case $n=5$ is treated
in Appendix~\ref{sec:five-dimens-galil}.}
\begin{equation}
  0 \subset \left<\eta \wedge \pi^a\right> \subset \wedge^2 V^*.
\end{equation}
By Lemma~\ref{lem:gal-coker-d-iso}, this is also the case for
$\coker\d$ and hence we see that there are three classes of galilean
structures depending on whether the intrinsic torsion vanishes, lands
in the submodule $\left< [H \otimes \pi^a \wedge \eta]\right> \cong
\left<\eta \wedge \pi^a\right>$ or is generic.

Note that the short exact sequence of $\g$-modules
\begin{equation}
  \label{eq:im-d-ses}
  \begin{tikzcd}
    0 \ar[r] & \im\d \ar[r] & V \otimes \wedge^2 V^* \ar[r] & \coker\d \ar[r] & 0
  \end{tikzcd}
\end{equation}
does \emph{not} split; although it does split as vector spaces.  This
means that whereas it is possible to find a vector subspace of $V
\otimes \wedge^2 V^*$ complementary to $\im\d$, it is not possible to
demand in addition that it should be stable under $\g$.  In this case
we have chosen $\left<H \otimes \pi^a \wedge \pi^b, H \otimes \eta
  \wedge \pi^a\right>$ as the vector space complement of $\im\d$ in $V
\otimes \wedge^2 V^*$.  This subspace is not preserved under $\g$, but
only modulo $\im \d$.  This has the following geometrical
consequence.  Having intrinsic torsion in the submodule $\eG :=
\left<[H \otimes \pi^a \wedge \eta]\right> \subset \coker \d$ does
\emph{not} mean that there exists an adapted connection $\nabla$
whose torsion $T^\nabla$ is a section of $P \times_G \overline{\eG}$,
where $\overline{\eG} = \left<H \otimes
  \pi^a \wedge \eta \right> \subset V \otimes \wedge^2 V^*$.  What it
\emph{does} mean is that relative to some local moving frame (in
$P$), the torsion will be represented by a function $U \to
\overline{\eG}$, but if we change the frame (while still in $P$), this
might not persist.  However one can modify the connection such that
relative to the new adapted connection, the torsion is again
represented by a function $U \to \overline{\eG}$.  This is why it is
important to derive consequences of the fact that the intrinsic
torsion lands in $\eG$ which are independent of the choice of the
adapted connection.

This is something one seldom sees in riemannian $G$-structures where
$G< O(n)$, since $G$, being compact, is reductive: sequences of
$G$-modules split and modules are fully reducible into irreducibles.
This is why results of the kind reported in this paper are typically
simpler to state in that situation.

\subsection{Geometric characterisation}
\label{sec:gal-geom-char}

It follows from the isomorphism in Lemma~\ref{lem:gal-coker-d-iso},
that there is bundle isomorphism $P \times_G \coker\d \cong \wedge^2
T^*M$ and therefore the intrinsic torsion of an adapted connection is
captured by a two-form.  To identify this two-form, we notice that the
$\g$-equivariant linear map $\eta \times \id_{\wedge^2 V^*}: V \otimes
\wedge^2 V^* \to \wedge^2 V^*$ in the proof of
Lemma~\ref{lem:gal-coker-d-iso}, induces a bundle map $TM
\otimes \wedge^2 T^*M \to \wedge^2 T^*M$ and hence a
$C^\infty(M)$-linear map $\Phi: \Omega^2(M,TM) \to \Omega^2(M)$ which
is given by composing with the clock form $\tau$.  In other words,
$\Phi(T) = \tau \circ T$ for any $T \in \Omega^2(M,TM)$.

\begin{proposition}\label{prop:gal-d-tau}
  Let $\nabla$ be an adapted affine connection with torsion
  $T^\nabla \in \Omega^2(M,TM)$.  Its image under $\Phi:
  \Omega^2(M,TM) \to \Omega^2(M)$ is given by $\Phi(T^\nabla) =
  d\tau$, where $\tau \in \Omega^1(M)$ is the clock one-form.
\end{proposition}

\begin{proof}
  Since the clock one-form $\tau$ is parallel relative to any adapted
  affine connection, we have that for all $X,Y \in \eX(M)$,
  \begin{equation*}
    X \tau(Y) = \tau(\nabla_X Y).
  \end{equation*}
  Skew-symmetrising,
  \begin{align*}
    X \tau(Y) - Y \tau(X) &= \tau(\nabla_X Y - \nabla_Y X)\\
                          &= \tau([X,Y] + T^\nabla(X,Y)) &&\tag{by definition of $T^\nabla$},
  \end{align*}
  so that
  \begin{equation*}
    d\tau(X,Y) = X \tau (Y) - Y \tau(X)- \tau([X,Y]) = \tau (T^\nabla(X,Y)).
  \end{equation*}
  In other words, $d\tau = \tau \circ T^\nabla = \Phi(T^\nabla)$, as
  desired.
\end{proof}

If the intrinsic torsion vanishes, then $d\tau = 0$.  If the intrinsic
torsion lands in the subbundle $P \times_G \eG$, then $d\tau$ is
represented locally by a function $U \to \left<\eta \wedge
  \pi^a\right>$, which says that $d\tau = \tau \wedge \alpha$ for some
$\alpha \in \Omega^1(U)$.  This implies that $d\tau \wedge \tau = 0$
which, as shown in Appendix~\ref{sec:hypers-orth} implies in turn that
$d\tau = \tau \wedge \alpha$ for a global one-form $\alpha \in
\Omega^1(M)$.  Finally, the generic case is where $d\tau \wedge
\tau \neq 0$.

We summarise this discussion as follows, which is to be compared with
\cite[Table~I]{Christensen:2013lma}.

\begin{theorem}\label{thm:gal}
  Let\footnote{See Appendices~\ref{sec:two-dimens-galil} for $n=2$ and
    \ref{sec:five-dimens-galil} for $n=5$.} $n>2$ and $n\neq 5$.  A
  galilean $G$-structure on an $n$-dimensional manifold $M$ may be of
  one of three classes, according to its intrinsic torsion.  If
  $\tau \in \Omega^1(M)$ is the clock one-form, then three cases can
  exist:
  \begin{itemize}
  \item[($\eG_0$)] $d\tau = 0$, corresponding to a torsionless Newton--Cartan geometry (NC);
  \item[($\eG_1$)] $d\tau \neq 0$ and $d\tau \wedge \tau = 0$, corresponding to a twistless torsional Newton--Cartan geometry (TTNC); and
  \item[($\eG_2$)] $d\tau \wedge \tau \neq 0$, corresponding to a torsional Newton--Cartan geometry (TNC).
  \end{itemize}
\end{theorem}

All spatially isotropic homogeneous galilean spacetimes in
\cite{Figueroa-OFarrill:2018ilb,Figueroa-OFarrill:2019sex} have $d\tau
=0$, but there exist homogeneous examples of all three kinds
\cite{Grosvenor:2017dfs}.

\section{Carrollian $G$-structures}
\label{sec:car-g-struct}

\subsection{The group $G$ of a carrollian structure}
\label{sec:car-group}

We use the same notation as in the previous section: with $V =
\left<H,P_a\right>$ and $V^*=\left<\eta,\pi^a\right>$.  Let $G <
\GL(V)$ be the subgroup leaving invariant $H \in V$ and
$\delta_{ab}\pi^a \pi^b \in \odot^2 V^*$.  Explicitly,
\begin{equation}
  \label{eq:G-car}
  G = \left\{
    \begin{pmatrix}1 & \bv^T \\ \bzero & A\end{pmatrix} ~ \middle | ~  \bv
  \in \RR^{n-1},~ A \in O(n-1)\right\} < \GL(n,\RR),
\end{equation}
with Lie algebra
\begin{equation}
  \label{eq:g-car}
  \g = \left\{
    \begin{pmatrix}0 & \bv^T \\ \bzero & A\end{pmatrix} ~ \middle | ~  \bv
    \in \RR^{n-1},~ A \in \so(n-1)\right\} < \gl(n,\RR).
\end{equation}
We remark that the groups for the galilean and carrollian structures
are abstractly isomorphic, being isomorphic to the semi-direct product
$O(n-1) \ltimes \RR^{n-1}$, but crucially they are not conjugate
inside $\GL(V)$.  Indeed, if they were conjugate, they would have
the invariants in the same representations.  To see this let
$\rho: \GL(V) \to \GL(\V)$ be a representation and suppose that $G,G'<\GL(V)$
are conjugate subgroups.  This means that there exists $\gamma \in \GL(V)$
such that $G' = \gamma G \gamma^{-1}$.  Suppose now that $v \in \V$ is
$G$-invariant, so that $\rho(g)v =v$ for all $g \in G$.  Then
$v'=\rho(\gamma)v \in \V$ is $G'$-invariant.  To show this let $g' \in
G'$ be arbitrary.  Then $g' = \gamma g \gamma^{-1}$ for some $g \in G$
and calculate
\begin{equation*}
  \rho(g')v' = \rho(g')\rho(\gamma)v = \rho(g'\gamma) v = \rho(\gamma
  g) v = \rho(\gamma) \rho(g) v = \rho(\gamma) v = v'.
\end{equation*}
Since the galilean structure group has an invariant in the
representation $V^*$ and the carrollian structure group does not, they
cannot be conjugate subgroups of $\GL(V)$.

The characteristic tensor fields of a carrollian $G$-structure are a
nowhere-vanishing vector field $\xi \in \eX(M)$, typically called
the \textbf{carrollian vector field} and a corank-one positive-semidefinite
$h \in \Gamma(\odot^2 T^*M)$ with $h(\xi,-) = 0$, typically
called the \textbf{spatial metric}.

The group $G$ has two connected components, corresponding to the value
of the determinant of the matrix $A \in O(n-1)$. If the $G$-structure
reduces further to a $G_0$-structure, where $G_0$ is the identity
component of $G$, there is an additional characteristic tensor: namely
a \textbf{volume form} $\mu \in \Omega^n(M)$, corresponding to the
$G_0$-invariant tensor
$\eta \wedge \pi^1 \wedge \cdots \wedge \pi^{n-1} \in \wedge^n V^*$.
Even if the $G$-structure does not reduce to $G_0$, the volume form
exists locally, but it may change by a sign on overlaps.

As before, let $\g = \left<J_{ab},B_a\right>$ with the same Lie
brackets as in \eqref{eq:gal-g-brackets}.  The action of $\g$ on $V$
and $V^*$ is given by
\begin{equation}
  \label{eq:car-g-action-V}
  \begin{aligned}[m]
    J_{ab} \cdot P_c &= \delta_{bc} P_a - \delta_{ac} P_b\\
    J_{ab} \cdot H &= 0\\
    B_a \cdot P_b &= \delta_{ab} H\\
    B_a \cdot H &= 0
  \end{aligned}\qquad\text{and}\qquad
  \begin{aligned}[m]
    J_{ab} \cdot \pi^c &= \left(-\delta^c_b \delta_{ad} + \delta^c_a \delta_{bd}\right) \pi^d\\
    J_{ab} \cdot \eta &= 0\\
    B_a \cdot \pi^b &= 0\\
    B_a \cdot \eta &= -\delta_{ab}\pi^b.
  \end{aligned}
\end{equation}
As in the galilean case, the $\g$-modules $V$ and $V^*$ are
indecomposable but not irreducible, since $\left<H\right> \subset V$
and its annihilator $\Ann H := \left<\pi^a\right> \subset V^*$ are
submodules without complementary submodules.

\subsection{The intrinsic torsion of a carrollian structure}
\label{sec:car-intr-tors}

The Spencer differential $\d : \g \otimes V^* \to V \otimes \wedge^2
V^*$ is given relative to our choice of basis by
\begin{equation}
  \label{eq:car-spencer-d}
  \begin{split}
    \d(J_{ab} \otimes \eta) &= (\delta_{bc} P_a - \delta_{ac} P_b) \otimes \pi^c \wedge \eta \\
    \d(J_{ab} \otimes \pi^c) &= (\delta_{bd} P_a - \delta_{ad} P_b) \otimes \pi^d \wedge \pi^c\\
    \d(B_a \otimes \eta) &= \delta_{ab} H \otimes \pi^b \wedge \eta\\
    \d(B_a \otimes \pi^b) &= \delta_{ac} H \otimes \pi^c \wedge \pi^b.
  \end{split}
\end{equation}

\begin{lemma}\label{lem:car-coker-d}
  As $\g$-modules, $\ker\d \cong \coker\d \cong \odot^2 \Ann H$, where
  $\Ann H \subset V^*$ is the annihilator of $H$.
\end{lemma}

\begin{proof}
  The kernel of the Spencer differential is easily seen to be
  \begin{equation*}
    \ker \d = \left<(\delta_{bc} B_a + \delta_{ac} B_b) \otimes \pi^c\right>
  \end{equation*}
  which suggests defining $\ker \d \to \odot^2 \Ann H$ by
  \begin{equation*}
    (\delta_{bc} B_a + \delta_{ac} B_b) \otimes \pi^c \mapsto \delta_{ac}\delta_{bd} \pi^c \pi^d.
  \end{equation*}
  This is clearly $\g$-equivariant, since it is manifestly
  $\so(n-1)$-equivariant and $B_a$ acts trivially on both sides.  It
  is also clearly an isomorphism.  Similarly the cokernel of the
  Spencer differential is the image in
  $\coker \d$ of
  \begin{equation*}
    \left<(\delta_{bc} P_a + \delta_{ac} P_b) \otimes \eta \wedge \pi^c\right> \subset V \otimes \wedge^2 V^*,
  \end{equation*}
  and we define $\coker \d \to \odot^2 \Ann H$ by
  \begin{equation*}
    (\delta_{bc} P_a + \delta_{ac} P_b) \otimes \eta \wedge \pi^c \mapsto \delta_{ac}\delta_{bd} \pi^c \pi^d,
  \end{equation*}
  which can be easily checked to be a $\g$-equivariant isomorphism.
\end{proof}

Since $B_a$ acts trivially on $\odot^2 \Ann H$, we may think of it
simply as an $\so(n-1)$-module.  It is therefore fully reducible into
a direct sum of two irreducible submodules:
\begin{equation}
  \odot^2 \Ann H = \odot_0^2 \Ann H \oplus \RR \delta^\perp,
\end{equation}
where $\odot_0^2 \Ann H$ are the traceless symmetric bilinear forms
and $\delta^\perp := \delta_{ab} \pi^a \pi^b$.  This decomposes
$\coker \d$ into a direct sum of two irreducible submodules
\begin{equation}
  \coker\d = \eC_1 \oplus \eC_2,
\end{equation}
where the submodule $\eC_1$ is of type $\odot_0^2 \Ann H$ and is the
image in $\coker\d$ of the subspace
\begin{equation}
  \left<(\delta_{bc} P_a + \delta_{ac} P_b - \tfrac2{n-1} \delta_{ab} P_c) \otimes \eta \wedge \pi^c\right> \subset V \otimes
  \wedge^2 V^*,
\end{equation}
whereas the trivial submodule $\eC_2$ is the image in $\coker\d$ of
the one-dimensional subspace
\begin{equation}
  \left<P_a \otimes \eta \wedge \pi^a\right> \subset V \otimes
  \wedge^2 V^*.
\end{equation}
Thus we see that there are four submodules of $\coker\d$: 0, $\eC_1$, $\eC_2$ and
$\coker\d = \eC_1 \oplus \eC_2$ and hence we conclude that there are
four classes of carrollian $G$-structures according to which submodule
of $\coker\d$ the intrinsic torsion lands in.

\subsection{Geometric characterisation of carrollian structures}
\label{sec:car-geom-char}

The isomorphism $\coker\d \cong \odot^2\Ann H$ of $\g$-modules in
Lemma~\ref{lem:car-coker-d} is induced (up to an inconsequential
factor of $2$) by the $\g$-equivariant linear map
\begin{equation}
  \phi: \Hom(\wedge^2V, V) \to \odot^2 \Ann H
\end{equation}
defined for $T \in \Hom(\wedge^2V,V)$ by
\begin{equation}
  \phi(T)(v,w) := \delta^\perp(T(H,v),w) +
    \delta^\perp(T(H,w),v)~\quad\forall~v,w \in V.
\end{equation}
We check that $\phi(T)$ does land in $\odot^2 \Ann H$:
\begin{equation}
  \phi(T)(H,v) = \delta^\perp(T(H,H),v) + \delta^\perp(T(H,v),H) = 0,
\end{equation}
where the first term vanishes because of skew-symmetry of $T$ and the
second because $\delta^\perp(H,-) = 0$.  Explicitly,
\begin{equation}
  \begin{split}
    \phi(P_a \otimes \pi^b \wedge \pi^c) &= 0\\
    \phi(P_a \otimes \eta \wedge \pi^b) &= \delta_{ac} \pi^b \pi^c\\
    \phi(H \otimes \pi^a \wedge \pi^b) &= 0\\
    \phi(H \otimes \eta \wedge \pi^a) &= 0,
  \end{split}
\end{equation}
and we check that $\im \d \subset \ker\phi$, so that $\phi$ does
induce a map $\coker\d \to \odot^2\Ann H$ which coincides with the one
in Lemma~\ref{lem:car-coker-d}, up to an overall factor of $2$.

The map $\phi$ induces a bundle map of the relevant associated vector
bundles and hence a $C^\infty(M)$-linear map
\begin{equation}\label{eq:car-bundle-map}
  \Phi: \Omega^2(M,TM) \to \Gamma(\odot^2 \Ann \xi)
\end{equation}
where, for $T \in \Omega^2(M,TM)$,
\begin{equation}
  \Phi(T)(X,Y) = h(T(\xi,X),Y) + h(T(\xi,Y),X) \quad\forall~X,Y\in\eX(M).
\end{equation}

\begin{proposition}\label{prop:L-xi-h}
  Let $\nabla$ be an affine connection adapted to a carrollian
  $G$-structure on $M$ with torsion $T^\nabla$.  Then under the map $\Phi$ in
  equation~\eqref{eq:car-bundle-map},
  \begin{equation}
    \Phi(T^\nabla) = \eL_\xi h.
  \end{equation}
\end{proposition}

\begin{proof}
  Since $\nabla$ is adapted, both the carrollian vector field $\xi \in
  \eX(M)$ and the spatial metric $h \in \Gamma(\odot^2\Ann\xi)$ are
  parallel.  From $\nabla \xi = 0$ we have that
  \begin{equation}\label{eq:car-xi-par}
    T^\nabla(\xi,X) = \nabla_\xi X - [\xi,X],\qquad \forall~X \in \eX(M),
  \end{equation}
  and from $\nabla_\xi h = 0$ we have that for all $X,Y \in \eX(M)$,
  \begin{equation*}
    \xi h(X,Y) - h(\nabla_\xi X, Y) - h(X,\nabla_\xi Y) = 0.
  \end{equation*}
  We may expand the first term using the Lie derivative and arrive at
  \begin{equation*}
    (\eL_\xi h)(X,Y) + h([\xi,X],Y) + h(X,[\xi,Y]) - h(\nabla_\xi X, Y) - h(X,\nabla_\xi Y) = 0,
  \end{equation*}
  which, using equation~\eqref{eq:car-xi-par}, becomes
  \begin{equation*}
    (\eL_\xi h)(X,Y) - h(T^\nabla(\xi,X),Y) + h(X,T^\nabla(\xi,Y)) = 0
  \end{equation*}
  or, equivalently,
  \begin{equation*}
    (\eL_\xi h)(X,Y) = \Phi(T^\nabla)(X,Y).
  \end{equation*}
\end{proof}

\begin{proposition}\label{prop:L-xi-mu}
  Let $\mu$ denote the (perhaps only locally defined) volume form on
  $M$.  Then
  \begin{equation}
    \eL_\xi \mu = \tr(S) \mu,
  \end{equation}
  where $S(X):= T^\nabla(\xi,X)$ for all $X \in \eX(M)$.
\end{proposition}

\begin{proof}
  Let $s = (X_0=\xi,\underbrace{X_1,\dots,X_{n-1}}_{X_a}) : U \to P$
  be a local moving frame with $h(X_a,X_b)=\delta_{ab}$ and, of course
  $h(X_0,-)=0$.  Let
  $(\theta^0,\underbrace{\theta^1,\dots,\theta^{n-1}}_{\theta^a})$ be
  the canonically dual coframe, so that $h =
  \delta_{ab}\theta^a\theta^b$.  Then the local expression for the
  volume form is $\mu = \theta^0 \wedge \theta^1 \wedge \cdots \wedge
  \theta^{n-1}$ and hence by the Cartan formula
  \begin{equation}
    \eL_\xi \mu = d\imath_\xi \mu = d(\theta^1 \wedge \cdots \wedge
    \theta^{n-1}).
  \end{equation}
  Pulling the first structure equation \eqref{eq:fss} back to $M$ via
  $s$, we have
  \begin{equation*}
    d\theta^a = T^a - \omega^a{}_0 \wedge \theta^0 - \omega^a{}_b \wedge \theta^b,
  \end{equation*}
  where $T^a = \theta^a \circ T^\nabla$.  Since $\nabla \xi =0$, we
  have that, for all $Y \in \eX(M)$,
  \begin{equation*}
    0 = \nabla_Y X_0 = X_0 \omega(Y)^0{}_0 + X_a \omega(Y)^a{}_0
    \implies \omega^0{}_0 = \omega^a{}_0 = 0,
  \end{equation*}
  so that
  \begin{equation}
    d\theta^a = T^a - \omega^a{}_b \wedge \theta^b,
  \end{equation}
  and hence
  \begin{equation*}
    d(\theta^1 \wedge \cdots \wedge \theta^{n-1}) = (T^1 -
    \omega^1{}_a \wedge \theta^a) \wedge \theta^2 \wedge \cdots
    \theta^{n-1} - \theta^1 \wedge (T^2 - \omega^2{}_a \wedge
    \theta^a) \wedge \theta^3 \wedge \cdots \wedge \theta^{n-1} +
    \cdots
  \end{equation*}
  The only terms which contribute to this sum are $T^a(X_0,X_a)
  \theta^0 \wedge \theta^a$ and $\omega(X_0)^a_a\theta^0$ with no
  summation implied in either term.  In summary,
  \begin{equation*}
    d(\theta^1 \wedge \cdots \wedge \theta^{n-1}) = \left( \theta^a
      \circ T^\nabla(\xi,X_a) - \omega(\xi)^a{}_a\right) \mu.
  \end{equation*}
  We claim that $\omega(\xi)^a{}_a = 0$ since $\nabla h = 0$.  Indeed,
  \begin{align*}
    0 &= \left( \nabla_\xi h \right)(X_a,X_b) \\
      &= \xi h(X_a,X_b) - h(\nabla_\xi X_a, X_b) - h(X_a, \nabla_\xi X_b)\\
      &= - h(X_c \omega(\xi)^c{}_a, X_b) - - h(X_a, X_c
        \omega(\xi)^c{}_b) && \tag{using that $h(\xi,-) = 0$}\\
      &= - \omega(\xi)_{ba} - \omega(\xi)_{ab},
  \end{align*}
  which implies that $\omega(\xi)^a{}_a = \delta^{ab} \omega(\xi)_{ab}
  = 0$.  In summary,
  \begin{equation*}
    d(\theta^1 \wedge \cdots \wedge \theta^{n-1}) = \left( \theta^a
      \circ T^\nabla(\xi,X_a) \right) \mu = \theta^a S(X_a) \mu =
    \tr(S) \mu.
  \end{equation*}
\end{proof}

We can now recognise the geometrical significance of the different
intrinsic torsion conditions.  If $\Phi(T^\nabla) = 0$, then $\eL_\xi
h = 0$ and hence $\xi$ is $h$-Killing.  If $\Phi(T^\nabla) = f h$, for
some $f \in C^\infty(M)$, then $\eL_\xi h = f h$ and hence $\xi$ is
$h$-conformal Killing.  Finally, if $\Phi(T^\nabla)$ is traceless,
then $\eL_\xi \mu = 0$, so that $\xi$ is volume-preserving.
Otherwise, we have a generic carrollian structure.

We may summarise this discussion as follows.

\begin{theorem}\label{thm:car}
  Let\footnote{See Appendix~\ref{sec:two-dimens-galil} for $n=2$.}
  $n>2$.  A carrollian $G$-structure on an $n$-dimensional manifold
  $M$ can be of one of four classes depending on the Lie derivative
  $\eL_\xi h$ of the spatial metric $h$ along the carrollian vector
  field $\xi$:
  \begin{itemize}
  \item[($\eC_0$)] $\eL_\xi h = 0$ ($\xi$ is $h$-Killing)
  \item[($\eC_1$)] $\eL_\xi \mu = 0$ ($\xi$ is volume-preserving);
  \item[($\eC_2$)] $\eL_\xi h = f h$ ($\exists 0 \neq f \in C^\infty(M)$)  ($\xi$  is $h$-conformal Killing); 
  \item[($\eC_3$)] none of the above.
  \end{itemize}
\end{theorem} 

The symmetric carrollian spaces in
\cite{Figueroa-OFarrill:2018ilb,Figueroa-OFarrill:2019sex} all have
$\eL_\xi h = 0$, but the formulae in
\cite[Section~7.3]{Figueroa-OFarrill:2019sex} show that the carrollian
lightcone has $\eL_\xi h = 2 h$, so that $\xi$ is $h$-homothetic.  I
am not aware of any explicit homogeneous carrollian manifolds in the
other two classes; although it should not be hard to construct them as
null hypersurfaces of lorentzian manifolds using as a hint the
relationship with bargmannian structures in
Section~\ref{sec:corr-betw-bargm}, where we will reformulate the
conditions in Theorem~\ref{thm:car} in terms of the different types of
null hypersurfaces in a lorentzian manifold.

\section{Aristotelian $G$-structures}
\label{sec:ari-g-struct}

\subsection{The group $G$ of an aristotelian structure}
\label{sec:ari-group}

An aristotelian space admits simultaneously a galilean and a
carrollian structure, so the group $G < \GL(V)$ corresponding to an
aristotelian $G$-structure is the intersection of the groups in
equations~\eqref{eq:G-gal} and \eqref{eq:G-car}, namely
\begin{equation}\label{eq:G-ari}
  G = \left\{
    \begin{pmatrix}1 & \bzero^T \\ \bzero & A\end{pmatrix} ~ \middle | ~ A \in O(n-1)\right\} < \GL(n,\RR),
\end{equation}
with Lie algebra
\begin{equation}
  \label{eq:g-ari}
  \g = \left\{
    \begin{pmatrix}0 & \bzero^T \\ \bzero & A\end{pmatrix} ~ \middle |
    ~  A \in \so(n-1)\right\} < \gl(n,\RR).
\end{equation}
In other words $G \cong O(n-1)$ and $\g \cong \so(n-1)$ is spanned by
$J_{ab}$, consistent with the fact that there are no boosts in an
aristotelian spacetime.

Under the action of $G$, both $H \in V$ and $\eta \in V^*$ are
invariant, as are $\delta^{ab} P_a P_b \in \odot^2 V$ and
$\delta_{ab}\pi^a\pi^b \in \odot^2 V^*$. Therefore, as $G$-modules, we
have decompositions into irreducible submodules:
\begin{equation}
  V = \left<H\right> \oplus \Ann \eta \qquad\text{and}\qquad V^* =
  \left<\eta\right> \oplus \Ann H.
\end{equation}
Moreover, $V$ and $V^*$ are isomorphic $G$-modules.  For example, the
map $\phi: V \to V^*$ defined by
\begin{equation}
  \phi(H) = \eta \qquad\text{and}\qquad \phi(P_a) =\delta_{ab} \pi^b
\end{equation}
is a $G$-equivariant isomorphism.

This means that an aristotelian spacetime has the following
characteristic tensor fields: a nowhere vanishing vector field $\xi$
and a nowhere-vanishing one-form $\tau$ which can be normalised to
$\tau(\xi)=1$, and corank-one positive-semidefinite
$\gamma \in \Gamma(\odot^2 TM)$ and $h \in \Gamma(\odot^2T^*M)$ with
$\gamma(\tau,-)=0$ and $h(\xi,-)=0$.

\subsection{The intrinsic torsion of an aristotelian structure}
\label{sec:ari-intr-tors}

Since $\g < \so(V)$, for either a lorentzian or euclidean inner
product on $V$, Lemma~\ref{lem:spencer-iso-soV} says that the Spencer
differential $\d : \g \otimes V^* \to V \otimes \wedge^2 V^*$ is
injective.  It is given explicitly by
\begin{equation}
  \label{eq:ari-spencer}
  \begin{split}
    \d(J_{ab} \otimes \pi^c) &= (\delta_{ad} P_b - \delta_{bd} P_a)  \otimes \pi^c \wedge \pi^d\\
    \d(J_{ab} \otimes \eta) &= (\delta_{bc} P_a - \delta_{ac}P_b) \otimes \pi^c \wedge \eta.
  \end{split}
\end{equation}
It then follows that the image of the Spencer differential is given by
\begin{equation}
  \im\d = \left<P_a \otimes \pi^b \wedge \pi^c, \left( \delta_{bc} P_a
      - \delta_{ac} P_b \right) \otimes \pi^c \wedge \eta\right>
\end{equation}
and hence the cokernel is the image in $\coker\d$ of
\begin{equation}
  \left<H \otimes \pi^a \wedge \pi^b, H \otimes \pi^a \wedge \eta,
    \left(\delta_{bc} P_a + \delta_{ac} P_b \right)\otimes \pi^c \wedge \eta\right>.
\end{equation}

The cokernel of the Spencer differential is fully reducible into
irreducible $G$-submodules:\footnote{If $n=5$ and assuming that the
  structure group $O(4)$ reduces further to $\SO(4)$, then the module
  $\eA_1$ is not irreducible but decomposes into selfdual and
  antiselfdual pieces. This is discussed briefly in
  Appendix~\ref{sec:five-dimens-arist}.}
\begin{equation}
  \coker\d \cong \eA_1 \oplus \eA_2 \oplus \eA_3 \oplus \eA_4,
\end{equation}
where, letting $W$ stand for the vector representation of $\g \cong
\so(n-1)$,
\begin{itemize}
\item $\eA_1 \cong \wedge^2 W$ consists of the image in $\coker \d$ of
  $\left<H \otimes \pi^a \wedge \pi^b\right>$;
\item $\eA_2 \cong W$ consists of the image in $\coker \d$ of
  $\left<H \otimes \pi^a \wedge \eta\right>$;
\item $\eA_3 \cong \odot_0^2 W$ (symmetric traceless) consists of the
  image in $\coker \d$ of
  \begin{equation}
    \left<\left( \delta_{bc} P_a + \delta_{ac} P_b - \tfrac2{n-1} \delta_{ab}
    P_c \right)\otimes \pi^c \wedge \eta\right>;
  \end{equation}
\item and $\eA_4 \cong \RR$ consists of the image in $\coker \d$ of
  $\left<P_a \otimes \pi^a \wedge \eta\right>$;
\end{itemize}
We conclude that there are sixteen $G$-submodules of $\coker\d$ and
therefore sixteen classes of aristotelian $G$-structures.

\subsection{Geometric characterisation of aristotelian structures}
\label{sec:ari-geom-char}

Since an aristotelian $G$-structure is a simultaneous reduction of
galilean and carrollian $G$-structures, we may reuse the results
in Sections~\ref{sec:gal-geom-char} and \ref{sec:car-geom-char} in
order to characterise the sixteen classes of aristotelian
$G$-structures geometrically.  This seems to give only twelve
aristotelian classes: four carrollian structures for each of the three
galilean structures.  There is, however, a new ingredient in the
aristotelian case: namely, the Lie derivative along the vector field
$\xi$ of the one-form $\tau$.

\begin{proposition}\label{prop:L-xi-tau}
  With the above notation, $\eL_\xi \tau = \tau \circ S$, where $S(X) =
  T^\nabla(\xi,X)$.
\end{proposition}

\begin{proof}
  First of all notice that if $\nabla$ is an adapted connection then
  both $\xi$ and $\tau$ are parallel and hence the function
  $\tau(\xi)$ is constant.  (Being nonzero, we can assume that it is
  equal to $1$ without loss of generality, simply by rescaling either
  $\tau$ or $\xi$.)  Using this and the Cartan formula, we have that
  \begin{equation}
    \eL_\xi \tau = \imath_\xi d\tau.
  \end{equation}
  But from Proposition~\ref{prop:gal-d-tau}, $d\tau = \tau \circ
  T^\nabla$, so that
  \begin{equation}
    \eL_\xi \tau = \imath_\xi (\tau \circ T^\nabla) = \tau \circ (\imath_\xi T^\nabla).
  \end{equation}
\end{proof}

It follows that if $d\tau = 0$ then $\eL_\xi\tau = 0$, whereas if
$d\tau \neq 0$ but $\tau \wedge d\tau = 0$, then
$\eL_\xi \tau \neq 0$. If $\tau \wedge d \tau \neq 0$, then it may or
may not happen that $\eL_\xi \tau = 0$.

It is now simply a matter of inspecting the sixteen classes and
determine whether $d\tau = 0$ or $\tau \wedge d\tau = 0$ or $d\tau$ is
unconstrained, and then whether $\xi$ leaves invariant $\tau$, $h$ and
the volume form $\mu$ or whether it rescales $h$.  The results are
summarised as follows.

\begin{theorem}\label{thm:ari}
  Let\footnote{See Appendices~\ref{sec:two-dimens-galil} for $n=2$ and
    \ref{sec:five-dimens-arist} for $n=5$.} $n>2$ and $n\neq 5$.  An
  aristotelian $G$-structure on an $n$-dimensional manifold $M$ can be
  of sixteen different classes depending on its intrinsic torsion.
  These classes are summarised in the table below.  Each class is
  labelled by the submodule of $\coker\d$ where the intrinsic torsion
  lands and is characterised geometrically as indicated.
  \begin{equation*}
    \rowcolors{2}{blue!10}{white}
    \begin{tabular}{>{$}l<{$}|*{3}{>{$}c<{$}}}\toprule
      \multicolumn{1}{c|}{Submodule of $\coker\d$} & \multicolumn{3}{c}{Geometric characterisation}\\\midrule
      \eA_0:=0 & d\tau = 0 & & \eL_\xi h = 0 \\
      \eA_1 & & \eL_\xi\tau = 0 & \eL_\xi h =0 \\
      \eA_2 & \tau \wedge d\tau = 0 & & \eL_\xi h = 0 \\
      \eA_3 & d\tau = 0 & \eL_\xi \tau = 0 & \eL_\xi \mu = 0\\
      \eA_4 & d\tau = 0 & \eL_\xi \tau = 0 & \eL_\xi h = f h\\
      \eA_5:=\eA_1 \oplus \eA_2 & & & \eL_\xi h = 0\\
      \eA_6:=\eA_1 \oplus \eA_3 & & \eL_\xi \tau = 0 & \eL_\xi \mu = 0 \\
      \eA_7:=\eA_1 \oplus \eA_4 & & \eL_\xi \tau = 0 & \eL_\xi h = fh \\
      \eA_8:=\eA_2 \oplus \eA_3 & \tau \wedge d\tau = 0 & & \eL_\xi \mu = 0\\
      \eA_9:=\eA_2 \oplus \eA_4 & \tau \wedge d\tau = 0 & & \eL_\xi h = fh\\
      \eA_{10}:=\eA_3 \oplus \eA_4 & d\tau = 0 & & \\
      \eA_{11}:=\eA_1 \oplus \eA_2 \oplus \eA_3 & & & \eL_\xi \mu = 0\\
      \eA_{12}:=\eA_1 \oplus \eA_2 \oplus \eA_4 & & & \eL_\xi h = fh\\
      \eA_{13}:=\eA_1 \oplus \eA_3 \oplus \eA_4 & & \eL_\xi \tau = 0 & \\
      \eA_{14}:=\eA_2 \oplus \eA_3 \oplus \eA_4 & \tau \wedge d\tau = 0 & & \\
      \eA_{15}:=\coker\d & & & \\
      \bottomrule
    \end{tabular}
  \end{equation*}
\end{theorem}

\section{Bargmannian $G$-structures}
\label{sec:bar-g-struct}

Bargmannian structures were introduced in \cite{PhysRevD.31.1841}, where
the relation between bargmannian and galilean structures was
initially explored.  In particular, it was shown that Newton--Cartan
gravity could be obtained as a null-reduction of a pp-wave: a
lorentzian manifold with a nonzero parallel null vector field.  The
relation between bargmannian, galilean and carrollian structures was
further explored in \cite{Duval:2014uoa} and in \cite{Morand:2018tke}.
Although both papers concentrate on pp-waves, the latter paper
announces some work where the bargmannian structure is allowed to be more
general.  Indeed, below we will see that pp-waves are precisely the
bargmannian manifolds with vanishing intrinsic torsion, which are one of
(generically) twenty-seven different classes of bargmannian structures.

\subsection{The group $G$ of a bargmannian structure}
\label{sec:bar-group}

In this section, we will assume that the dimension of the manifold $M$
is $n+1$.  Therefore in this section $V = \RR^{n+1}$.

An ($n+1$)-dimensional \textbf{bargmannian structure} on $M$
is a pair $(g,\xi)$ consisting of a lorentzian metric $g$ and a
nowhere-vanishing null vector field $\xi$: $g(\xi,\xi) = 0$.  Since
$\xi$ is nowhere-vanishing, around every point in $M$ we may construct
local Witt frames $(e_+=\xi,e_-,e_a)$, with $a = 1,\dots,n-1$, where
$g(e_\pm, e_a) = 0$, $g(e_+,e_-)=1$ and $g(e_a,e_b)=\delta_{ab}$.
On overlaps, such frames are related by local $G$-transformations,
where $G$ is the subgroup of the Lorentz group of $V$ which preserves
$e_+$.  Explicitly,
\begin{equation}
  \label{eq:G-bar}
  G = \left\{
    \begin{pmatrix}1 & -\tfrac12 \bv^T \bv & \bv^T \\ 0 & 1 & \bzero^T \\
    \bzero & \bv & A \end{pmatrix} ~ \middle | ~  \bv
  \in \RR^{n-1},~A \in O(n-1)\right\} < \GL(n+1,\RR),
\end{equation}
with Lie algebra
\begin{equation}
  \label{eq:g-bar}
  \g = \left\{
    \begin{pmatrix}0 & 0 & \bv^T \\ 0 & 0 & \bzero^T \\
    \bzero & \bv & A \end{pmatrix} ~ \middle | ~  \bv
  \in \RR^{n-1},~A \in \so(n-1)\right\} < \gl(n+1,\RR).
\end{equation}

Let us choose a Witt basis $(Z,H,P_a)$ for $V$ with lorentzian inner
product $\gamma \in \odot^2V^*$ given by $\gamma(Z,H)=1$ and
$\gamma(P_a,P_b)=\delta_{ab}$ and all other inner products not related
to these by symmetry vanishing.  The canonical dual basis for $V^*$
will be denoted $(\zeta,\eta, \pi^a)$.  We may choose basis
$J_{ab},B_a$ for $\g$ with brackets given by
equation~\eqref{eq:gal-g-brackets}.  Indeed, the Lie algebras $\g$ in
the galilean (and carrollian) and bargmannian cases are abstractly
isomorphic, but whereas the galilean algebra $\g < \gl(n,\RR)$, the
bargmannian algebra $\g < \gl(n+1,\RR)$.

The $G$-modules $V$ and $V^*$ are isomorphic: they are indecomposable,
but not irreducible.

\begin{lemma}\label{lem:bar-G-mod-V}
  There are $G$-invariant filtrations
  \begin{equation*}
    0 \subset \left<Z\right> \subset Z^\perp \subset V
    \qquad\text{and}\qquad
    0 \subset \left<\eta\right> \subset \Ann Z \subset V^*,
  \end{equation*}
  where $Z^\perp = \left<Z,P_a\right>$ and $\Ann Z =
\left<\eta,\pi^a\right>$.
\end{lemma}

\begin{proof}
  This follows from the explicit actions of $\g$ on $V$ and $V^*$:
  \begin{equation*}
    \begin{aligned}[m]
      J_{ab} \cdot P_c &= \delta_{bc} P_a - \delta_{ac} P_b\\
      J_{ab} \cdot H &= 0\\
      J_{ab} \cdot Z &= 0\\
      B_a \cdot P_b &= \delta_{ab} Z\\
      B_a \cdot H &= P_a\\
      B_a \cdot Z &=0
    \end{aligned}
    \qquad\text{and}\qquad
    \begin{aligned}[m]
      J_{ab} \cdot \pi^c &= (-\delta^c_a \delta_{bd} + \delta^c_b \delta_{ad}) \pi^d\\
      J_{ab} \cdot \eta &= 0\\
      J_{ab} \cdot \zeta &= 0\\
      B_a \cdot \pi^b &= -\delta_a^b \eta\\
      B_a \cdot \eta &= 0\\
      B_a \cdot \zeta &= - \delta_{ab} \pi^b.
    \end{aligned}
  \end{equation*}
\end{proof}

\subsection{The intrinsic torsion of a bargmannian structure}
\label{sec:bar-intr-tors}

Since $\g \subset \so(V)$, Lemma~\ref{lem:spencer-iso-soV} says that
the Spencer differential $\d : \g \otimes V^* \to V \otimes \wedge^2
V^*$ is injective.

\begin{proposition}\label{prop:bar-coker-d}
  As $G$-modules, $\coker \d \cong Z^\perp \otimes V^*$.
\end{proposition}

\begin{proof}
  Since $\ker\d =0$, it is enough to exhibit a short exact sequence of
  $G$-modules
  \begin{equation}
    \begin{tikzcd}
      0 \ar[r] & \Hom(V,\g) \ar[r,"\d"] & \Hom(\wedge^2 V,V)
      \ar[r,"\lambda"] & \Hom(V,Z^\perp)\ar[r] & 0,
    \end{tikzcd}
  \end{equation}
  for some $G$-equivariant map $\lambda: \Hom(\wedge^2 V,V) \to
  \Hom(V,Z^\perp)$, sending $\varphi \mapsto \lambda_\varphi$.  Let us
  define $\lambda_\varphi$ by
  \begin{equation}
   2 \gamma(\lambda_\varphi(v),w) :=\gamma(\varphi(Z,v),w) +
   \gamma(\varphi(Z,w),v) + \gamma(\varphi(v,w),Z).
  \end{equation}
  Putting $w=Z$ and using the skew-symmetry of $\varphi$, we see that
  $\gamma(\lambda_\varphi(v),Z) = 0$ for all $v$. Therefore
  $\lambda_\varphi : V \to Z^\perp$, as desired.  The map $\lambda$ is
  $G$-equivariant since it is constructed out of $\gamma$ and $Z$,
  which are $G$-invariant.  It remains to show that $\im \d = \ker
  \lambda$.

  Let us first show that $\im \d \subset \ker \lambda$.  Suppose that
  $\varphi = \d\kappa$, so that $\varphi(u,v) = \kappa_u v - \kappa_v
  u$.  Then
  \begin{equation*}
    2 \gamma(\lambda_{\d\kappa}(v),w) = \gamma(\kappa_Z v -\kappa_v Z,w) +
    \gamma(\kappa_Z w - \kappa_w Z,v) + \gamma(\kappa_v w - \kappa_w  v,Z).
  \end{equation*}
  Since $\kappa_v \in \so(V)$ for all $v \in V$,  we have that
  $\gamma(\kappa_v u, w) = - \gamma(\kappa_vw, u)$ for all $u,v,w \in
  V$, and using this we can show that the terms in the RHS cancel
  pairwise.

  Conversely, let $\varphi \in \ker\lambda$, so that
  \begin{equation}
    \label{eq:bar-aux}
    \gamma(\varphi(Z,v),w) + \gamma(\varphi(Z,w),v) +
    \gamma(\varphi(v,w),Z) = 0 \quad \forall~v,w \in V.
  \end{equation}
  Notice that the first two terms are symmetric in $(v,w)$, whereas
  the third term is skew-symmetric, so that both terms are zero
  separately.  From Lemma~\ref{lem:spencer-iso-soV}, we know that
  there exists a unique $\kappa: V \to \so(V)$ such that $\varphi
  = \d\kappa$.  We claim that if $\lambda_\varphi = 0$, then $\kappa$
  actually maps to $\g$.  Write $\varphi(v,w) = \kappa_v w - \kappa_w
  v$ with $\kappa : V \to \so(V)$ and insert this expression into the
  symmetric and skew-symmetric components of
  equation~\eqref{eq:bar-aux}.  The skew-symmetric component gives
  \begin{equation*}
    \gamma(\kappa_v w - \kappa_w v, Z) = 0
  \end{equation*}
  and the symmetric component gives
  \begin{equation*}
    \gamma(\kappa_Z v - \kappa_v Z,w) + \gamma(\kappa_Z w - \kappa_w
    Z, v) = 0 \iff \gamma(\kappa_v w + \kappa_w v, Z)  =0.
  \end{equation*}
  Adding the two equations we see that $\gamma(Z,\kappa_v w) = 0$  for
  all $v,w \in V$, which is equivalent to $\gamma(\kappa_v Z, w) = 0$
  for all $v,w \in V$.  Since $\gamma$ is nondegenerate, this says
  $\kappa_v Z = 0$ for all $v \in V$ and hence $\kappa : V \to \g$ as
  desired.
\end{proof}

We must now determine the $G$-submodules of $Z^\perp\otimes V^*$ and
hence the $G$-submodules of $\coker\d$.  Our strategy is to first
decompose $Z^\perp \otimes V^*$ into irreducible $\so(n-1)$-modules and
then to see how the bargmannian boosts $B_a$ act on them.  First of all,
$Z^\perp \otimes V^* = \left<Z,P_a\right> \otimes \left<\zeta, \eta,
  \pi^b\right>$, resulting in the following $\so(n-1)$-submodules:
\begin{equation}
  \left<Z \otimes \zeta\right>, \quad \left<Z \otimes \eta\right>,
  \quad \left<Z \otimes \pi^b\right>, \quad \left<P_a \otimes
    \zeta\right>, \quad \left<P_a \otimes \eta\right>
  \quad\text{and}\quad \left<P_a \otimes \pi^b\right>.
\end{equation}
All submodules but the last are irreducible.  The last submodule
breaks up into three\footnote{except for $n=5$, in which
  case $\left<P_a \otimes
    \pi^b\right>_{\wedge^2}$ breaks up further into selfdual and
  antiselfdual pieces. That case will be treated separately in
  Appendix~\ref{sec:six-dimens-bargm}.}
irreducible submodules:
\begin{equation}
  \left<P_a \otimes \pi^b\right>=  \left<P_a \otimes
    \pi^b\right>_{\wedge^2}\oplus \left<P_a \otimes
    \pi^b\right>_{\odot_0^2} \oplus \left<P_a \otimes
    \pi^b\right>_{\mathrm{tr}},
\end{equation}
where
\begin{equation}
  \begin{split}
    \left<P_a \otimes \pi^b\right>_{\wedge^2} &= \left<(\delta_{bc} P_a - \delta_{ac} P_b) \otimes \pi^c\right>\\
    \left<P_a \otimes \pi^b\right>_{\odot_0^2} &= \left<(\delta_{bc} P_a + \delta_{ac} P_b - \tfrac2{n-1}\delta_{ab} P_c) \otimes \pi^c\right>\\
    \left<P_a \otimes \pi^b\right>_{\mathrm{tr}} &= \left< P_a \otimes \pi^a\right>.
  \end{split}
\end{equation}
The action of the boosts are given by
\begin{equation}
  \begin{split}
    B_c \cdot (Z \otimes \pi^a) &= - \delta_c^a Z \otimes \eta\\
    B_c \cdot (Z \otimes \eta) &= 0\\
    B_c \cdot (Z \otimes \zeta) &= -\delta_{ac} Z \otimes \pi^a\\
    B_c \cdot (P_a \otimes \eta) &= \delta_{ac} Z \otimes \eta\\
    B_c \cdot (P_a \otimes \zeta) &= \delta_{ac} Z \otimes \zeta - \delta_{cd} P_a \otimes \pi^d\\
    B_c \cdot (P_a \otimes \pi^b) &= \delta_{ac} Z \otimes \pi^b - \delta^b_c P_a \otimes \eta.
  \end{split}
\end{equation}
Projecting the last term into its three irreducible components,
\begin{equation}
  \begin{split}
    B_c \cdot (P_a \otimes \pi^a) &= \delta_{ac} Z \otimes \pi^a - P_c \otimes \eta\\
    B_c \cdot ((\delta_{bd} P_a + \delta_{ad} P_b -
    \tfrac2{n-1}\delta_{ab} P_d) \otimes \pi^d) &=  (\delta_{ac}\delta_{bd} + \delta_{bc}\delta_{ad} -
    \tfrac2{n-1}\delta_{ab}\delta_{cd}) Z \otimes \pi^d \\
    & \quad {} - (\delta_{bc}P_a + \delta_{ac} P_b - \tfrac2{n-1}\delta_{ab} P_c) \otimes \eta\\
    B_c \cdot (\delta_{bd} P_a - \delta_{ad} P_b) \otimes \pi^d &=
    (\delta_{ac}\delta_{bd} - \delta_{bc}\delta_{ad}) Z \otimes \pi^d
    - (\delta_{bc} P_a - \delta_{ac} P_b) \otimes \eta.
  \end{split}
\end{equation}
We observe that $Z \otimes \pi^a + \delta^{ab} P_b \otimes \eta$ is
boost-invariant, and hence we introduce
\begin{equation}
  \Xi_{\pm}^a := Z \otimes \pi^a \pm \delta^{ab} P_b \otimes \eta,
\end{equation}
so that
\begin{equation}
  B_c \cdot \Xi_+^a = 0 \qquad\text{and}\qquad B_c \cdot \Xi_-^a = -2
  \delta^a_c Z \otimes \eta.  
\end{equation}
We rewrite
\begin{equation}
  Z \otimes \pi^a = \tfrac12 (\Xi_+^a + \Xi_-^a)
  \qquad\text{and}\qquad P_a \otimes \eta = \tfrac12 \delta_{ab}
  (\Xi_+^b - \Xi_-^b),
\end{equation}
in terms of which
\begin{equation}
  \begin{split}
    B_c \cdot (P_a \otimes \pi^a) &= \delta_{cd} \Xi^d_-\\
    B_c \cdot (\delta_{bd} P_a - \delta_{ad}P_b) \otimes \pi^d &= (\delta_{bd}\delta_{ac}- \delta_{ad}\delta_{bc}) \Xi_+^d\\
    B_c \cdot (\delta_{bd} P_a + \delta_{ad}P_b -
    \tfrac2{n-1}\delta_{ab}P_d) \otimes \pi^d &=  (\delta_{bd}\delta_{ac} +\delta_{ad}\delta_{bc} -\frac2{n-1}\delta_{ab}\delta_{cd}) \Xi^d_-.
  \end{split}
\end{equation}
We may summarise this graphically as in
Figure~\ref{fig:hasse-bar-subspaces}, where we have used
self-explanatory abbreviations for the irreducible
$\so(n-1)$-submodules and where the arrows represent the action of the
boosts $B_a$.

\begin{figure}[h!]
  \centering
  \begin{tikzpicture}[scale=1]
    \tikzset{myshortedge/.style={->,thick,red,shorten <=2pt,shorten >=24pt}}
    \tikzset{mylongedge/.style={->,thick,red,shorten <=2pt,shorten >=18pt}}
    \coordinate[label={$\left<P\otimes\zeta\right>$}] (8) at (0,5); 
    \coordinate[label={$\left<P\otimes\pi\right>_{\wedge^2}$}] (4) at (-3,3); 
    \coordinate[label={$\left<Z\otimes\zeta\right>$}] (5) at (-1,3); 
    \coordinate[label={$\left<P\otimes\pi\right>_{\odot_0^2}$}] (6) at (1,3); 
    \coordinate[label={$\left<P\otimes\pi\right>_{\mathrm{tr}}$}] (7) at (3,3); 
    \coordinate[label={$\left<\Xi_+\right>$}] (2) at (-2,1); 
    \coordinate[label={$\left<\Xi_-\right>$}] (3) at (1,1); 
    \coordinate[label={$\left<Z\otimes\eta\right>$}] (1) at (1,0);
    \coordinate[label={$0$}] (0) at (0,-1);
    \draw[myshortedge] (8) -- (4); 
    \draw[myshortedge] (8) -- (5); 
    \draw[myshortedge] (8) -- (6); 
    \draw[myshortedge] (8) -- (7); 
    \draw[mylongedge] (4) -- (2); 
    \draw[mylongedge] (5) -- (2); 
    \draw[mylongedge] (5) -- (3); 
    \draw[mylongedge] (6) -- (3); 
    \draw[mylongedge] (7) -- (3); 
    \draw[mylongedge] (3) -- (1); 
    \draw[mylongedge] (2) -- (0); 
    \draw[mylongedge] (1) -- (0);
  \end{tikzpicture}
  \caption{Action of boosts on $\so(n-1)$-submodules of $Z^\perp \otimes V^*$}
  \label{fig:hasse-bar-subspaces}
\end{figure}
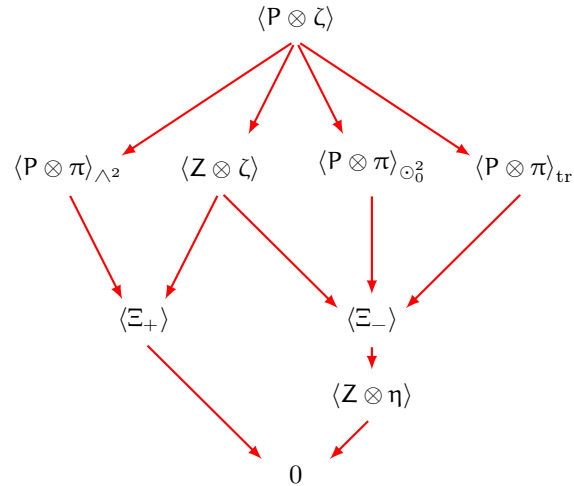

It is possible now to list the $\g$-submodules of $Z^\perp \otimes
V^*$ from Figure~\ref{fig:hasse-bar-subspaces}: if a certain
$\so(n-1)$-submodule appears, then all other $\so(n-1)$-submodules
which can be reached from it following the arrows in the diagram must
appear as well.  This process yields twenty-seven $\g$-submodules,
which we proceed to list below in abbreviated form.  The reason for
the primes is that these are the submodules of $Z^\perp \otimes V^*$
and we are eventually interested in the submodules of $\coker\d$.

\begin{itemize}
\item $\eB'_0 = 0$
\item $\eB'_1 = \left<Z\otimes\eta\right>$
\item $\eB'_2 = \left<\Xi_+\right>$
\item $\eB'_3 = \left<Z\otimes\eta, \Xi_+\right>$
\item $\eB'_4 = \left<Z\otimes\eta, \Xi_-\right>$
\item $\eB'_5 = \left<\Xi_+,(P\otimes\pi)_{\wedge^2}\right> $
\item $\eB'_6 = \left<Z\otimes\eta,\Xi_-, (P\otimes\pi)_{\odot_0^2}\right>$
\item $\eB'_7 = \left<Z\otimes\eta,\Xi_-, (P\otimes\pi)_{\mathrm{tr}}\right>$
\item $\eB'_8 = \left<Z\otimes\eta, \Xi_+,\Xi_-\right>$
\item $\eB'_9 = \left<Z\otimes\eta, \Xi_+,(P\otimes\pi)_{\wedge^2}\right>$
\item $\eB'_{10} = \left<Z\otimes\eta,\Xi_-, (P\otimes\pi)_{\odot_0^2},(P\otimes\pi)_{\mathrm{tr}}\right>$
\item $\eB'_{11} = \left<Z\otimes\eta, \Xi_+,\Xi_-, (P\otimes\pi)_{\odot_0^2}\right>$
\item $\eB'_{12} = \left<Z\otimes\eta, \Xi_+,\Xi_-, (P\otimes\pi)_{\mathrm{tr}}\right>$
\item $\eB'_{13} = \left<Z\otimes\eta, \Xi_+,\Xi_-, (P\otimes\pi)_{\wedge^2}\right>$
\item $\eB'_{14} = \left<Z\otimes\eta, \Xi_+,\Xi_-, Z\otimes\zeta\right>$
\item $\eB'_{15} = \left<Z\otimes\eta, \Xi_+,\Xi_-,  (P\otimes\pi)_{\wedge^2}, Z\otimes\zeta\right>$
\item $\eB'_{16} = \left<Z\otimes\eta, \Xi_+,\Xi_-, (P\otimes\pi)_{\odot_0^2}, Z\otimes\zeta\right>$
\item $\eB'_{17} = \left<Z\otimes\eta, \Xi_+,\Xi_-, (P\otimes\pi)_{\mathrm{tr}}, Z\otimes\zeta\right>$
\item $\eB'_{18} = \left<Z\otimes\eta, \Xi_+,\Xi_-, (P\otimes\pi)_{\odot_0^2}, (P\otimes\pi)_{\wedge^2}\right>$
\item $\eB'_{19} = \left<Z\otimes\eta, \Xi_+,\Xi_-, (P\otimes\pi)_{\mathrm{tr}}, (P\otimes\pi)_{\wedge^2}\right>$
\item $\eB'_{20} = \left<Z\otimes\eta, \Xi_+,\Xi_-, (P\otimes\pi)_{\mathrm{tr}}, (P\otimes\pi)_{\odot_0^2}\right>$
\item $\eB'_{21} = \left<Z\otimes\eta, \Xi_+,\Xi_-,  (P\otimes\pi)_{\wedge^2}, Z\otimes\zeta, (P\otimes\pi)_{\odot_0^2}\right>$
\item $\eB'_{22} = \left<Z\otimes\eta, \Xi_+,\Xi_-,  (P\otimes\pi)_{\wedge^2}, Z\otimes\zeta, (P\otimes\pi)_{\mathrm{tr}}\right>$
\item $\eB'_{23} = \left<Z\otimes\eta, \Xi_+,\Xi_-,  (P\otimes\pi)_{\odot_0^2}, Z\otimes\zeta, (P\otimes\pi)_{\mathrm{tr}}\right>$
\item $\eB'_{24} = \left<Z\otimes\eta, \Xi_+,\Xi_-,  (P\otimes\pi)_{\odot_0^2},  (P\otimes\pi)_{\wedge^2}, (P\otimes\pi)_{\mathrm{tr}}\right>$
\item $\eB'_{25} = \left<Z\otimes\eta, \Xi_+,\Xi_-,  (P\otimes\pi)_{\odot_0^2},  (P\otimes\pi)_{\wedge^2}, (P\otimes\pi)_{\mathrm{tr}},Z\otimes\zeta\right>$
\item $\eB'_{26} = Z^\perp \otimes V^*$
\end{itemize}

Next, we exhibit the dictionary between the $\so(n-1)$-submodules of
$Z^\perp \otimes V^*$ and those of $\coker\d$, which we list in
abbreviated form:
\begin{equation*}
  \begin{tabular}{>{$}c<{$}|>{$}c<{$}}\toprule
    Z^\perp \otimes V^* &  \coker\d\\\midrule
     \left<Z \otimes \eta\right> & \left<Z \otimes \eta \wedge \zeta\right> \\
     \left<\Xi_+\right> & \left<H \otimes \pi \wedge \eta\right> \\
     \left<\Xi_-\right> & \left<(Z\otimes \pi - P \otimes \eta)\wedge \zeta\right> \\
     \left<(P\otimes \pi)_{\wedge^2}\right> & \left<H \otimes \pi \wedge \pi\right> \\
     \left<Z \otimes \zeta\right> & \left<H \otimes \eta \wedge \zeta\right> \\
     \left<(P \otimes\pi)_{\odot_0^2}\right> & \left<(P\otimes \pi)_{\odot_0^2} \wedge \zeta\right> \\
     \left<(P \otimes\pi)_{\mathrm{tr}}\right> & \left<(P\otimes \pi)_{\mathrm{tr}} \wedge \zeta\right> \\
     \left<P \otimes \zeta\right> & \left<H \otimes \pi \wedge \zeta\right> \\\bottomrule
  \end{tabular}
\end{equation*}
One can use this dictionary to read off the $\g$-submodules of
$\coker\d$ from the ones of $Z^\perp \otimes V^*$ listed above and in
this way set up a correspondence between
$\eB'_i \subset Z^\perp \otimes V^*$ and $\eB_i \subset \coker \d$ for
$i=0,1,\dots,26$.  The set of twenty-seven $\g$-submodules of
$\coker\d$ is partially ordered by inclusion.
Figure~\ref{fig:bar-hasse-submodules} illustrates the Hasse diagram of
this poset.  The node labelled $i$ corresponds to the submodule $\eB_i
\subset \coker\d$ (or $\eB'_i \subset Z^\perp \otimes V^*$) and an
arrow indicates inclusion.  The meaning of the labels is explained in
Section~\ref{sec:bar-geom-char}.

\begin{figure}[h!]
  \centering
  \begin{tikzpicture}[scale=1.3]
    %
    % nodes
    %
    \coordinate[label=right:{\scriptsize $0$}] (b0) at (-1,7); 
    \coordinate[label=left:{\scriptsize $1$}] (b1) at (-1,6); 
    \coordinate[label=left:{\scriptsize $2$}] (b2) at (-2,6); 
    \coordinate[label=right:{\scriptsize $3$}] (b3) at (-1,5); 
    \coordinate[label=above right:{\scriptsize $4$}] (b4) at (0,5); 
    \coordinate[label=left:{\scriptsize $5$}] (b5) at (-2,5); 
    \coordinate[label=left:{\scriptsize $6$}] (b6) at (0,4); 
    \coordinate[label=above right:{\scriptsize $7$}] (b7) at (1,4); 
    \coordinate[label=above left:{\scriptsize $8$}] (b8) at (-1,4); 
    \coordinate[label=left:{\scriptsize $9$}] (b9) at (-2,4); 
    \coordinate[label=above right:{\scriptsize $10$}] (b10) at (2,3); 
    \coordinate[label=above right:{\scriptsize $11$}] (b11) at (0,3); 
    \coordinate[label=above right:{\scriptsize $12$}] (b12) at (1,3); 
    \coordinate[label=above left:{\scriptsize $13$}] (b13) at (-2,3); 
    \coordinate[label=above left:{\scriptsize $14$}] (b14) at (-1,3); 
    \coordinate[label=left:{\scriptsize $15$}] (b15) at (-3,2); 
    \coordinate[label=above left:{\scriptsize $16$}] (b16) at (-1,2); 
    \coordinate[label=right:{\scriptsize $17$}] (b17) at (1,2); 
    \coordinate[label=left:{\scriptsize $18$}] (b18) at (-2,2);
    \coordinate[label=below right:{\scriptsize $19$}] (b19) at (0,2);
    \coordinate[label=right:{\scriptsize $20$}] (b20) at (2,2);
    \coordinate[label=below left:{\scriptsize $21$}] (b21) at (-2,1);
    \coordinate[label=below left:{\scriptsize $22$}] (b22) at (-1,1); 
    \coordinate[label=below right:{\scriptsize $23$}] (b23) at (1,1); 
    \coordinate[label=below right:{\scriptsize $24$}] (b24) at (0,1); 
    \coordinate[label=below right:{\scriptsize $25$}] (b25) at (0,0); 
    \coordinate[label=right:{\scriptsize $26$}] (b26) at (0,-1);
    %
    % edges
    %
    \tikzset{every path/.style={->,thick,gris,shorten >=3pt, shorten <=3pt}}
    \draw (b0) -- (b1);
    \draw (b0) -- (b2);
    \draw (b1) -- (b3);
    \draw (b1) -- (b4);
    \draw (b2) -- (b3);
    \draw (b2) -- (b5);
    \draw (b3) -- (b8);
    \draw (b3) -- (b9);
    \draw (b4) -- (b6);
    \draw (b4) -- (b7);
    \draw (b4) -- (b8);
    \draw (b5) -- (b9);
    \draw (b6) -- (b10); 
    \draw (b6) -- (b11); 
    \draw (b7) -- (b10); 
    \draw (b7) -- (b12); 
    \draw (b8) -- (b11); 
    \draw (b8) -- (b12); 
    \draw (b8) -- (b13); 
    \draw (b8) -- (b14); 
    \draw (b9) -- (b13); 
    \draw (b10) -- (b20);
    \draw (b11) -- (b16);
    \draw (b11) -- (b18);
    \draw (b11) -- (b20);
    \draw (b12) -- (b17);
    \draw (b12) -- (b19);
    \draw (b12) -- (b20);
    \draw (b13) -- (b15);
    \draw (b13) -- (b18);
    \draw (b13) -- (b19);
    \draw (b14) -- (b15);
    \draw (b14) -- (b16);
    \draw (b14) -- (b17);
    \draw (b15) -- (b21);
    \draw (b15) -- (b22);
    \draw (b16) -- (b21);
    \draw (b16) -- (b23);
    \draw (b17) -- (b22);
    \draw (b17) -- (b23);
    \draw (b18) -- (b21);
    \draw (b18) -- (b24);
    \draw (b19) -- (b22);
    \draw (b19) -- (b24);
    \draw (b20) -- (b23);
    \draw (b20) -- (b24);
    \draw (b21) -- (b25); 
    \draw (b22) -- (b25); 
    \draw (b23) -- (b25); 
    \draw (b24) -- (b25); 
    \draw (b25) -- (b26);
    %
    % points
    %
    \foreach \point in {b0,b1,b2,b3,b4,b5,b8,b9,b13,b14,b15}
    \filldraw [color=blue!70!black,fill=blue!70!white] (\point) circle (1.5pt);
    \foreach \point in {b7,b12,b17,b19,b22}
    \filldraw [color=red!70!black,fill=red!70!white] (\point) circle (1.5pt);
    \foreach \point in {b6,b11,b16,b18,b21}
    \filldraw [color=green!70!black,fill=green!70!white] (\point) circle (1.5pt);
    \foreach \point in {b10,b20,b23,b24,b25,b26}
    \filldraw [color=black,fill=white] (\point) circle (1.5pt);
    %
    % legend
    %
    \draw [black,thick] (2.8,6.8) -- (2.8,4.8) -- (5.3,4.8) -- (5.3,6.8) -- cycle;
    \filldraw [color=blue!70!black,fill=blue!70!white] (3,6.5) circle (1.5pt) node [right,black] {totally geodesic}; 
    \filldraw [color=green!70!black,fill=green!70!white] (3,6) circle (1.5pt) node [right,black] {minimal};
    \filldraw [color=red!70!black,fill=red!70!white] (3,5.5) circle (1.5pt) node [right,black] {totally umbilical}; 
    \filldraw [color=black,fill=white] (3,5) circle (1.5pt) node  [right,black] {none of the above};
  \end{tikzpicture}
  \caption{Hasse diagram of bargmannian structures}
  \label{fig:bar-hasse-submodules}
\end{figure}
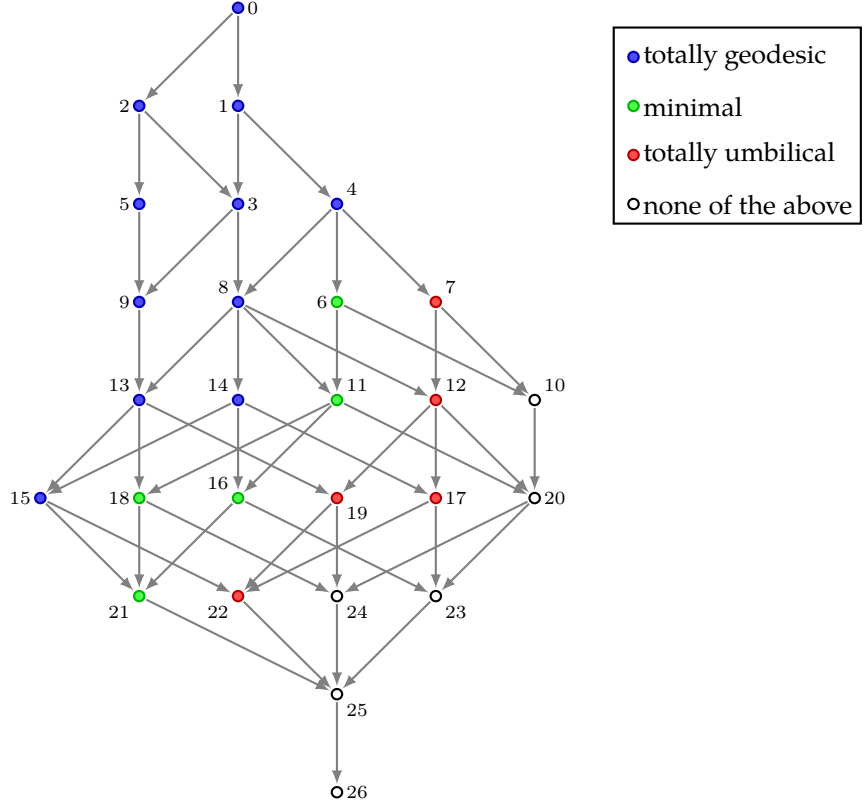

In summary, we see that there are twenty-seven
classes of bargmannian $G$-structures, which we will characterise
geometrically in the next section.

Before doing that, let us prove that the intrinsic torsion of the
bargmannian structure can be identified with the covariant derivative
$\nabla^g \xi$ of the null vector $\xi$ with respect to the
Levi-Civita connection of $g$. The Levi-Civita connection is not
adapted unless $\nabla^g \xi = 0$ so that $g$ is a Brinkmann metric
(a.k.a. a generalised pp-wave).  Nevertheless it is an intrinsic
object in $(M,g)$ and therefore $\nabla^g \xi \in \Omega^1(M,TM)$ is
intrinsic to the $G$-structure.  Since $g(\xi,\xi) =0$ and $\nabla^g$
is metric, we have that $g(\nabla^g_X\xi,\xi) = 0$ for all
$X \in \eX(M)$. Therefore, $\nabla^g \xi : TM \to \xi^\perp$ or,
equivalently, $\nabla^g \xi$ is a section of the associated vector
bundle $P \times_G (Z^\perp \otimes V^*)$, which as we saw before is
isomorphic to $P\times_G \coker \d$.  The next Proposition shows that
this is not a coincidence.

The map $\lambda : V \otimes \wedge^2 V^* \to Z^\perp \otimes V^*$ in
the proof of Proposition~\ref{prop:bar-coker-d} induces a bundle map
$TM \otimes \wedge^2 T^*M \to \xi^\perp \otimes T^*M$ and hence a
$C^\infty(M)$-linear map
\begin{equation}
  \Lambda : \Omega^2(M,TM) \to \Omega^1(M,\xi^\perp), \qquad T \mapsto \Lambda_T,
\end{equation}
defined, for all $T \in \Omega^2(M,TM)$ and $X,Y \in \eX(M)$, by
\begin{equation}\label{eq:Lambda-map}
  2 g(\Lambda_T(X), Y) :=g(T(\xi,X),Y) + g(T(\xi,Y),X) + g(T(X,Y),\xi).
\end{equation}

\begin{proposition}\label{prop:intr-tors-is-nabla}
  Let $\nabla$ be an adapted connection to a bargmannian $G$-structure
  with torsion $T^\nabla$.  Then $\Lambda_{T^\nabla} = \nabla^g \xi$.
\end{proposition}

\begin{proof}
  Let us denote by $\kappa \in \Omega^1(M,\so(TM))$ the
  contorsion $\kappa: = \nabla - \nabla^g$.  It takes values in $\so(TM)$, the
  bundle of skew-symmetric endomorphisms of $TM$, because both
  connections are metric-compatible.  Since $\nabla^g$ has zero
  torsion, $T^\nabla = \d\kappa$ and since $\nabla$ is adapted,
  \begin{equation*}
    0 = \nabla_X \xi = \nabla^g_X \xi + \kappa_X \xi \implies \kappa_X
    \xi = - \nabla^g_X \xi.
  \end{equation*}
  We calculate for $X,Y \in \eX(M)$,
  \begin{align*}
    2 g(\Lambda_{T^\nabla}(X),Y) &= g(T^\nabla(\xi,X),Y) +  g(T^\nabla(\xi,Y),X) +  g(T^\nabla(X,Y),\xi)\\
                                 &= g(\kappa_\xi X - \kappa_X \xi, Y) + g(\kappa_\xi Y - \kappa_Y \xi, X) + g(\kappa_X Y - \kappa_Y X, \xi)\\
                                 &= - g(\kappa_X \xi, Y) - g(\kappa_Y \xi, X) - g(\kappa_X \xi, Y) + g(\kappa_Y \xi, X) && \tag{using $\kappa_X \in \so(TM)$}\\
                                 &= -2 g(\kappa_X\xi, Y).
  \end{align*}
  Hence $\Lambda_{T^\nabla}(X) =- \kappa_X\xi = \nabla^g_X \xi$, as claimed.
\end{proof}

\subsection{Geometric characterisation of bargmannian structures}
\label{sec:bar-geom-char}

Let $(M^{n+1},g,\xi)$ be a bargmannian manifold; that is, $(M,g)$ is an
$(n+1)$-dimensional lorentzian manifold and $\xi \in \eX(M)$ is a
nowhere-vanishing null vector. Let $\xi^\flat \in \Omega^1(M)$ denote
the one-form dual to $\xi$: $\xi^\flat(X) = g(\xi,X)$ for all
$X \in \eX(M)$.  Let $\nu$ denote the (possibly only locally defined)
volume form.  If we assume that $M$ is orientable (e.g., if it is
simply-connected) then $\nu \in \Omega^{n+1}(M)$ defines an
orientation; that is, a nowhere-vanishing top form.

Let $\xi^\perp = \ker\xi^\flat$ denote the characteristic distribution
consisting of tangent vectors perpendicular to $\xi$.  Since $\xi$ is
null, $\xi$ belongs to the distribution and hence the restriction of
the metric to $\xi^\perp$ is degenerate.  If the distribution is involutive
($[\xi^\perp,\xi^\perp] \subset \xi^\perp$), which is equivalent to
$\xi^\flat \wedge d\xi^\flat = 0$, then $M$ is foliated by null
hypersurfaces $N$ whose tangent space $T_p N$ at $p$ coincides with
$\xi^\perp_p \subset T_p M$.  There is a well established theory of null
hypersurfaces (see, e.g., \cite{MR886772,MR1777311}) from which we
will borrow in this section.  However since not all bargmannian
structures are such that $\xi^\perp$ is involutive, we will have to extend
this theory slightly to non-involutive distributions.

It is not hard to see that in all bargmannian structures but the generic
one ($\eB_{26}$), the distribution is $\xi$-invariant; that is, for
all $X \in \Gamma(\xi^\perp)$, $[\xi, X] \in \Gamma(\xi^\perp)$, which we abbreviate
by $[\xi,\xi^\perp] \subset \xi^\perp$, with some abuse of notation.  We may also
write $X \perp \xi$ for $X \in \Gamma(\xi^\perp)$.

We recall that $\xi$ is said to be \textbf{geodetic} if $\nabla^g_\xi
\xi = f \xi$ for some $f \in C^\infty(M)$.  The name is apt, because
integral curves of $\xi$ can be parametrised in such a way that they
satisfy the geodesic equation.

\begin{lemma}\label{lem:geodetic}
  Let $\xi^\perp = \ker \xi^\flat$.  Then $[\xi,\xi^\perp] \subset \xi^\perp$ if and only if $\xi$ is geodetic.
\end{lemma}

\begin{proof}
  We have that $\xi$ is geodetic if and only if $g(\nabla^g_\xi \xi,X)
  = 0$ for all $X \perp \xi$ and then
  \begin{align*}
    g(\nabla^g_\xi \xi,X) = 0 & \iff g(\xi,\nabla^g_\xi X) = 0\\
                              & \iff g(\xi,[\xi,X] + \nabla^g_X\xi)= 0 &&\tag{since $\nabla^g$ has zero torsion}\\
                              & \iff g(\xi,[\xi,X]) = 0&&\tag{since $g(\xi,\nabla^g_X\xi) = 0$ for any $X$}\\
                              & \iff [\xi,X] \perp \xi.
  \end{align*}
\end{proof}

From now on we will assume that $[\xi,\xi^\perp] \subset \xi^\perp$; that is, we are
dealing with any one but the generic bargmannian structure.

Let $L \subset \xi^\perp$ denote the line sub-bundle spanned by $\xi$ and let
$E := \xi^\perp/L$ denote the quotient vector bundle: it is a corank-2 vector
bundle over $M$.  If $X \perp \xi$, we let $\Xbar \in \Gamma(E)$
denote its equivalence class modulo $L$; that is, $X,Y \perp \xi$
satisfy $\Xbar = \Ybar$ if and only if $X - Y = f \xi$ for some
$f \in C^\infty(M)$.  On $E$ we have a positive-definite metric $h$
defined by
\begin{equation}
  \label{eq:metric-on-quotient}
  h(\Xbar, \Ybar) := g(X,Y) \qquad\forall~X,Y \perp \xi.
\end{equation}
This is well-defined on equivalence classes precisely because $X,Y
\perp \xi$ and it is positive-definite because $(M,g)$ is lorentzian.
This makes $(E,h)$ into a corank-2 riemannian vector bundle over $M$.

Following \cite{MR1777311} we define the \textbf{null Weingarten map}
$W : \Gamma(E) \to \Gamma(E)$ by
\begin{equation}
  \label{eq:weingarten}
  W(\Xbar) := \overline{\nabla^g_X \xi}.
\end{equation}
Although in \cite{MR1777311} this map is shown to be well-defined for
the case of involutive $\xi^\perp$, it turns out that it is well-defined under the
weaker hypothesis that $[\xi,\xi^\perp]\subset \xi^\perp$.  Indeed, let $\Xbar =
\Ybar$ and calculate
\begin{align*}
  W(\Xbar) - W(\Ybar) &= \overline{\nabla^g_X\xi} - \overline{\nabla^g_Y\xi}\\
                        % &= \overline{\nabla^g_X\xi - \nabla^g_Y\xi}\\
                        &= \overline{\nabla^g_{X-Y} \xi}\\
                        &= \overline{\nabla^g_{f\xi}\xi}\\
                        &= f \overline{\nabla^g_\xi \xi}\\
                        &= 0,
\end{align*}
where we have used that $\xi$ is geodetic, which as shown in
Lemma~\ref{lem:geodetic} follows by virtue of $[\xi,\xi^\perp] \subset \xi^\perp$.

We define the \textbf{null second fundamental form}
$B \in \Gamma(E^* \otimes E^*)$ by
\begin{equation}
  \label{eq:null-sff}
  B(\Xbar, \Ybar) := h(W(\Xbar), \Ybar) = g(\nabla^g_X \xi, Y).
\end{equation}
We see that this is well-defined because both $h$ and $W$ are
well-defined.  In contrast to the case of a null hypersurface, the
second fundamental form of $\xi^\perp$ need not be symmetric.

\begin{lemma}
  The null second fundamental form $B$ is symmetric if and only if $\xi^\perp$
  is involutive.
\end{lemma}

\begin{proof}
  We calculate
  \begin{align*}
    B(\Xbar, \Ybar) - B(\Ybar, \Xbar) &= g(\nabla^g_X \xi, Y) - g(\nabla^g_Y \xi, X)\\
                                          &= - g(\xi, \nabla^g_X Y) + g(\xi, \nabla^g_Y,X) &&\tag{since $\nabla^gg=0$ and $X,Y \perp \xi$}\\
                                          &= - g(\xi, [X,Y]) &&\tag{since $\nabla^g$ has zero torsion}
  \end{align*}
  and hence $B$ is symmetric if and only if $[X,Y] \perp \xi$ for all
  $X,Y \perp \xi$; that is, if and only if $[\xi^\perp,\xi^\perp] \subset \xi^\perp$.
\end{proof}

Let $B_{\mathrm{sym}}$ denote (twice) the symmetric part of $B$:
\begin{equation}
  B_{\mathrm{sym}} (\Xbar, \Ybar) := B(\Xbar,\Ybar) + B(\Ybar, \Xbar).
\end{equation}

\begin{lemma}
  As sections of $\odot^2 E^*$, $B_{\mathrm{sym}} = \eL_\xi h$.
\end{lemma}

\begin{proof}
  Let $X,Y \perp \xi$.  Then
  \begin{align*}
    (\eL_\xi h)(\Xbar, \Ybar) &= (\eL_\xi g)(X,Y)\\
                                &= g(\nabla^g_X\xi,Y) + g(\nabla^g_Y\xi,X)\\
                                &= B(\Xbar,\Ybar) + B(\Ybar, \Xbar).
  \end{align*}
\end{proof}

\begin{definition}
  Let $\xi^\perp= \ker \xi^\flat$ be $\xi$-invariant, so that $[\xi,\xi^\perp] \subset
  \xi^\perp$.   We say that $\xi^\perp$ is
  \begin{itemize}
  \item \textbf{totally geodesic} if $B_{\mathrm{sym}} = 0$;
  \item \textbf{minimal} if $\tr B_{\mathrm{sym}} = 0$; and
  \item \textbf{totally umbilical} if $B_{\mathrm{sym}} = f h$ for some $f \in C^\infty(M)$.
  \end{itemize}
\end{definition}

Of course, if $\xi^\perp$ is involutive, then $B_{\mathrm{sym}}=B$ and hence
these are the natural extension to null hypersurfaces of the
well-known concepts for hypersurfaces of riemannian manifolds.  Even
if $\xi^\perp$ is not involutive, the condition of being totally geodesic
simply says that any lorentzian geodesic whose initial velocity
belongs to $\xi^\perp$ is such that its velocity remains in $\xi^\perp$.  Indeed,
suppose that $\gamma$ is a geodesic with $\gamma(0) = p$ and
$\dot\gamma(0) \in \xi^\perp_p$.  Consider the function of $t$ defined by $t
\mapsto g(\dot\gamma(t),\xi(\gamma(t)))$.  This function vanishes at
$t=0$ because $\dot\gamma(0) \in \xi^\perp_p$.  Differentiating with respect
to $t$, we obtain
\begin{align*}
  \tfrac{d}{dt} g(\dot\gamma, \xi) &= g(\tfrac{D}{dt}\dot\gamma, \xi) + g(\dot\gamma, \nabla^g_{\dot\gamma} \xi)\\
                                 &= g(\dot\gamma, \nabla^g_{\dot\gamma} \xi) && \tag{since $\gamma$ is a geodesic}\\
                                 &= B(\overline{\dot\gamma},\overline{\dot\gamma}).
\end{align*}
By polarisation, this vanishes for all $\dot\gamma$ if and only if $B$
is skew-symmetric.  If (and only if) that is the case, then
$g(\dot\gamma,\xi) = 0$ for all $t$.

Before we go on to characterise geometrically the different bargmannian
structures, let us observe that many of the calculations already done
in the galilean, carrollian and aristotelian sections imply some
results also for bargmannian structures.

If $\nabla$ is an adapted connection, then $\nabla g = 0$, $\nabla \xi
= 0$, $\nabla \xi^\flat = 0$ and $\nabla \nu = 0$.  Let $T^\nabla \in
\Omega^2(M,TM)$ denote its torsion.  We define $S \in \Omega^1(M,TM)$
by $S(X) := T^\nabla(\xi,X)$ for all $X \in \eX(M)$ and $\Sigma \in
\Gamma(\odot^2 T^*M)$ by $\Sigma(X,Y) := g(S(X),Y) + g(S(Y),X)$.  The
following result follows from the calculations already done in the
galilean, carrollian and aristotelian sections.

\begin{corollary}
  With the notation of the previous paragraph, the following
  identities hold:
  \begin{equation}
    d\xi^\flat = \xi^\flat \circ T^\nabla, \qquad \eL_\xi g = \Sigma,
    \qquad \eL_\xi \nu = \tr(S) \nu \qquad\text{and}\qquad \eL_\xi
    \xi^\flat = \xi^\flat \circ S.
  \end{equation}
\end{corollary}

\begin{proof}
  The proof of Proposition~\ref{prop:gal-d-tau} shows that $d\xi^\flat =
  \xi^\flat \circ T^\nabla$, whereas the proof of
  Proposition~\ref{prop:L-xi-h} shows that $\eL_\xi g = \Sigma$ and
  that of Proposition~\ref{prop:L-xi-mu} shows that $\eL_\xi \nu =
  \tr(S) \nu$.  Finally, the proof of Proposition~\ref{prop:L-xi-tau}
  shows that $\eL_\xi \xi^\flat = \xi^\flat \circ S$.
\end{proof}

It follows from this corollary, from
Proposition~\ref{prop:intr-tors-is-nabla} and from the explicit form
of the map $\Lambda$ in equation~\eqref{eq:Lambda-map} that
\begin{equation}
  2 g(\nabla^g_X\xi,Y) = \Sigma(X,Y) +
  d\xi^\flat(X,Y),\qquad\forall~X,Y\in\eX(M).
\end{equation}
In other words, $\Sigma$ and $d\xi^\flat$ are the extensions to $TM$
of the symmetric and skew-symmetric parts of the second fundamental
form of the distribution $\xi^\perp$, respectively.  In particular, it should
be emphasised that the skew-symmetric component of $S$ is not related to
the skew-symmetric component of the second fundamental form.

In order to make full use of these results in the geometric
characterisation of bargmannian structures, we should first determine
which $\so(n-1)$-submodules of $\coker\d$ (or of $Z^\perp \otimes
V^*$) contribute to which geometric data.  This can be read off by
inspection and the results are collected in
Table~\ref{tab:contributions}, which turns out to be quite handy.  In
that table, a $\bullet$ indicates that the submodule does contribute
and $\circ$ indicates that it does not.

\begin{table}[h!]
  \centering
  \begin{tabular}{*{2}{>{$}c<{$}}|*{9}{>{$}c<{$}}}\toprule
    \multicolumn{2}{c|}{$\so(n-1)$-submodule of} & & \multicolumn{2}{c}{$\Sigma$} & & & & & \multicolumn{2}{c}{$B_{\mathrm{sym}}$} \\
    \multicolumn{1}{>{$}c<{$}}{Z^\perp \otimes V^*} & \multicolumn{1}{>{$}c<{$}|}{\coker\d} & S_{\wedge^2} & S_{\odot_0^2} & \tr(S) & \xi^\flat \circ S & d\xi^\flat & \xi^\flat \wedge d\xi^\flat & B_{\wedge^2} & B_{\odot_0^2} & \tr(B) \\ \midrule
    \left<Z \otimes \eta\right> & \left<Z\otimes \eta \wedge \zeta\right> & \circ  &\bullet & \circ & \circ & \circ & \circ  & \circ & \circ & \circ \\
    \left<Z \otimes \pi + P \otimes \eta \right> & \left<H\otimes \pi \wedge \eta\right> & \circ &\circ & \circ & \circ & \bullet & \circ & \circ & \circ & \circ \\
    \left<Z \otimes \pi - P \otimes \eta \right> & \left<(Z\otimes\pi - P \otimes \eta) \wedge \zeta\right> &\bullet & \bullet & \circ & \circ & \circ & \circ & \circ & \circ & \circ \\
    \left<(P\otimes \pi)_{\wedge^2}\right> & \left<H\otimes \pi \wedge \pi\right> &\circ & \circ & \circ  & \circ & \bullet & \bullet & \bullet & \circ & \circ \\
    \left<Z \otimes \zeta\right> & \left<H\otimes \eta \wedge \zeta\right> &\bullet &\bullet&\bullet&\bullet&\bullet&\circ & \circ & \circ & \circ \\
    \left<(P \otimes\pi)_{\odot_0^2}\right> & \left<(P\otimes \pi)_{\odot_0^2} \wedge \zeta\right> &\circ &\bullet&\circ &\circ &\circ &\circ  & \circ & \bullet & \circ \\
    \left<(P \otimes\pi)_{\mathrm{tr}}\right> & \left<(P\otimes \pi)_{\mathrm{tr}} \wedge \zeta\right> &\circ &\circ&\bullet&\circ &\circ &\circ  & \circ & \circ & \bullet \\
    \left<P \otimes \zeta\right> & \left<H\otimes \pi \wedge \zeta\right> &\bullet &\bullet&\circ&\bullet&\bullet&\bullet & \bullet & \bullet & \circ \\\bottomrule
  \end{tabular}
  \vspace{1em}
  \caption{Contributions of each $\so(n-1)$-submodule of $Z^\perp \otimes V^* \cong \coker\d$\\
    ($\bullet$ contributes and $\circ$ does not)}
  \label{tab:contributions}
\end{table}

\subsubsection{Totally geodesic bargmannian structures}
\label{sec:bar-tot-geo}

We start with those bargmannian structures which are totally geodesic;
that is, for which the symmetric part of the second fundamental form
vanishes.  From Table~\ref{tab:contributions} we see that only the
$\so(n-1)$-submodules $\left<(P \otimes\pi)_{\odot_0^2}\right> \cong
\left<(P\otimes \pi)_{\odot_0^2} \wedge \zeta\right>$, $\left<(P
  \otimes\pi)_{\mathrm{tr}}\right> \cong \left<(P\otimes \pi)_{\mathrm{tr}} \wedge
  \zeta\right>$ and $\left<P \otimes \zeta\right> \cong
\left<H\otimes \pi \wedge \zeta\right>$ contribute to
$B_{\mathrm{sym}}$.  Therefore these modules cannot be present in the
$\eB_i$. There are precisely eleven bargmannian structures not containing
any of these $\so(n-1)$-submodules:  $\eB_0$, $\eB_1$, $\eB_2$,
$\eB_3$, $\eB_4$, $\eB_5$, $\eB_8$, $\eB_9$, $\eB_{13}$, $\eB_{14}$
and $\eB_{15}$, which are depicted in
Figure~\ref{fig:bar-hasse-tot-geod}.

They can be distinguished by the properties in
Table~\ref{tab:bar-tot-geo}.  Notice that $\nabla^g\xi$ is a one-form
on $M$ with values in the distribution $\xi^\perp$.  We can restrict
$\nabla^g\xi$ to $\xi^\perp$ (resulting in the column
$\left.\nabla^g\xi\right|_{\xi^\perp}$ of the table).  Alternatively we can
quotient by (the line bundle associated to) $\xi$ in order to define
$\overline{\nabla^g\xi}$, which is a one-form on $M$ with values in
the quotient vector bundle $E=\xi^\perp/L$.  Structures $\eB_9$ and $\eB_{13}$
can be distinguished by the fact that for $\eB_9$, $T^\nabla(\xi,X) =
0$ for all $X \perp \xi$, whereas this is not the case for
$\eB_{13}$; but it would be better to distinguish them in a way which
is manifestly independent of $\nabla$.

\begin{table}[h!]
  \centering
  \rowcolors{2}{blue!10}{white}
  \begin{tabular}{>{$}c<{$}|*{8}{>{$}c<{$}}}\toprule
    \multicolumn{1}{c|}{Structure} & \left.\nabla^g\xi\right|_{\xi^\perp} &  \overline{\nabla^g\xi} & \eL_\xi g & d\xi^\flat & \xi^\flat \wedge d\xi^\flat & \nabla^g_\xi \xi & \eL_\xi \xi^\flat & \eL_\xi \nu\\\midrule
    \eB_0 & \yes & \yes & \yes & \yes & \yes & \yes & \yes & \yes \\
    \eB_1 & \yes & \yes & & \yes & \yes & \yes & \yes & \yes \\
    \eB_2 & & & \yes & & \yes & & \yes & \yes \\
    \eB_3 & & \yes & & & \yes & & \yes & \yes \\
    \eB_4 & & & & \yes & \yes & \yes & \yes & \yes \\
    \eB_5 & & & \yes & & & \yes & \yes & \yes \\
    \eB_8 & & & & & \yes & & \yes & \yes \\
    \eB_9 & & & & & & \yes & \yes & \yes\\
    \eB_{13} & & & & & & \yes & \yes & \yes \\
    \eB_{14} & & & & & \yes & & & \\
    \eB_{15} & & & & & & & & \\\bottomrule
  \end{tabular}
  \vspace{1em}
  \caption{Totally geodesic bargmannian structures\\
    ($\yes$ means the expression vanishes)}
  \label{tab:bar-tot-geo}
\end{table}

\begin{figure}[h!]
  \centering
  \begin{tikzpicture}[scale=1.3]
    % 
    % nodes
    % 
    \coordinate[label=right:{\scriptsize $0$}] (b0) at (-1,7); 
    \coordinate[label=left:{\scriptsize $1$}] (b1) at (-1,6); 
    \coordinate[label=left:{\scriptsize $2$}] (b2) at (-2,6); 
    \coordinate[label=right:{\scriptsize $3$}] (b3) at (-1,5); 
    \coordinate[label=above right:{\scriptsize $4$}] (b4) at (0,5); 
    \coordinate[label=left:{\scriptsize $5$}] (b5) at (-2,5); 
    \coordinate[label=above left:{\scriptsize $8$}] (b8) at (-1,4); 
    \coordinate[label=left:{\scriptsize $9$}] (b9) at (-2,4); 
    \coordinate[label=above left:{\scriptsize $13$}] (b13) at (-2,3); 
    \coordinate[label=above left:{\scriptsize $14$}] (b14) at (-1,3); 
    \coordinate[label=left:{\scriptsize $15$}] (b15) at (-3,2); 
    % 
    % edges
    %
    \tikzset{every path/.style={->,thick,gris,shorten >=3pt, shorten <=3pt}}
    \draw (b0) -- (b1);
    \draw (b0) -- (b2);
    \draw (b2) -- (b3);
    \draw (b2) -- (b5);
    \draw (b1) -- (b3);
    \draw (b1) -- (b4);
    \draw (b3) -- (b8);
    \draw (b3) -- (b9);
    \draw (b4) -- (b8);
    \draw (b5) -- (b9);
    \draw (b9) -- (b13); 
    \draw (b8) -- (b13); 
    \draw (b8) -- (b14); 
    \draw (b13) -- (b15);
    \draw (b14) -- (b15);
    % 
    % points
    % 
    \foreach \point in {b0,b1,b2,b3,b4,b5,b8,b9,b13,b14,b15}
    \filldraw [color=blue!70!black,fill=blue!70!white] (\point) circle (1.5pt);
    \begin{scope}[opacity=0.15]
      \coordinate[label=left:{\scriptsize $6$}] (b6) at (0,4); 
      \coordinate[label=above right:{\scriptsize $7$}] (b7) at (1,4); 
      \coordinate[label=above right:{\scriptsize $10$}] (b10) at (2,3); 
      \coordinate[label=above right:{\scriptsize $11$}] (b11) at (0,3); 
      \coordinate[label=above right:{\scriptsize $12$}] (b12) at (1,3); 
      \coordinate[label=above left:{\scriptsize $16$}] (b16) at (-1,2); 
      \coordinate[label=right:{\scriptsize $17$}] (b17) at (1,2); 
      \coordinate[label=left:{\scriptsize $18$}] (b18) at (-2,2); 
      \coordinate[label=below right:{\scriptsize $19$}] (b19) at (0,2); 
      \coordinate[label=right:{\scriptsize $20$}] (b20) at (2,2); 
      \coordinate[label=below left:{\scriptsize $21$}] (b21) at (-2,1); 
      \coordinate[label=below left:{\scriptsize $22$}] (b22) at (-1,1); 
      \coordinate[label=below right:{\scriptsize $23$}] (b23) at (1,1); 
      \coordinate[label=below right:{\scriptsize $24$}] (b24) at (0,1); 
      \coordinate[label=below right:{\scriptsize $25$}] (b25) at (0,0); 
      \coordinate[label=right:{\scriptsize $26$}] (b26) at (0,-1);
      % 
      % edges
      %
      \tikzset{every path/.style={->,thick,black,shorten >=3pt, shorten <=3pt}}
      \draw (b4) -- (b6);
      \draw (b4) -- (b7);
      \draw (b8) -- (b11); 
      \draw (b8) -- (b12); 
      \draw (b6) -- (b10); 
      \draw (b6) -- (b11); 
      \draw (b7) -- (b10); 
      \draw (b7) -- (b12); 
      \draw (b13) -- (b18);
      \draw (b13) -- (b19);
      \draw (b14) -- (b16);
      \draw (b14) -- (b17);
      \draw (b11) -- (b16);
      \draw (b11) -- (b18);
      \draw (b11) -- (b20);
      \draw (b12) -- (b17);
      \draw (b12) -- (b19);
      \draw (b12) -- (b20);
      \draw (b10) -- (b20);
      \draw (b15) -- (b21);
      \draw (b15) -- (b22);
      \draw (b18) -- (b21);
      \draw (b18) -- (b24);
      \draw (b16) -- (b21);
      \draw (b16) -- (b23);
      \draw (b19) -- (b22);
      \draw (b19) -- (b24);
      \draw (b17) -- (b22);
      \draw (b17) -- (b23);
      \draw (b20) -- (b23);
      \draw (b20) -- (b24);
      \draw (b21) -- (b25); 
      \draw (b22) -- (b25); 
      \draw (b23) -- (b25); 
      \draw (b24) -- (b25); 
      \draw (b25) -- (b26); 
      % 
      % points
      % 
      \foreach \point in {b7,b12,b17,b19,b22,b6,b11,b16,b18,b21,b10,b20,b23,b24,b25,b26}
      \filldraw [color=white!70!black,fill=black!70!white] (\point) circle (1.5pt);
    \end{scope}
  \end{tikzpicture}
  \caption{Hasse diagram of totally geodesic bargmannian structures}
  \label{fig:bar-hasse-tot-geod}
\end{figure}
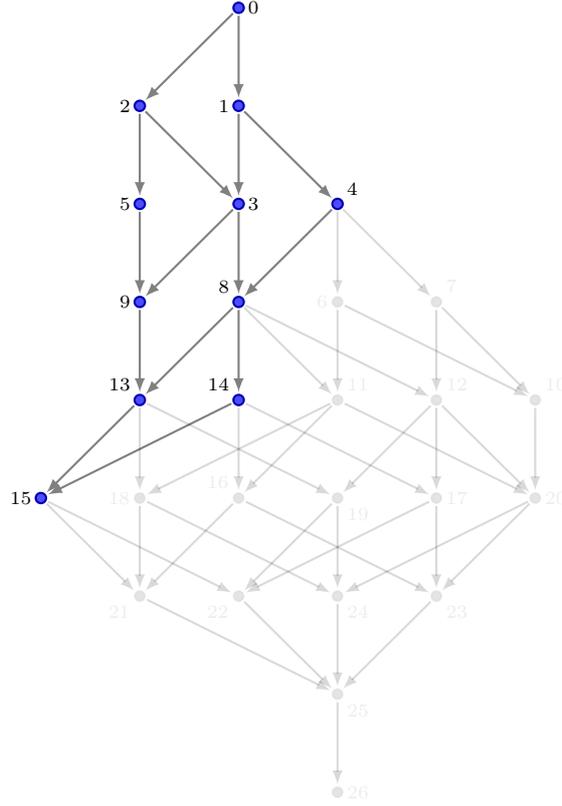

\subsubsection{Minimal bargmannian structures}
\label{sec:bar-min}

We continue with those bargmannian structures which are minimal.  From
Table~\ref{tab:contributions} it follows that such 
$\eB_i$ cannot contain the $\so(n-1)$-modules $\left<(P
  \otimes\pi)_{\mathrm{tr}}\right> \cong \left<(P\otimes \pi)_{\mathrm{tr}} \wedge
  \zeta\right>$ nor $\left<P \otimes \zeta\right> \cong
\left<H\otimes \pi \wedge \zeta\right>$ and must contain $\left<(P
  \otimes\pi)_{\odot_0^2}\right> \cong \left<(P\otimes \pi)_{\odot_0^2} \wedge
  \zeta\right>$.  There are precisely five bargmannian
structures satisfying these conditions: $\eB_6$, $\eB_{11}$,
$\eB_{16}$, $\eB_{18}$ and $\eB_{21}$, which are depicted in
Figure~\ref{fig:bar-hasse-min}.  They can be distinguished by the
properties listed in Table~\ref{tab:bar-min}.

\begin{table}[h!]
  \centering
  \rowcolors{2}{blue!10}{white}
  \begin{tabular}{>{$}c<{$}|*{4}{>{$}c<{$}}}\toprule
    \multicolumn{1}{c|}{Structure} & d\xi^\flat & \xi^\flat \wedge  d\xi^\flat & \nabla^g_\xi\xi & \eL_\xi \xi^\flat \\\midrule
    \eB_6 & \yes & \yes & \yes & \yes \\
    \eB_{11} & & \yes & \yes & \yes \\
    \eB_{16} & & \yes  & & \\
    \eB_{18} & & & \yes  & \yes \\
    \eB_{21} & & & & \yes\\\bottomrule
  \end{tabular}
  \vspace{1em}
  \caption{Minimal bargmannian structures\\
  ($\yes$ means the expression vanishes)}
  \label{tab:bar-min}
\end{table}

\begin{figure}[h!]
  \centering
  \begin{tikzpicture}[scale=1.3]
    %
    % nodes
    %
    \coordinate[label=left:{\scriptsize $6$}] (b6) at (0,4); 
    \coordinate[label=above right:{\scriptsize $11$}] (b11) at (0,3); 
    \coordinate[label=above left:{\scriptsize $16$}] (b16) at (-1,2); 
    \coordinate[label=left:{\scriptsize $18$}] (b18) at (-2,2); 
    \coordinate[label=below left:{\scriptsize $21$}] (b21) at (-2,1); 
    %
    % edges
    %
    \tikzset{every path/.style={->,thick,gris,shorten >=3pt, shorten <=3pt}}
    \draw (b6) -- (b11); 
    \draw (b11) -- (b16);
    \draw (b11) -- (b18);
    \draw (b18) -- (b21);
    \draw (b16) -- (b21);
    %
    % points
    %
    \foreach \point in {b6,b11,b16,b18,b21}
    \filldraw [color=green!70!black,fill=green!70!white] (\point) circle (1.5pt);
    \begin{scope}[opacity=0.15]
      \coordinate[label=right:{\scriptsize $0$}] (b0) at (-1,7); 
      \coordinate[label=left:{\scriptsize $1$}] (b1) at (-1,6); 
      \coordinate[label=left:{\scriptsize $2$}] (b2) at (-2,6); 
      \coordinate[label=right:{\scriptsize $3$}] (b3) at (-1,5); 
      \coordinate[label=above right:{\scriptsize $4$}] (b4) at (0,5); 
      \coordinate[label=left:{\scriptsize $5$}] (b5) at (-2,5); 
      \coordinate[label=above right:{\scriptsize $7$}] (b7) at (1,4); 
      \coordinate[label=above left:{\scriptsize $8$}] (b8) at (-1,4); 
      \coordinate[label=left:{\scriptsize $9$}] (b9) at (-2,4); 
      \coordinate[label=above right:{\scriptsize $10$}] (b10) at (2,3); 
      \coordinate[label=above right:{\scriptsize $12$}] (b12) at (1,3); 
      \coordinate[label=above left:{\scriptsize $13$}] (b13) at (-2,3); 
      \coordinate[label=above left:{\scriptsize $14$}] (b14) at (-1,3); 
      \coordinate[label=left:{\scriptsize $15$}] (b15) at (-3,2); 
      \coordinate[label=right:{\scriptsize $17$}] (b17) at (1,2); 
      \coordinate[label=below right:{\scriptsize $19$}] (b19) at (0,2); 
      \coordinate[label=right:{\scriptsize $20$}] (b20) at (2,2); 
      \coordinate[label=below left:{\scriptsize $22$}] (b22) at (-1,1); 
      \coordinate[label=below right:{\scriptsize $23$}] (b23) at (1,1); 
      \coordinate[label=below right:{\scriptsize $24$}] (b24) at (0,1); 
      \coordinate[label=below right:{\scriptsize $25$}] (b25) at (0,0); 
      \coordinate[label=right:{\scriptsize $26$}] (b26) at (0,-1);
      % 
      % edges
      %
      \tikzset{every path/.style={->,thick,black,shorten >=3pt, shorten <=3pt}}
      \draw (b0) -- (b1);
      \draw (b0) -- (b2);
      \draw (b2) -- (b3);
      \draw (b2) -- (b5);
      \draw (b1) -- (b3);
      \draw (b1) -- (b4);
      \draw (b3) -- (b8);
      \draw (b3) -- (b9);
      \draw (b4) -- (b6);
      \draw (b4) -- (b7);
      \draw (b4) -- (b8);
      \draw (b5) -- (b9);
      \draw (b9) -- (b13); 
      \draw (b8) -- (b11); 
      \draw (b8) -- (b12); 
      \draw (b8) -- (b13); 
      \draw (b8) -- (b14); 
      \draw (b6) -- (b10); 
      \draw (b7) -- (b10); 
      \draw (b7) -- (b12); 
      \draw (b13) -- (b15);
      \draw (b13) -- (b18);
      \draw (b13) -- (b19);
      \draw (b14) -- (b15);
      \draw (b14) -- (b16);
      \draw (b14) -- (b17);
      \draw (b11) -- (b20);
      \draw (b12) -- (b17);
      \draw (b12) -- (b19);
      \draw (b12) -- (b20);
      \draw (b10) -- (b20);
      \draw (b15) -- (b21);
      \draw (b15) -- (b22);
      \draw (b18) -- (b24);
      \draw (b16) -- (b23);
      \draw (b19) -- (b22);
      \draw (b19) -- (b24);
      \draw (b17) -- (b22);
      \draw (b17) -- (b23);
      \draw (b20) -- (b23);
      \draw (b20) -- (b24);
      \draw (b21) -- (b25); 
      \draw (b22) -- (b25); 
      \draw (b23) -- (b25); 
      \draw (b24) -- (b25); 
      \draw (b25) -- (b26); 
      % 
      % points
      % 
      \foreach \point in {b0,b1,b2,b3,b4,b5,b8,b9,b13,b14,b15,b7,b12,b17,b19,b22,b10,b20,b23,b24,b25,b26}
      \filldraw [color=white!70!black,fill=black!70!white] (\point) circle (1.5pt);
    \end{scope}

  \end{tikzpicture}
  \caption{Hasse diagram of minimal bargmannian structures}
  \label{fig:bar-hasse-min}
\end{figure}
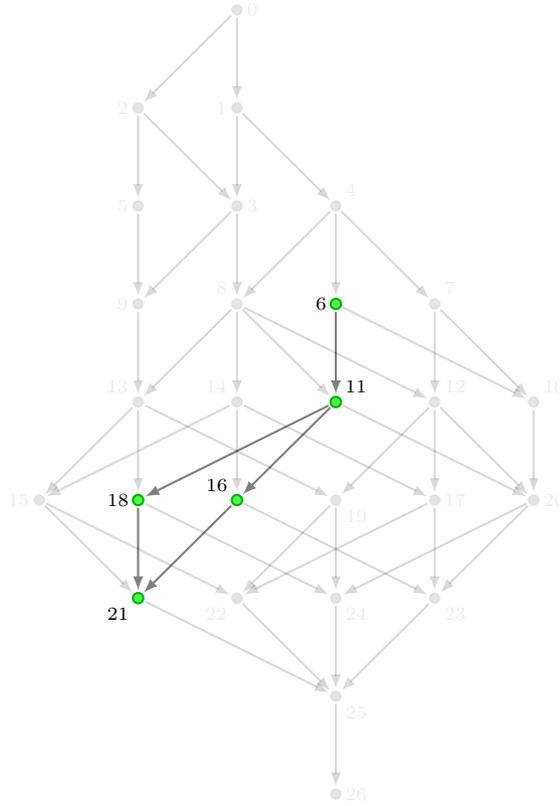

\subsubsection{Totally umbilical bargmannian structures}
\label{sec:bar-tot-umb}

We continue with those bargmannian structures which are totally
umbilical.  From Table~\ref{tab:contributions} it follows that such
$\eB_i$ cannot contain the $\so(n-1)$-modules $\left<(P
  \otimes\pi)_{\odot_0^2}\right> \cong \left<(P\otimes \pi)_{\odot_0^2} \wedge
  \zeta\right>$ nor $\left<P \otimes \zeta\right> \cong
\left<H\otimes \pi \wedge \zeta\right>$ and must contain $\left<(P
  \otimes\pi)_{\mathrm{tr}}\right> \cong \left<(P\otimes \pi)_{\mathrm{tr}} \wedge
  \zeta\right>$.  There are precisely five bargmannian
structures satisfying these conditions: $\eB_7$, $\eB_{12}$,
$\eB_{17}$, $\eB_{19}$ and $\eB_{22}$, which are depicted in
Figure~\ref{fig:bar-hasse-tot-umb}.  They can be distinguished by the
properties listed in Table~\ref{tab:bar-tot-umb}.

\begin{table}[h!]
  \centering
  \rowcolors{2}{blue!10}{white}
  \begin{tabular}{>{$}c<{$}|*{4}{>{$}c<{$}}}\toprule
    \multicolumn{1}{c|}{Structure} & d\xi^\flat & \xi^\flat \wedge  d\xi^\flat& \nabla^g_\xi\xi & \eL_\xi \xi^\flat \\\midrule
    \eB_7 & \yes & \yes & \yes & \yes \\
    \eB_{12} & & \yes & \yes & \yes \\
    \eB_{17} & & \yes  & & \\
    \eB_{19} & & & \yes  & \yes \\
    \eB_{22} & & & & \yes\\\bottomrule
  \end{tabular}
  \vspace{1em}
  \caption{Totally umbilical bargmannian structures\\
    ($\yes$ means the expression vanishes)}
  \label{tab:bar-tot-umb}
\end{table}
    
\begin{figure}[h!]
  \centering
  \begin{tikzpicture}[scale=1.3]
    %
    % nodes
    %
    \coordinate[label=above right:{\scriptsize $7$}] (b7) at (1,4); 
    \coordinate[label=above right:{\scriptsize $12$}] (b12) at (1,3); 
    \coordinate[label=right:{\scriptsize $17$}] (b17) at (1,2); 
    \coordinate[label=below right:{\scriptsize $19$}] (b19) at (0,2); 
    \coordinate[label=below left:{\scriptsize $22$}] (b22) at (-1,1); 
    %
    % edges
    %
    \tikzset{every path/.style={->,thick,gris,shorten >=3pt, shorten <=3pt}}
    \draw (b7) -- (b12); 
    \draw (b12) -- (b17);
    \draw (b12) -- (b19);
    \draw (b17) -- (b22);
    \draw (b19) -- (b22);
    %
    % points
    %
    \foreach \point in {b7,b12,b17,b19,b22}
    \filldraw [color=red!70!black,fill=red!70!white] (\point) circle (1.5pt);
    \begin{scope}[opacity=0.15]
      \coordinate[label=right:{\scriptsize $0$}] (b0) at (-1,7); 
      \coordinate[label=left:{\scriptsize $1$}] (b1) at (-1,6); 
      \coordinate[label=left:{\scriptsize $2$}] (b2) at (-2,6); 
      \coordinate[label=right:{\scriptsize $3$}] (b3) at (-1,5); 
      \coordinate[label=above right:{\scriptsize $4$}] (b4) at (0,5); 
      \coordinate[label=left:{\scriptsize $5$}] (b5) at (-2,5); 
      \coordinate[label=left:{\scriptsize $6$}] (b6) at (0,4); 
      \coordinate[label=above left:{\scriptsize $8$}] (b8) at (-1,4); 
      \coordinate[label=left:{\scriptsize $9$}] (b9) at (-2,4); 
      \coordinate[label=above right:{\scriptsize $10$}] (b10) at (2,3); 
      \coordinate[label=above right:{\scriptsize $11$}] (b11) at (0,3); 
      \coordinate[label=above left:{\scriptsize $13$}] (b13) at (-2,3); 
      \coordinate[label=above left:{\scriptsize $14$}] (b14) at (-1,3); 
      \coordinate[label=left:{\scriptsize $15$}] (b15) at (-3,2); 
      \coordinate[label=above left:{\scriptsize $16$}] (b16) at (-1,2); 
      \coordinate[label=left:{\scriptsize $18$}] (b18) at (-2,2); 
      \coordinate[label=right:{\scriptsize $20$}] (b20) at (2,2); 
      \coordinate[label=below left:{\scriptsize $21$}] (b21) at (-2,1); 
      \coordinate[label=below right:{\scriptsize $23$}] (b23) at (1,1); 
      \coordinate[label=below right:{\scriptsize $24$}] (b24) at (0,1); 
      \coordinate[label=below right:{\scriptsize $25$}] (b25) at (0,0); 
      \coordinate[label=right:{\scriptsize $26$}] (b26) at (0,-1);
      % 
      % edges
      %
      \tikzset{every path/.style={->,thick,black,shorten >=3pt, shorten <=3pt}}
      \draw (b0) -- (b1);
      \draw (b0) -- (b2);
      \draw (b2) -- (b3);
      \draw (b2) -- (b5);
      \draw (b1) -- (b3);
      \draw (b1) -- (b4);
      \draw (b3) -- (b8);
      \draw (b3) -- (b9);
      \draw (b4) -- (b6);
      \draw (b4) -- (b7);
      \draw (b4) -- (b8);
      \draw (b5) -- (b9);
      \draw (b9) -- (b13); 
      \draw (b8) -- (b11); 
      \draw (b8) -- (b12); 
      \draw (b8) -- (b13); 
      \draw (b8) -- (b14); 
      \draw (b6) -- (b10); 
      \draw (b6) -- (b11); 
      \draw (b7) -- (b10); 
      \draw (b13) -- (b15);
      \draw (b13) -- (b18);
      \draw (b13) -- (b19);
      \draw (b14) -- (b15);
      \draw (b14) -- (b16);
      \draw (b14) -- (b17);
      \draw (b11) -- (b16);
      \draw (b11) -- (b18);
      \draw (b11) -- (b20);
      \draw (b12) -- (b20);
      \draw (b10) -- (b20);
      \draw (b15) -- (b21);
      \draw (b15) -- (b22);
      \draw (b18) -- (b21);
      \draw (b18) -- (b24);
      \draw (b16) -- (b21);
      \draw (b16) -- (b23);
      \draw (b19) -- (b24);
      \draw (b17) -- (b23);
      \draw (b20) -- (b23);
      \draw (b20) -- (b24);
      \draw (b21) -- (b25); 
      \draw (b22) -- (b25); 
      \draw (b23) -- (b25); 
      \draw (b24) -- (b25); 
      \draw (b25) -- (b26); 
      % 
      % points
      % 
      \foreach \point in {b0,b1,b2,b3,b4,b5,b8,b9,b13,b14,b15,b6,b11,b16,b18,b21,b10,b20,b23,b24,b25,b26}
      \filldraw [color=white!70!black,fill=black!70!white] (\point) circle (1.5pt);
    \end{scope}
  \end{tikzpicture}
  \caption{Hasse diagram of totally umbilical bargmannian structures}
  \label{fig:bar-hasse-tot-umb}
\end{figure}
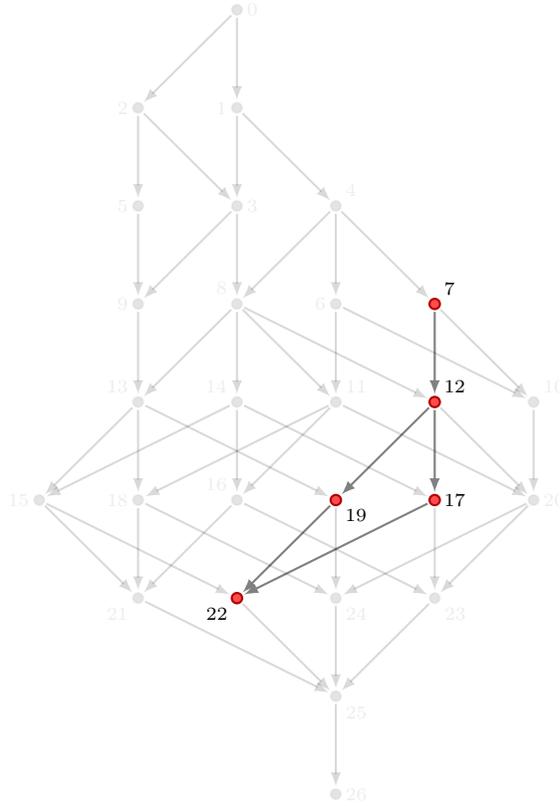

\subsubsection{Other bargmannian structures}
\label{sec:bar-other}

We end with those bargmannian structures which are neither totally
geodesic, totally umbilical nor minimal, but still not generic.  From
Table~\ref{tab:contributions} it follows that such 
$\eB_i$ cannot contain the $\so(n-1)$-module $\left<P \otimes \zeta\right> \cong
\left<H\otimes \pi \wedge \zeta\right>$ and must contain both $\left<(P
  \otimes\pi)_{\odot_0^2}\right> \cong \left<(P\otimes \pi)_{\odot_0^2} \wedge
  \zeta\right>$ and $\left<(P
  \otimes\pi)_{\mathrm{tr}}\right> \cong \left<(P\otimes \pi)_{\mathrm{tr}} \wedge
  \zeta\right>$.  There are precisely five bargmannian
structures satisfying these conditions: $\eB_{10}$, $\eB_{20}$,
$\eB_{23}$, $\eB_{24}$ and $\eB_{25}$, which are depicted in
Figure~\ref{fig:bar-hasse-other}.  They can be distinguished by the
properties listed in Table~\ref{tab:bar-other}.

\begin{table}[h!]
  \centering
  \rowcolors{2}{blue!10}{white}
  \begin{tabular}{>{$}c<{$}|*{4}{>{$}c<{$}}}\toprule
    \multicolumn{1}{c|}{Structure} & d\xi^\flat & \xi^\flat \wedge  d\xi^\flat & \nabla^g_\xi\xi & \eL_\xi \xi^\flat \\\midrule
    \eB_{10} & \yes & \yes & \yes & \yes \\
    \eB_{20} & & \yes & \yes & \yes \\
    \eB_{23} & & \yes  & & \\
    \eB_{24} & & & \yes  & \yes \\
    \eB_{25} & & & & \yes\\\bottomrule
  \end{tabular}
  \vspace{1em}
  \caption{Other bargmannian structures\\
    ($\yes$ means the expression vanishes)}
  \label{tab:bar-other}
\end{table}

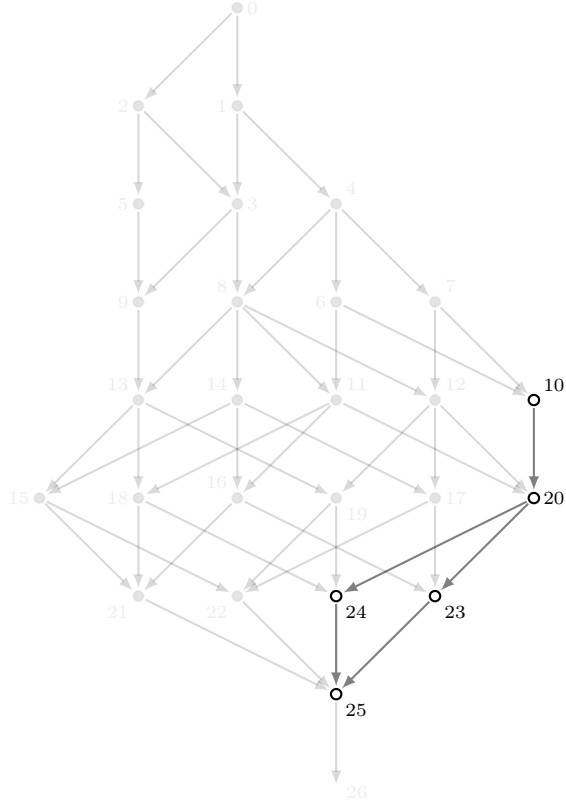
\begin{figure}[h!]
  \centering
  \begin{tikzpicture}[scale=1.3]
    %
    % nodes
    %
    \coordinate[label=above right:{\scriptsize $10$}] (b10) at (2,3); 
    \coordinate[label=right:{\scriptsize $20$}] (b20) at (2,2); 
    \coordinate[label=below right:{\scriptsize $23$}] (b23) at (1,1); 
    \coordinate[label=below right:{\scriptsize $24$}] (b24) at (0,1); 
    \coordinate[label=below right:{\scriptsize $25$}] (b25) at (0,0); 
    %
    % edges
    %
    \tikzset{every path/.style={->,thick,gris,shorten >=3pt, shorten <=3pt}}
    \draw (b10) -- (b20);
    \draw (b20) -- (b23);
    \draw (b20) -- (b24);
    \draw (b23) -- (b25); 
    \draw (b24) -- (b25); 
    %
    % points
    %
    \foreach \point in {b10,b20,b23,b24,b25}
    \filldraw [color=black,fill=white] (\point) circle (1.5pt);
    \begin{scope}[opacity=0.15]
      \coordinate[label=right:{\scriptsize $0$}] (b0) at (-1,7); 
      \coordinate[label=left:{\scriptsize $1$}] (b1) at (-1,6); 
      \coordinate[label=left:{\scriptsize $2$}] (b2) at (-2,6); 
      \coordinate[label=right:{\scriptsize $3$}] (b3) at (-1,5); 
      \coordinate[label=above right:{\scriptsize $4$}] (b4) at (0,5); 
      \coordinate[label=left:{\scriptsize $5$}] (b5) at (-2,5); 
      \coordinate[label=left:{\scriptsize $6$}] (b6) at (0,4); 
      \coordinate[label=above right:{\scriptsize $7$}] (b7) at (1,4); 
      \coordinate[label=above left:{\scriptsize $8$}] (b8) at (-1,4); 
      \coordinate[label=left:{\scriptsize $9$}] (b9) at (-2,4); 
      \coordinate[label=above right:{\scriptsize $11$}] (b11) at (0,3); 
      \coordinate[label=above right:{\scriptsize $12$}] (b12) at (1,3); 
      \coordinate[label=above left:{\scriptsize $13$}] (b13) at (-2,3); 
      \coordinate[label=above left:{\scriptsize $14$}] (b14) at (-1,3); 
      \coordinate[label=left:{\scriptsize $15$}] (b15) at (-3,2); 
      \coordinate[label=above left:{\scriptsize $16$}] (b16) at (-1,2); 
      \coordinate[label=right:{\scriptsize $17$}] (b17) at (1,2); 
      \coordinate[label=left:{\scriptsize $18$}] (b18) at (-2,2); 
      \coordinate[label=below right:{\scriptsize $19$}] (b19) at (0,2); 
      \coordinate[label=below left:{\scriptsize $21$}] (b21) at (-2,1); 
      \coordinate[label=below left:{\scriptsize $22$}] (b22) at (-1,1); 
      \coordinate[label=right:{\scriptsize $26$}] (b26) at (0,-1);
      % 
      % edges
      % 
      \tikzset{every path/.style={->,thick,black,shorten >=3pt, shorten <=3pt}}
      \draw (b0) -- (b1);
      \draw (b0) -- (b2);
      \draw (b2) -- (b3);
      \draw (b2) -- (b5);
      \draw (b1) -- (b3);
      \draw (b1) -- (b4);
      \draw (b3) -- (b8);
      \draw (b3) -- (b9);
      \draw (b4) -- (b6);
      \draw (b4) -- (b7);
      \draw (b4) -- (b8);
      \draw (b5) -- (b9);
      \draw (b9) -- (b13); 
      \draw (b8) -- (b11); 
      \draw (b8) -- (b12); 
      \draw (b8) -- (b13); 
      \draw (b8) -- (b14); 
      \draw (b6) -- (b10); 
      \draw (b6) -- (b11); 
      \draw (b7) -- (b10); 
      \draw (b7) -- (b12); 
      \draw (b13) -- (b15);
      \draw (b13) -- (b18);
      \draw (b13) -- (b19);
      \draw (b14) -- (b15);
      \draw (b14) -- (b16);
      \draw (b14) -- (b17);
      \draw (b11) -- (b16);
      \draw (b11) -- (b18);
      \draw (b11) -- (b20);
      \draw (b12) -- (b17);
      \draw (b12) -- (b19);
      \draw (b12) -- (b20);
      \draw (b15) -- (b21);
      \draw (b15) -- (b22);
      \draw (b18) -- (b21);
      \draw (b18) -- (b24);
      \draw (b16) -- (b21);
      \draw (b16) -- (b23);
      \draw (b19) -- (b22);
      \draw (b19) -- (b24);
      \draw (b17) -- (b22);
      \draw (b17) -- (b23);
      \draw (b21) -- (b25); 
      \draw (b22) -- (b25); 
      \draw (b25) -- (b26); 
      % 
      % points
      % 
      \foreach \point in {b0,b1,b2,b3,b4,b5,b8,b9,b13,b14,b15,b7,b12,b17,b19,b22,b6,b11,b16,b18,b21}
      \filldraw [color=white!70!black,fill=black!70!white] (\point) circle (1.5pt);
    \end{scope}
  \end{tikzpicture}
  \caption{Hasse diagram of other bargmannian structures}
  \label{fig:bar-hasse-other}
\end{figure}

We may summarise the preceding discussion as follows.

\begin{theorem}\label{thm:bar}
  Let\footnote{See Appendices~\ref{sec:three-dimens-bargm} for $n=2$
    and \ref{sec:six-dimens-bargm} for $n=5$.} $n>2$ and $n\neq 5$.  A
  bargmannian $G$-structure on an ($n+1$)-dimensional manifold
  $(M,g,\xi)$ can be of twenty-seven different classes depending on
  its intrinsic torsion.  These classes are summarised in
  Table~\ref{tab:bar-summary}, where each class is labelled by the
  smallest $G$-submodule of $\coker \d$ containing the intrinsic
  torsion and is characterised geometrically as indicated in the
  table.
\end{theorem}

\begin{table}[h!]
  \centering
  \rowcolors{2}{blue!10}{white}
    \begin{tabular}{>{$}c<{$}|c*{8}{>{$}c<{$}}}\toprule
      \multicolumn{1}{c|}{Structure} & \multicolumn{1}{c}{Type of $\xi^\perp$} & \multicolumn{8}{c}{Geometrical characterisation}\\\midrule
      \eB_0 & totally geodesic & \multicolumn{8}{c}{pp-wave ($\nabla^g\xi = 0$)} \\
      \eB_1 & totally geodesic & \multicolumn{8}{l}{$\left.\nabla^g\xi\right|_{\xi^\perp} = 0$} \\
      \eB_2 & totally geodesic & \multicolumn{4}{l}{$\eL_\xi g = 0$} & \multicolumn{4}{l}{$\xi^\flat \wedge d\xi^\flat = 0$}\\
      \eB_3 & totally geodesic & \multicolumn{4}{l}{$\overline{\nabla^g\xi} = 0$} & \multicolumn{4}{l}{$\xi^\flat \wedge d\xi^\flat = 0$}\\
      \eB_4 & totally geodesic & \multicolumn{8}{l}{$d\xi^\flat = 0$ ($\left.\nabla^g\xi\right|_{\xi^\perp} \neq 0$)} \\
      \eB_5 & totally geodesic & \multicolumn{4}{l}{$\eL_\xi g =0$} & \multicolumn{4}{l}{$\xi^\flat \wedge d\xi^\flat \neq 0$} \\
      \eB_6 & minimal & \multicolumn{8}{l}{$d\xi^\flat = 0$}\\
      \eB_7 & totally umbilical & \multicolumn{8}{l}{$d\xi^\flat = 0$}\\
      \eB_8 & totally geodesic & \multicolumn{4}{l}{$\xi^\flat \wedge d\xi^\flat =0$} & \multicolumn{4}{l}{$\eL_\xi \xi^\flat = 0$} \\
      \eB_9 & totally geodesic & \multicolumn{2}{l}{$\xi^\flat \wedge d\xi^\flat \neq 0$} & \multicolumn{2}{l}{$\nabla^g_\xi \xi =0$} & \multicolumn{2}{l}{$\eL_\xi \xi^\flat = 0$} & \multicolumn{2}{l}{$\left.S\right|_{\xi^\perp}=0$} \\
      \eB_{10} & other & \multicolumn{8}{l}{$d\xi^\flat = 0$}\\
      \eB_{11} & minimal & \multicolumn{4}{l}{$\xi^\flat \wedge d\xi^\flat = 0$} & \multicolumn{4}{l}{$\nabla^g_\xi \xi = 0$}\\
      \eB_{12} & totally umbilical & \multicolumn{4}{l}{$\xi^\flat \wedge d\xi^\flat = 0$} & \multicolumn{4}{l}{$\nabla^g_\xi \xi = 0$}\\
      \eB_{13} & totally geodesic & \multicolumn{2}{l}{$\xi^\flat \wedge d\xi^\flat \neq 0$} & \multicolumn{2}{l}{$\nabla^g_\xi \xi =0$} & \multicolumn{2}{l}{$\eL_\xi \xi^\flat = 0$} & \multicolumn{2}{l}{$\left.S\right|_{\xi^\perp}\neq0$} \\
      \eB_{14} & totally geodesic & \multicolumn{2}{l}{$\xi^\flat \wedge d\xi^\flat = 0$} & \multicolumn{2}{l}{$\nabla^g_\xi \xi \neq 0$} & \multicolumn{2}{l}{$\eL_\xi \xi^\flat \neq 0$} & & \\
      \eB_{15} & totally geodesic & \multicolumn{2}{l}{$\xi^\flat \wedge d\xi^\flat \neq 0$} & \multicolumn{2}{l}{$\nabla^g_\xi \xi \neq 0$} & \multicolumn{2}{l}{$\eL_\xi \xi^\flat \neq 0$} & & \\
      \eB_{16} & minimal & \multicolumn{4}{l}{$\xi^\flat \wedge d\xi^\flat = 0$} & \multicolumn{4}{l}{$\eL_\xi \xi^\flat \neq 0$}\\
      \eB_{17} & totally umbilical & \multicolumn{4}{l}{$\xi^\flat \wedge d\xi^\flat = 0$} & \multicolumn{4}{l}{$\eL_\xi \xi^\flat \neq 0$}\\
      \eB_{18} & minimal & \multicolumn{4}{l}{$\xi^\flat \wedge d\xi^\flat \neq 0$} & \multicolumn{4}{l}{$\nabla^g_\xi \xi= 0$}\\
      \eB_{19} & totally umbilical & \multicolumn{4}{l}{$\xi^\flat \wedge d\xi^\flat \neq 0$} & \multicolumn{4}{l}{$\nabla^g_\xi \xi= 0$}\\
      \eB_{20} & other & \multicolumn{4}{l}{$\xi^\flat \wedge d\xi^\flat = 0$} & \multicolumn{4}{l}{$\nabla^g_\xi \xi = 0$}\\
      \eB_{21} & minimal & \multicolumn{4}{l}{$\xi^\flat \wedge d\xi^\flat \neq 0$} & \multicolumn{4}{l}{$\nabla^g_\xi \xi\neq 0$}\\
      \eB_{22} & totally umbilical & \multicolumn{4}{l}{$\xi^\flat \wedge d\xi^\flat \neq 0$} & \multicolumn{4}{l}{$\nabla^g_\xi \xi\neq 0$}\\
      \eB_{23} & other & \multicolumn{4}{l}{$\xi^\flat \wedge d\xi^\flat = 0$} & \multicolumn{4}{l}{$\eL_\xi \xi^\flat \neq 0$}\\
      \eB_{24} & other & \multicolumn{4}{l}{$\xi^\flat \wedge d\xi^\flat \neq 0$} & \multicolumn{4}{l}{$\nabla^g_\xi \xi= 0$}\\
      \eB_{25} & other & \multicolumn{4}{l}{$\xi^\flat \wedge d\xi^\flat \neq 0$} & \multicolumn{4}{l}{$\nabla^g_\xi \xi\neq 0$}\\
      \eB_{26} & \multicolumn{9}{c}{generic bargmannian structure} \\
      \bottomrule
    \end{tabular}
  \vspace{1em}
  \caption{Summary of bargmannian structures}
  \label{tab:bar-summary}
\end{table}

\subsection{Correspondences between bargmannian, galilean and carrollian
  structures}
\label{sec:corr-betw-bargm}

As pioneered in \cite{Duval:2014uoa}, bargmannian structures may be
related to galilean and carrollian structures and the interplay
between these structures can prove to be very useful.

\subsubsection{Bargmannian structures reducing to galilean structures}
\label{sec:bar-struc-gal}

These are the bargmannian structures where $\xi$ is a Killing vector:
$\eL_\xi g = 0$.  Let us assume for the purposes of exposition that
$\xi$ generates the action of a one-dimensional Lie group $\Gamma$ and
we can perform the null reduction of the bargmannian structure as in
\cite{PhysRevD.31.1841,Julia:1994bs}.

Indeed, we may view $M$ as the total space of a principal
$\Gamma$-bundle $\pi: M \to N$ over an $n$-dimensional manifold
$N = M/\Gamma$. The one-form $\xi^\flat$ is both horizontal (since
$\xi$ is null) and invariant (since $\xi$ is Killing). Then
$\xi^\flat = \pi^*\tau$ for a nowhere-vanishing one-form
$\tau \in \Omega^1(N)$. If $\alpha,\beta \in \Omega^1(N)$, then
$g((\pi^*\alpha)^\sharp, (\pi^*\beta)^\sharp)$, where
$\sharp: \Omega^1(M) \to \eX(M)$ is one of the musical isomorphisms
associated to $g$, is a $\Gamma$-invariant function on $M$ since so
are $g$, $\pi^*\alpha$ and $\pi^*\beta$. We can define
$\gamma \in \Gamma(\odot^2 TN)$ by
$\pi^*\gamma(\alpha,\beta) = g((\pi^*\alpha)^\sharp,
(\pi^*\beta)^\sharp)$. Notice that $\gamma(\tau,\alpha) = 0$ since
$(\pi^*\tau)^\sharp = \xi$ and hence for all
$\alpha \in \Omega^1(N)$,
\begin{equation}
  g((\pi^*\tau)^\sharp, (\pi^*\alpha)^\sharp) = g(\xi,(\pi^*\alpha)^\sharp) =
  (\pi^*\alpha)(\xi)= \pi^* (\alpha(\pi_*\xi) = 0.
\end{equation}
It follows that $(N, \tau, \gamma)$ is a galilean structure and we may
distinguish these bargmannian structures by which of the three galilean
structures they induce.

It turns out that there are precisely three bargmannian structures where
$\xi$ is Killing: $\eB_0$, $\eB_2$ and $\eB_5$, and they can be
distinguished by the galilean structure induced on their null
reductions.

\begin{itemize}
\item[($\eB_0$)]
  Here $\nabla^g\xi = 0$ and hence $g$ is a Brinkmann
  metric (i.e., a generalised pp-wave).  Since $\nabla^g\xi^\flat =
  0$, it follows that $d\xi^\flat = 0$ and hence the null reduction
  gives rise to a torsionless Newton--Cartan structure.

\item[($\eB_2$)]
  Here $d\xi^\flat \neq 0$ but $\xi^\flat \wedge d\xi^\flat = 0$, so that
  the null reduction gives a twistless torsional Newton--Cartan
  structure.

\item[($\eB_5$)]
  Here $\xi^\flat \wedge d\xi^\flat \neq 0$, so that the null
  reduction gives a torsional Newton--Cartan structure.
\end{itemize}

% As an application of this result we can exhibit the spatially
% isotropic homogeneous galilean spacetimes in
% \cite{Figueroa-OFarrill:2018ilb,Figueroa-OFarrill:2019sex} as null
% reductions of lorentzian spacetimes.  Since the galilean spacetimes in
% \cite{Figueroa-OFarrill:2018ilb,Figueroa-OFarrill:2019sex} are such
% that $d\tau = 0$, they arise as null reductions of bargmannian manifolds
% with vanishing intrinsic torsion; that is, generalised pp-waves.

% Relative to adapted coordinates $(u,v,x^a)$ where $\xi = \frac{\d}{\d
%   v}$, the most general Brinkmann metric \cite{Brinkmann2} is given by
% \begin{equation}
%   \label{eq:brinkmann}
%   2 du dv + a(u,x) du^2 + b_a(u,x) dx^a du + g_{ab}(u,x) dx^a dx^b.
% \end{equation}
% Here $\xi^\flat = du$

\subsubsection{Bargmannian structures with embedded carrollian  structures}
\label{sec:bar-struc-car}

As shown in \cite{Duval:2014uoa} (see also \cite{Hartong:2015xda}), a
null hypersurface in a lorentzian manifold admits a carrollian
structure.  A bargmannian manifold $(M,g,\xi)$ where
$\xi^\flat \wedge d\xi^\flat = 0$, is foliated by null hypersurfaces
and we can relate the carrollian structure on the null hypersurfaces
to the ambient bargmannian structure.

\begin{lemma}
  If $d\xi^\flat = 0$, the vector field $\xi$ is self-parallel
  relative to the Levi-Civita connection: $\nabla^g_\xi \xi = 0$,
  whereas if $d\xi^\flat \neq 0$ but $\xi^\flat \wedge d\xi^\flat =
  0$, then $\nabla^g_\xi\xi = f \xi$ for some nonzero function $f \in
  C^\infty(M)$.
\end{lemma}

\begin{proof}
  Let $\xi^\flat \wedge d\xi^\flat =0$.  Then by
  Proposition~\ref{prop:hso}, there exists some
  $\alpha \in \Omega^1(M)$ such that
  $d\xi^\flat = \alpha \wedge \xi^\flat$.  The one-form $\alpha$ is
  defined up to the addition of a one-form $f \xi^\flat$ for some
  $f \in C^\infty(M)$.  If $d\xi^\flat = 0$ we can choose
  $\alpha = 0$.

  For all $X, Y \in \eX(M)$, the equation $d\xi^\flat  = \xi^\flat
  \wedge \alpha$ becomes
  \begin{equation*}
    X g(\xi, Y) - Y g(\xi,X) - g(\xi,[X,Y]) = \alpha(X) g(\xi,Y) - g(\xi,X) \alpha(Y).
  \end{equation*}
  Putting $Y=\xi$, and using that $g(\xi,\xi) = 0$, we have that
  \begin{equation*}
    \xi g(\xi, X) + g(\xi, [X,\xi]) - g(\xi,X) \alpha(\xi) = 0.
  \end{equation*}
  We use that $\nabla^g$ is metric to expand the first term as
  \begin{equation*}
    \xi g(\xi, X)  =g(\nabla^g_\xi \xi, X) + g(\xi, \nabla^g_\xi X),
  \end{equation*}
  resulting in
  \begin{equation*}
    g(\nabla^g_\xi \xi - \alpha(\xi)\xi, X) + g(\xi, \nabla^g_\xi X +
    [X,\xi]) = 0.
  \end{equation*}
  Using that $\nabla^g$ has zero torsion, $\nabla^g\xi X + [X,\xi] =
  \nabla^g_X \xi$ and hence the second term becomes $g(\xi,
  \nabla^g_X\xi)$, which vanishes since this is half the derivative of
  $g(\xi,\xi)$ along $X$ and $\xi$ is null.  This leaves the first
  term: since $g$ is nondegenerate and $X \in \eX(M)$ is arbitrary, we
  conclude that $\nabla^g_\xi \xi = \alpha(\xi)\xi$.  It follows that
  if $d\xi^\flat = 0$ then $\nabla^g_\xi \xi = 0$, otherwise
  $f:=\alpha(\xi)$ is not identically zero and hence $\nabla^g_\xi \xi
  = f \xi$.
\end{proof}

It bears repeating that there is no converse to the above
result: there are bargmannian structures with $\nabla^g_\xi \xi = 0$
for which $d\xi^\flat \neq 0$ and bargmannian structures with
$\nabla^g_\xi \xi = f \xi$ for which $\xi^\flat \wedge d\xi^\flat \neq
0$.

If $M$ is orientable, then since $\xi^\flat$ is null, we have that
$\xi^\flat \wedge \star \xi^\flat = 0$, where $\star$ is the Hodge
star. This says that $\star \xi^\flat = \xi^\flat \wedge \mu$, for
some $\mu \in \Omega^{n-1}(M)$ which is defined up to the addition of
a term $\xi^\flat \wedge \varphi$ for some
$\varphi \in \Omega^{n-2}(M)$. In particular, $\mu$ is well defined on
the distribution $\xi^\perp$ and gives a ``volume form'' on the associated
null hypersurfaces, which is precisely the volume form of the
carrollian structure, when it exists.  Even if $M$ is not orientable,
$\mu$ exists locally.

\begin{proposition}\label{prop:bar-car-struct}
  Let $(M,g,\xi)$ be a bargmannian structure with $\xi^\perp = \ker \xi^\flat$
  involutive. Then any affine connection $\nabla$ on $M$ adapted
  to the bargmannian structure induces a connection on every leaf $N$ of $\xi^\perp$
  which is adapted to the carrollian structure on $N$ and whose
  torsion is the restriction of $T^\nabla$ to $N$.
\end{proposition}

\begin{proof}
  Since $\xi$ and $g$ are parallel, it follows that so is $\xi^\flat$:
  \begin{equation*}
    \xi^\flat(\nabla_X Y)  = X \xi^\flat(Y)
  \end{equation*}
  for all $X,Y \in \eX(M)$. In particular, if $Y \in
  \Gamma(\xi^\perp)$, so that $\xi^\flat(Y) = 0$, then $\nabla_X Y \in
  \Gamma(\xi^\perp)$ for all $X \in \eX(M)$. In other words, $\nabla$ induces
  a connection on the distribution or, equivalently, an affine
  connection on every leaf of the associated foliation.  Since $\xi$
  and $g$ are parallel, so are their restriction to the leaves of the
  foliation and hence the induced connection is adapted to the
  carrollian structure.  Finally, notice that if $X,Y \in \Gamma(\xi^\perp)$,
  then
  \begin{equation}\label{eq:torsion-D}
    T^\nabla(X,Y) = \nabla_X Y - \nabla_Y X - [X,Y] \in \Gamma(\xi^\perp),
  \end{equation}
  where we have used that $\xi^\perp$ is involutive to show that
  $T^\nabla(X,Y) \in \Gamma(\xi^\perp)$.  Finally we notice that by
  definition, $T^\nabla$ given by equation~\eqref{eq:torsion-D} is the
  torsion of the induced connection.
\end{proof}

The different classes of null hypersurfaces can be distinguished by
their second fundamental form.  Explicitly, the condition $B=0$ is
equivalent to $\eL_\xi h = 0$:
\begin{equation}\label{eq:bar-perp-killing}
  g(\nabla^g_X \xi, Y) + g(\nabla^g_Y \xi, X) = 0,\quad\forall~X,Y \perp \xi,
\end{equation}
whereas the condition $B = f h$ is equivalent to $\eL_\xi h = f h$:
\begin{equation}\label{eq:bar-perp-ckv}
  g(\nabla^g_X \xi, Y) + g(\nabla^g_Y \xi, X) = \tfrac2{n-1} g(X,Y)
  \div \xi\neq 0,\quad\forall~X,Y \perp \xi,
\end{equation}
where $\div \xi$ is the Levi-Civita divergence of $\xi$.  Finally, the
condition that $\tr B = 0$ is equivalent to $\eL_\xi \mu = 0$.

Therefore we see that the type of carrollian structure induced on the
null hypersurfaces corresponds with the type of
the distribution $\xi^\perp$.  This suggests that we rename the four types
of carrollian structures in Theorem~\ref{thm:car} as \textbf{totally geodesic}
(if $\eL_\xi h =0$), \textbf{minimal} (if $\eL_\xi \mu = 0$),
\textbf{totally umbilical} (if $\eL_\xi h = f h$) and otherwise
\textbf{generic}.

\begin{figure}[h!]
  \centering
  \begin{tikzpicture}[scale=1.3]
    %
    % nodes
    %
    \coordinate[label=right:{\scriptsize $0$}] (b0) at (-1,7); 
    \coordinate[label=left:{\scriptsize $1$}] (b1) at (-1,6); 
    \coordinate[label=left:{\scriptsize $2$}] (b2) at (-2,6); 
    \coordinate[label=right:{\scriptsize $3$}] (b3) at (-1,5); 
    \coordinate[label=above right:{\scriptsize $4$}] (b4) at (0,5); 
    \coordinate[label=left:{\scriptsize $6$}] (b6) at (0,4); 
    \coordinate[label=above right:{\scriptsize $7$}] (b7) at (1,4); 
    \coordinate[label=above left:{\scriptsize $8$}] (b8) at (-1,4); 
    \coordinate[label=above right:{\scriptsize $10$}] (b10) at (2,3); 
    \coordinate[label=above right:{\scriptsize $11$}] (b11) at (0,3); 
    \coordinate[label=above right:{\scriptsize $12$}] (b12) at (1,3); 
    \coordinate[label=above left:{\scriptsize $14$}] (b14) at (-1,3); 
    \coordinate[label=above left:{\scriptsize $16$}] (b16) at (-1,2); 
    \coordinate[label=right:{\scriptsize $17$}] (b17) at (1,2); 
    \coordinate[label=right:{\scriptsize $20$}] (b20) at (2,2); 
    \coordinate[label=below right:{\scriptsize $23$}] (b23) at (1,1); 
    %
    % edges
    %
    \tikzset{every path/.style={->,thick,gris,shorten >=3pt, shorten <=3pt}}
    \draw (b0) -- (b1);
    \draw (b0) -- (b2);
    \draw (b2) -- (b3);
    \draw (b1) -- (b3);
    \draw (b1) -- (b4);
    \draw (b3) -- (b8);
    \draw (b4) -- (b6);
    \draw (b4) -- (b7);
    \draw (b4) -- (b8);
    \draw (b8) -- (b11); 
    \draw (b8) -- (b12); 
    \draw (b8) -- (b14); 
    \draw (b6) -- (b10); 
    \draw (b6) -- (b11); 
    \draw (b7) -- (b10); 
    \draw (b7) -- (b12); 
    \draw (b14) -- (b16);
    \draw (b14) -- (b17);
    \draw (b11) -- (b16);
    \draw (b11) -- (b20);
    \draw (b12) -- (b17);
    \draw (b12) -- (b20);
    \draw (b10) -- (b20);
    \draw (b16) -- (b23);
    \draw (b17) -- (b23);
    \draw (b20) -- (b23);
    %
    % points
    %
    \foreach \point in {b0,b1,b2,b3,b4,b8,b14}
    \filldraw [color=blue!70!black,fill=blue!70!white] (\point) circle (1.5pt);
    \foreach \point in {b7,b12,b17}
    \filldraw [color=red!70!black,fill=red!70!white] (\point) circle (1.5pt);
    \foreach \point in {b6,b11,b16}
    \filldraw [color=green!70!black,fill=green!70!white] (\point) circle (1.5pt);
    \foreach \point in {b10,b20,b23}
    \filldraw [color=black,fill=white] (\point) circle (1.5pt);
    %
    % legend
    %
    \draw [black,thick] (2.8,6.8) -- (2.8,4.8) -- (5.3,4.8) -- (5.3,6.8) -- cycle;
    \filldraw [color=blue!70!black,fill=blue!70!white] (3,6.5) circle (1.5pt) node [right,black] {totally geodesic}; 
    \filldraw [color=green!70!black,fill=green!70!white] (3,6) circle (1.5pt) node [right,black] {minimal};
    \filldraw [color=red!70!black,fill=red!70!white] (3,5.5) circle (1.5pt) node [right,black] {totally umbilical}; 
    \filldraw [color=black,fill=white] (3,5) circle (1.5pt) node  [right,black] {none of the above};
    \begin{scope}[opacity=0.15]
      \coordinate[label=left:{\scriptsize $5$}] (b5) at (-2,5); 
      \coordinate[label=left:{\scriptsize $9$}] (b9) at (-2,4); 
      \coordinate[label=above left:{\scriptsize $13$}] (b13) at (-2,3); 
      \coordinate[label=left:{\scriptsize $15$}] (b15) at (-3,2); 
      \coordinate[label=left:{\scriptsize $18$}] (b18) at (-2,2); 
      \coordinate[label=below right:{\scriptsize $19$}] (b19) at (0,2); 
      \coordinate[label=below left:{\scriptsize $21$}] (b21) at (-2,1); 
      \coordinate[label=below left:{\scriptsize $22$}] (b22) at (-1,1); 
      \coordinate[label=below right:{\scriptsize $24$}] (b24) at (0,1); 
      \coordinate[label=below right:{\scriptsize $25$}] (b25) at (0,0); 
      \coordinate[label=right:{\scriptsize $26$}] (b26) at (0,-1);
      % 
      % edges
      % 
      \tikzset{every path/.style={->,thick,black,shorten >=3pt, shorten <=3pt}}
      \draw (b2) -- (b5);
      \draw (b3) -- (b9);
      \draw (b5) -- (b9);
      \draw (b9) -- (b13); 
      \draw (b8) -- (b13); 
      \draw (b13) -- (b15);
      \draw (b13) -- (b18);
      \draw (b13) -- (b19);
      \draw (b14) -- (b15);
      \draw (b11) -- (b18);
      \draw (b12) -- (b19);
      \draw (b15) -- (b21);
      \draw (b15) -- (b22);
      \draw (b18) -- (b21);
      \draw (b18) -- (b24);
      \draw (b16) -- (b21);
      \draw (b19) -- (b22);
      \draw (b19) -- (b24);
      \draw (b17) -- (b22);
      \draw (b20) -- (b24);
      \draw (b21) -- (b25); 
      \draw (b22) -- (b25); 
      \draw (b23) -- (b25); 
      \draw (b24) -- (b25); 
      \draw (b25) -- (b26); 
      % 
      % points
      % 
      \foreach \point in {b5,b9,b13,b15,b19,b22,b18,b21,b24,b25,b26}
      \filldraw [color=white!70!black,fill=black!70!white] (\point) circle (1.5pt);
    \end{scope}
  \end{tikzpicture}
  \caption{Hasse diagram of bargmannian structures with involutive $\xi^\perp$}
  \label{fig:bar-hasse-involutive}
\end{figure}
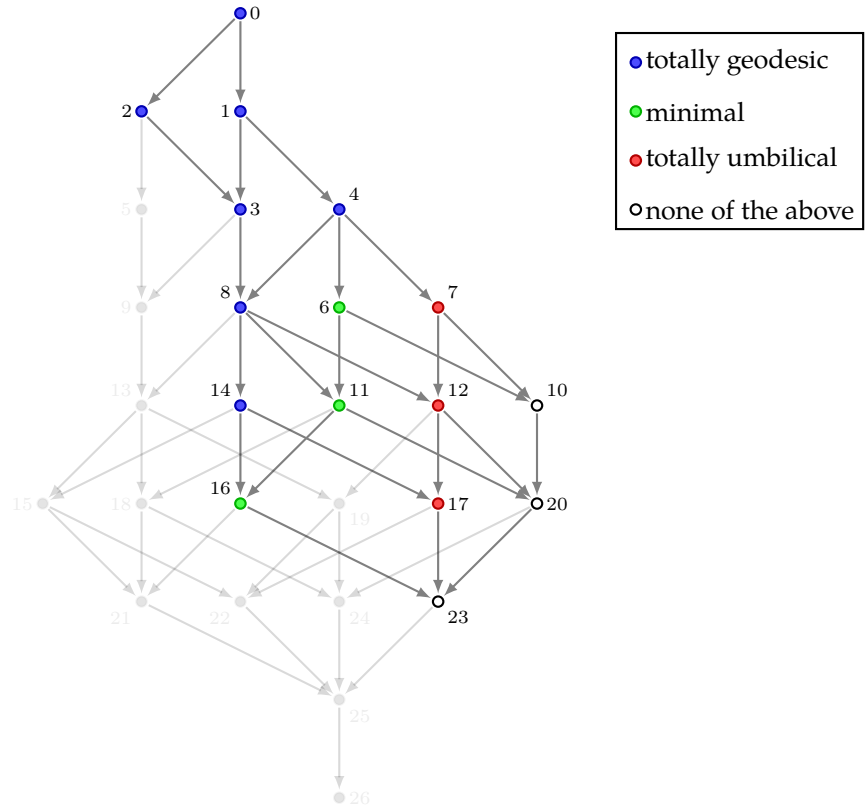

\section{Conclusions}
\label{sec:conclusions}

In this paper we have studied spacetime structures from the point of
view of $G$-structures and studied their intrinsic torsion.  Whereas
this provides no information for the case of lorentzian spacetimes,
the situation for non-lorentzian spacetimes is very different.  As
Theorem~\ref{thm:gal} shows, the classification of galilean structures
by intrinsic torsion coincides with the classification of
Newton--Cartan geometries in \cite{Christensen:2013lma} into what
those authors call torsionless, twistless torsional and torsional
Newton--Cartan geometries.  As Theorem~\ref{thm:car} shows there are
4 types of carrollian structures, which as discussed in
Section~\ref{sec:bar-struc-car}, may be distinguished by the
geometrical properties of the null hypersurfaces of bargmannian manifolds
into which they embed: totally geodesic, totally umbilical, minimal
or generic.  The intersection of the galilean and
carrollian structures consists of the aristotelian structures, and as
Theorem~\ref{thm:ari} shows, there are 16 classes depending on
their intrinsic torsion.  As advocated in \cite{Duval:2014uoa},
bargmannian structures are a subclass of lorentzian structures which are
intimately linked with both galilean and carrollian structures.  The
study of the intrinsic torsion of bargmannian structures is surprisingly
rich and as Theorem~\ref{thm:bar} shows there are 27 bargmannian
structures, many of which can be related to galilean and carrollian
structures in a way made explicit in
Section~\ref{sec:corr-betw-bargm}.  We find that all three classes of
galilean structures can arise as null reductions of bargmannian
structures, whereas all four classes of carrollian structures can be
induced from suitable bargmannian structures by restriction to null
hypersurfaces integrating the distribution $\xi^\perp$.

The above results hold in generic dimension, which means that $n \neq
2,5$.  As shown in Appendix~\ref{sec:some-spec-dimens}, there are only
2 galilean and carrollian $G$-structures in two dimensions, and hence
4 aristotelian structures, whereas there are 11 three-dimensional
bargmannian structures.  Similarly, there are 5 five-dimensional
galilean structures, 32 five-dimensional aristotelian structures and
47 six-dimensional bargmannian structures.

It remains to understand whether all the different classes of
(five-dimensional) galilean, aristotelian and bargmannian structures
can be realised geometrically or whether, as is the case with
$G_2$-structures on $7$-manifolds  \cite{FernandezGray,CMS}, some of
the inclusions between the different classes (e.g., those in
Figure~\ref{fig:bar-hasse-submodules} for bargmannian structures) are
not strict.

The classification of $G$-structures via intrinsic torsion is still
somewhat coarse -- after all, the intrinsic torsion is the first of a
sequence of obstructions to the integrability of the $G$-structure --
but the results in this paper may help to add some structure to the
zoo of non-lorentzian geometries.

\section*{Acknowledgments}

I would like to acknowledge fruitful conversations and correspondence
on these and related topics with Jelle Hartong, James Lucietti, Stefan
Prohazka and Andrea Santi. I am particularly grateful to Jelle Hartong
for organising an online
\href{https://indico.nbi.ku.dk/event/1374/page/564-lectures-on-non-lorentzian-g-structures-july-141517}{lecture
  series} on this topic which I delivered in the framework of the
\href{https://indico.nbi.ku.dk/event/1374/}{NL Zoom meetings 2020}.  I
would like to thank the participants of these meetings for their
interest and the many questions, which I hope have improved the
presentation in this paper.  My own interest in this topic was
re-awakened thanks to correspondence with Dieter Van den Bleeken and I
would also like to record my gratitude to him, especially for the
careful reading of a previous version of this paper.  Last, but by no
means least, I would like to dedicate this paper to Dmitri Alekseevsky
on his eightieth birthday, in hopes that he might derive some pleasure
in seeing some familiar structures in a possibly novel context.

\appendix

\begin{appendices}

\section{Hypersurface orthogonality}
\label{sec:hypers-orth}

It is of course a well-known fact that if a nowhere-vanishing one-form
$\tau \in \Omega^1(M)$ satisfies $d\tau \wedge \tau = 0$ then there
exists a one-form $\omega \in \Omega^1(M)$ such that
$d\tau = \tau \wedge \omega$. The statement is ubiquitous in the
literature, but the proof is not. In this appendix I record a proof of
this fact. Of course the condition simply says that the characteristic
distribution $\ker\tau \subset TM$ is Frobenius integrable and hence
$M$ is foliated by hypersurfaces whose tangent spaces agree with
$\ker \tau$. By abuse of language one says that $\tau$ is hypersurface
orthogonal, a concept taken from riemannian geometry where the vector
field dual to $\tau$ would indeed be orthogonal to the hypersurfaces
integrating $\ker \tau$.

\begin{proposition}\label{prop:hso}
  Let $\tau \in \Omega^1(M)$ be nowhere vanishing.  Then the following
  are equivalent
  \begin{enumerate}
  \item $d\tau \wedge \tau = 0$
  \item $d\tau = \tau \wedge \omega$, for some $\omega \in \Omega^1(M)$.
  \end{enumerate}
\end{proposition}

\begin{proof}
  It is clear that (2) implies (1), so we need to prove that (1)
  implies (2).  The idea is to show this locally and then to show that
  the local $\omega$'s glue to a global one-form.

  Since $\tau$ is nowhere-vanishing, we may complete to a local
  coframe $(\theta^1=\tau, \theta^2,\dots,\theta^n)$ defined on some
  chart $(U,\varphi)$ for $M$.  Then
  \begin{equation*}
    d\tau = \sum_{i<j} f_{ij} \theta^i \wedge \theta^j
  \end{equation*}
  for some $f_{ij} \in \F(U)$.  Then
  \begin{equation*}
    \tau \wedge d\tau = \sum_{i<j} f_{ij} \theta^1 \wedge \theta^i
    \wedge \theta^j = \sum_{1<i<j} f_{ij} \theta^1 \wedge \theta^i
    \wedge \theta^j,
  \end{equation*}
  so that $\tau \wedge d\tau = 0$ says that $f_{ij} = 0$ for
  $1<i<j$, and hence
  \begin{equation*}
    d\tau = \sum_{1<j} f_{1j} \theta^1 \wedge \theta^j = \theta^1
    \wedge \sum_{1<j} f_{1j} \theta^j = \tau \wedge \omega,
  \end{equation*}
  for $\omega = \sum_{1<j} f_{1j} \theta^j \in \Omega^1(U)$.  Notice
  that $\omega$ is not unique, since we could always add a component
  along $\tau$.  We will exploit this ambiguity when we glue the
  local $\omega$s.

  Let $\{(U_\alpha,\varphi_\alpha)\}_{\alpha\in A}$ be an atlas for
  $M$.  Then we have just shown that there exists
  $\overline\omega_\alpha \in \Omega^1(U_\alpha)$, where
  $d\tau = \tau \wedge \overline\omega_\alpha$ on $U_\alpha$.
  Since $\tau$ and $d\tau$ are global forms, on a non-empty
  overlap $U_{\alpha\beta}$,
  \begin{equation*}
    \tau \wedge (\overline\omega_\alpha - \overline\omega_\beta)= 0.
  \end{equation*}
  We claim that $\overline\omega_\alpha - \overline\omega_\beta =
  f_{\alpha\beta} \tau$ for some $f_{\alpha\beta} \in
  \F(U_{\alpha\beta})$.  To see this, write
  \begin{equation*}
    \overline\omega_\alpha - \overline\omega_\beta = \sum_{i=1}^n g_i \theta^i,
  \end{equation*}
  for some $g_i \in \F(U_{\alpha\beta})$, so that
  \begin{equation*}
    \tau \wedge (\overline\omega_\alpha - \overline\omega_\beta) =
    \tau \wedge \sum_i g_i \theta^i = \sum_i g_i \theta^1 \wedge
    \theta^i = \sum_{i>1} g_i \theta^1 \wedge \theta^i.
  \end{equation*}
  If $\tau \wedge (\overline\omega_\alpha - \overline\omega_\beta)=
  0$, we see that $g_i=0$ for $i>1$ and hence
  \begin{equation*}
    \overline\omega_\alpha - \overline\omega_\beta= f_{\alpha\beta} \tau,
  \end{equation*}
  where $f_{\alpha\beta} = g_1$.

  On a triple overlap $U_{\alpha\beta\gamma}$, we have that
  \begin{equation*}
    (f_{\alpha\beta}  + f_{\beta\gamma} + f_{\gamma\alpha}) \tau =
    \overline\omega_\alpha - \overline\omega_\beta +
    \overline\omega_\beta - \overline\omega_\gamma +
    \overline\omega_\gamma - \overline\omega_\alpha  = 0,
  \end{equation*}
  and since $\tau$ is nowhere-vanishing,
  \begin{equation}\label{eq:cocycle}
    f_{\alpha\beta}  + f_{\beta\gamma} + f_{\gamma\alpha} = 0.
  \end{equation}
  Let $\{\rho_\alpha\}$ denote a partition of unity subordinate to the
  atlas, with $\rho_\alpha$ supported in $U_\alpha$.  Define $g_\beta
  = \sum_\alpha \rho_\alpha f_{\alpha\beta} \in \F(M)$.  Then
  \begin{align*}
    g_\alpha - g_\beta &= \sum_\gamma (\rho_\gamma f_{\gamma\alpha} - \rho_\gamma f_{\gamma\beta})\\
                       &= \sum_\gamma (\rho_\gamma f_{\gamma\alpha} + \rho_\gamma f_{\beta\gamma}) && \tag{since $f_{\beta\gamma} = - f_{\gamma\beta}$}\\
                       &= - \sum_\gamma \rho_\gamma f_{\alpha\beta} && \tag{by \eqref{eq:cocycle}}\\
                       &= - f_{\alpha\beta}. && \tag{since $\sum_\gamma \rho_\gamma = 1$}
  \end{align*}
  Therefore $\overline\omega_\alpha - \overline\omega_\beta = (g_\beta -
  g_\alpha) \tau$, so that on $U_{\alpha\beta}$,
  \begin{equation*}
    \overline\omega_\alpha + g_\alpha \tau = \overline\omega_\beta + g_\beta \tau.
  \end{equation*}
  Let $\omega_\alpha = \overline\omega_\alpha + g_\alpha \tau\in
  \Omega^1(U_\alpha)$.  Then $\omega_\alpha = \omega_\beta$ on
  $U_{\alpha\beta}$ and hence it glues to a global form $\omega \in
  \Omega^1(M)$.  Notice that on $U_\alpha$,
  \begin{equation*}
    \tau \wedge \omega = \tau \wedge (\overline\omega_\alpha +
    g_\alpha \tau) = \tau \wedge \overline\omega_\alpha = d\tau,
  \end{equation*}
  as desired.
\end{proof}

\section{Some special dimensions}
\label{sec:some-spec-dimens}

In the bulk of the paper we have taken the dimension $n$ to be
generic; but for some special values of $n$ (i.e., $n=2$ and $n=5$)
the discussion needs to be refined.  We will briefly comment on how
the results in the bulk of paper are modified for such values of $n$.

\subsection{Two-dimensional galilean, carrollian and aristotelian structures}
\label{sec:two-dimens-galil}

If we think of galilean and carrollian structures as arising from non-
and ultra-relativistic limits of lorentzian geometry, it is visually
clear that in two dimensions the limits are equivalent simply by
re-interpreting what we call time and space, which are geometrically
equivalent in this dimension.  This would seem to contradict the
results of Sections~\ref{sec:gal-g-struct} and \ref{sec:car-g-struct},
which therefore require modification.

When it comes to galilean structures, $\coker\d$ is now
one-dimensional and the intrinsic torsion is still determined by
$d\tau$.  The main difference is that now $\tau \wedge d\tau = 0$ by
dimension, so we only have two (and not three) galilean structures,
depending on whether or not $d\tau$ vanishes.

Similarly, in the case of carrollian structures $\coker \d$ is again
one-dimensional and the intrinsic torsion is still determined by
$\eL_\xi h$, except that since $h$ is rank-one, there are no non-zero
traceless symmetric tensors.  Hence here too we have only two (and not
four) carrollian structures, depending on whether or not $\eL_\xi h$
vanishes.

Therefore the seeming discrepancy between galilean and carrollian
structures is not there in two dimensions after all.

The classes of aristotelian structures for $n=2$ also simplifies as a
result.  Now the structure group is $O(1) \cong \ZZ_2$ and the
submodules $\eA_1$ and $\eA_3$ are absent.  The submodules $\eA_2$ and
$\eA_4$ are one-dimensional: $G$ acts trivially on $\eA_4$ and via the
``determinant'' on $\eA_2$.  All said, there are four aristotelian
structures for $n=2$, depending on whether either of $d\tau$ and
$\eL_\xi h$ vanishes or not.

\subsection{Three-dimensional bargmannian structures}
\label{sec:three-dimens-bargm}

The classification of bargmannian structures also changes when $n=2$.
Now two submodules are absent: $(P\otimes \pi)_{\wedge^2} \cong
H\otimes \pi \wedge \pi$ and $(P \otimes \pi)_{\odot_0^2} \cong (P
\otimes \pi)_{\odot_0^2} \wedge \zeta$.  This results in a somewhat
simplified version of Figure~\ref{fig:hasse-bar-subspaces}, which we
omit.  There are some coincidences between the twenty-six bargmannian
structures: $\eB_2 = \eB_5$, $\eB_3 = \eB_9$, $\eB_4 = \eB_6$,
$\eB_7=\eB_{10}$, $\eB_8 = \eB_{11} = \eB_{13} = \eB_{18}$, $\eB_{14}
= \eB_{15} = \eB_{16} = \eB_{21}$, $\eB_{12} = \eB_{19} = \eB_{20}$,
$\eB_{17} = \eB_{22}= \eB_{23}= \eB_{24} = \eB_{25}$.  The minimal
structures coincide with the totally geodesic structures and the
``none of the above'' structures (except for the generic structure
$\eB_{26}$) are now totally umbilical.  In summary, there are eleven
three-dimensional bargmannian structures, whose Hasse diagram is depicted 
in Figure~\ref{fig:bar-hasse-3d}.

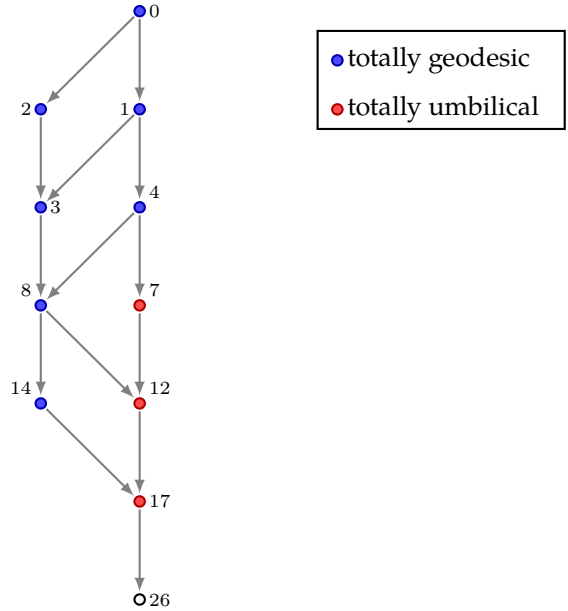
\begin{figure}[h!]
  \centering
  \begin{tikzpicture}[scale=1.3]
    %
    % nodes
    %
    \coordinate[label=right:{\scriptsize $0$}] (b0) at (-1,7); 
    \coordinate[label=left:{\scriptsize $1$}] (b1) at (-1,6); 
    \coordinate[label=left:{\scriptsize $2$}] (b2) at (-2,6); 
    \coordinate[label=right:{\scriptsize $3$}] (b3) at (-2,5); 
    \coordinate[label=above right:{\scriptsize $4$}] (b4) at (-1,5); 
    \coordinate[label=above right:{\scriptsize $7$}] (b7) at (-1,4); 
    \coordinate[label=above left:{\scriptsize $8$}] (b8) at (-2,4); 
    \coordinate[label=above right:{\scriptsize $12$}] (b12) at (-1,3); 
    \coordinate[label=above left:{\scriptsize $14$}] (b14) at (-2,3); 
    \coordinate[label=right:{\scriptsize $17$}] (b17) at (-1,2); 
    \coordinate[label=right:{\scriptsize $26$}] (b26) at (-1,1);
    %
    % edges
    %
    \tikzset{every path/.style={->,thick,gris,shorten >=3pt, shorten <=3pt}}
    \draw (b0) -- (b1);
    \draw (b0) -- (b2);
    \draw (b2) -- (b3);
    \draw (b1) -- (b3);
    \draw (b1) -- (b4);
    \draw (b3) -- (b8);
    \draw (b4) -- (b7);
    \draw (b4) -- (b8);
    \draw (b7) -- (b12);     
    \draw (b8) -- (b12); 
    \draw (b8) -- (b14); 
    \draw (b14) -- (b17);
    \draw (b12) -- (b17);
    \draw (b17) -- (b26);
    %
    % points
    %
    \foreach \point in {b0,b1,b2,b3,b4,b8,b14}
    \filldraw [color=blue!70!black,fill=blue!70!white] (\point) circle (1.5pt);
    \foreach \point in {b7,b12,b17}
    \filldraw [color=red!70!black,fill=red!70!white] (\point) circle (1.5pt);
    \foreach \point in {b26}
    \filldraw [color=black,fill=white] (\point) circle (1.5pt);
    %
    % legend
    %
    \draw [black,thick] (0.8,6.8) -- (0.8,5.8) -- (3.3,5.8) -- (3.3,6.8) -- cycle;
    \filldraw [color=blue!70!black,fill=blue!70!white] (1,6.5) circle (1.5pt) node [right,black] {totally geodesic}; 
    \filldraw [color=red!70!black,fill=red!70!white] (1,6) circle (1.5pt) node [right,black] {totally umbilical}; 
   \end{tikzpicture}
  \caption{Hasse diagram of three-dimensional bargmannian structures}
  \label{fig:bar-hasse-3d}
\end{figure}

\subsection{Five-dimensional galilean structures}
\label{sec:five-dimens-galil}

When $n=5$, the $\so(4)$-submodule $H\otimes \pi \wedge \pi$ in
$\coker\d$ described in Section~\ref{sec:gal-intr-tors} is not
irreducible, breaking up into selfdual and antiselfdual summands.
This means that if the galilean structure reduces further to a
$G_0$-structure, with $G_0 \cong \SO(4) \ltimes \RR^4$, we have five
$G_0$-submodules of $\coker \d$ and hence five galilean structures
instead of three.  The torsional Newton--Cartan geometries, where
$d\tau \wedge \tau \neq 0$, now come in three flavours: selfdual,
antiselfdual and neither, according to whether the restriction of
$d\tau$ to the four-dimensional oriented sub-bundle $\ker \tau$ is
selfdual, antiselfdual or neither.

\subsection{Five-dimensional aristotelian structures}
\label{sec:five-dimens-arist}

When $n=5$, the $\so(4)$-submodule $\eA_1 \cong \wedge^2 W$ defined in
Section~\ref{sec:ari-intr-tors}, with $W$ the four-dimensional vector
representation of $\so(4)$, is no longer irreducible.  Indeed, it
decomposes into selfdual and antiselfdual summands:
\begin{equation}
  \eA_1 = \eA_1^+ \oplus \eA_1^- = \wedge^2_+ W \oplus \wedge^2_- W.
\end{equation}
If the aristotelian structure reduces further to $G_0=\SO(4)$, then
$\eA_1^\pm$ are $G_0$-submodules and we must refine the classification
of aristotelian structures.  Theorem~\ref{thm:ari} gets modified:
there are not sixteen, but thirty-two aristotelian structures.  Each
of the eight structures in Theorem~\ref{thm:ari} whose intrinsic
torsion have a nonzero component in $\eA_1$ -- namely, those for which
$\tau \wedge d\tau \neq 0$ -- now can be of three distinct types,
depending on whether $d\tau$ is selfdual, antiselfdual or neither when
restricted to the four-dimensional distribution $\ker\tau$.

\subsection{Six-dimensional bargmannian structures}
\label{sec:six-dimens-bargm}

If $n=5$, and if the group of the bargmannian structure reduces to the
identity component $G_0 \cong \SO(4) \ltimes \RR^4$, then
$\SO(4)$-submodule $(P\otimes \pi)_{\wedge^2}$ is no longer
irreducible and decomposes into selfdual and antiselfdual parts.
Under the action of the boosts, it is still the case that the
$\SO(4)$-submodule $P\otimes \zeta$ maps into $(P\otimes
\pi)_{\wedge^2}$ and each of $(P\otimes \pi)_{\wedge_\pm^2}$ maps into
$\Xi_+$.  Therefore all that happens is that every bargmannian structure
$\eB_i$ (except for $\eB_{26}$) which contains $(P\otimes
\pi)_{\wedge^2}$ now comes in two more flavours: $\eB_i^+$ and
$\eB_i^-$, where $(P\otimes \pi)_{\wedge^2}$ is replaced by the
submodules $(P\otimes \pi)_{\wedge_+^2}$ or $(P\otimes
\pi)_{\wedge_-^2}$, respectively.  So now we have twenty additional
structures: $\eB_5^\pm$, $\eB_9^\pm$, $\eB_{13}^\pm$, $\eB_{15}^\pm$,
$\eB_{18}^\pm$, $\eB_{19}^\pm$, $\eB_{21}^\pm$, $\eB_{22}^\pm$,
$\eB_{24}^\pm$  and $\eB_{25}^\pm$.  I omit the rather more involved
Hasse diagram of the 47 six-dimensional bargmannian structures, as I do
their geometric characterisation.

\end{appendices}

\providecommand{\href}[2]{#2}\begingroup\raggedright\endgroup

\end{document}